\documentclass[11pt]{article}

\usepackage[T1]{fontenc}
\usepackage[utf8]{inputenc}

\usepackage{fullpage}
\usepackage{float}
\usepackage{pdflscape}
\usepackage{comment} 
\usepackage{boxedminipage}
\usepackage{enumerate}
\usepackage{enumitem}
\usepackage{outlines}

\usepackage{graphicx}
\usepackage[font=small]{caption} 
\usepackage{subcaption}
\usepackage{tikz}
\usetikzlibrary{
    arrows.meta, 
    decorations.pathreplacing, 
    decorations.pathmorphing, 
    patterns, 
    positioning
}

\usepackage{xcolor}
\usepackage[many]{tcolorbox}

\usepackage{amsmath, amssymb, amsthm}
\usepackage{mathtools}
\usepackage{mathrsfs}
\usepackage{bbm}
\usepackage{braket}
\usepackage{thmtools, thm-restate}

\usepackage{algorithm}

\usepackage{upref}
\usepackage{varioref}
\usepackage[bookmarks, backref=page]{hyperref}
\usepackage[capitalize, noabbrev, nameinlink]{cleveref}
\definecolor{darkblue}{rgb}{0,0.1,0.55}
\hypersetup{ 
	colorlinks,
	linkcolor={darkblue}, 
	citecolor={darkblue}, 
	urlcolor={darkblue}, 
} 

\AtBeginDocument{} 

\newtheoremstyle{plain2} 
  {15pt}
  {15pt}
  {\itshape}%
  {}%
  {\bfseries}%
  {.}%
  {.5em}%
  {}%
\theoremstyle{plain2}
\newtheorem{theorem}{Theorem}[section]
\newtheorem{lemma}[theorem]{Lemma}
\newtheorem{corollary}[theorem]{Corollary} 
\newtheorem{heuristic}{Heuristic}{\bfseries}{\itshape}
\newtheorem{hclaim}{Heuristic Claim}{\bfseries}{\itshape}

\newtheoremstyle{definition2}
  {12pt}
  {12pt}
  {\upshape}%
  {}%
  {\bfseries}%
  {.}%
  {.5em}%
  {}%
\theoremstyle{definition2}
\newtheorem{definition}[theorem]{Definition}
\newtheorem{problem}{Problem}

\newtheoremstyle{remark2}
  {12pt}
  {12pt}
  {\upshape}%
  {}%
  {\itshape}%
  {.}%
  {.5em}%
  {}%
\theoremstyle{remark2}
\newtheorem{remark}{Remark}

\newcommand{\E}{\mathbb{E}} 	
\newcommand{\N}{\mathbb{N}}
\newcommand{\R}{\mathbb{R}}
\newcommand{\Z}{\mathbb{Z}}

\newcommand{\cA}{{\mathcal A}}

\newcommand{\cC}{{\mathcal C}}

\newcommand{\cH}{{\mathcal H}}

\newcommand{\cO}{{\mathcal O}}

\newcommand{\cS}{{\mathcal S}}
\newcommand{\cT}{{\mathcal T}}
\newcommand{\cU}{{\mathcal U}}

\newcommand{\cW}{{\mathcal W}}

\newcommand{\vb}{\mathbf{b}}
\newcommand{\vc}{\mathbf{c}}
\newcommand{\vs}{\mathbf{s}}

\newcommand{\vv}{\mathbf{v}}

\newcommand{\vx}{\mathbf{x}}
\newcommand{\vy}{\mathbf{y}}
\newcommand{\vz}{\mathbf{z}}

\newcommand{\mQ}{\mathbf{Q}}

\DeclarePairedDelimiter{\norm}{\lVert}{\rVert}

\newcommand{\innerP}[2]{\langle #1, #2 \rangle}

\newcommand{\poly}{\mathrm{poly}}

\newcommand{\epsapprox}{\epsilon}
\newcommand{\nrep}{\ell}
\newcommand{\searchthreshold}{t}
\newcommand*{\transpose}{{\mathsf{T}\!}}

\title{An Improved Quantum Algorithm for 3-Tuple Lattice Sieving}
\author{Lynn Engelberts\thanks{QuSoft and CWI, the Netherlands. Supported by the Dutch National Growth Fund (NGF), as part of the Quantum Delta NL program. {\tt lynn.engelberts@cwi.nl}} \and Yanlin Chen\thanks{QuICS and University of Maryland, United States. Most of this work was completed while the author was at QuSoft and CWI. {\tt yanlin@umd.edu}} \and Amin Shiraz Gilani\thanks{QuICS and University of Maryland, United States. Partially supported by the DOE ASCR Quantum Testbed Pathfinder program (awards DE-SC0019040 and DE-SC0024220). {\tt asgilani@umd.edu}} \and
Maya-Iggy van Hoof\thanks{Horst G\"ortz Institute for IT-Security, Ruhr University Bochum, Bochum, Germany. Partially supported by the Deutsche Forschungsgemeinschaft (DFG, German Research Foundation) under Germany's Excellence Strategy-EXC 2092 CASA-390781972 ``Cyber Security in the Age of Large-Scale Adversaries''. {\tt iggy.hoof@ruhr-uni-bochum.de}}
\and Stacey Jeffery\thanks{QuSoft, CWI and University of Amsterdam, the Netherlands. Partially supported by the European Union (ERC, ASC-Q, 101040624). {\tt jeffery@cwi.nl}} \and
Ronald de Wolf\thanks{QuSoft, CWI and University of Amsterdam, the Netherlands. Partially supported by the Dutch Research Council (NWO) through Gravitation-grant Quantum Software Consortium, 024.003.037. {\tt rdewolf@cwi.nl}}} 

\begin{document}
\maketitle
\allowdisplaybreaks 

\begin{abstract}
The assumed hardness of the Shortest Vector Problem (SVP) in high-dimensional lattices is one of the cornerstones of post-quantum cryptography. 
The fastest known heuristic attacks on SVP are via so-called sieving methods. While these still take exponential time in the dimension~$d$, they are significantly faster than non-heuristic approaches and their heuristic assumptions are verified by extensive experiments. 
$k$-Tuple sieving is an iterative method where each iteration takes as input a large number of lattice vectors of a certain norm, and produces an equal number of lattice vectors of slightly smaller norm, by taking sums and differences of $k$ of the input vectors. Iterating these ``sieving steps'' sufficiently many times produces a short lattice vector. The fastest attacks (both classical and quantum) are for $k=2$, but taking larger $k$ reduces the amount of memory required for the attack. In this paper we improve the quantum time complexity of 3-tuple sieving from $2^{0.3098 d}$ to $2^{0.2846 d}$, using a two-level amplitude amplification aided by a preprocessing step that associates the given lattice vectors with nearby ``center points'' to focus the search on the neighborhoods of these center points. Our algorithm uses $2^{0.1887d}$ classical bits and QCRAM bits, and $2^{o(d)}$ qubits. This is the fastest known quantum algorithm for SVP when total memory is limited to $2^{0.1887d}$.
\end{abstract}

\thispagestyle{empty} 
\newpage
\setcounter{page}{1}

\section{Introduction}

\subsection{Classical and quantum sieving approaches to the Shortest Vector Problem}\label{sec:intro_sieving}

The Shortest Vector Problem (SVP) is the following: given linearly independent vectors $\vb_1,\dots,\vb_d\in\R^d$, find a shortest nonzero vector in the lattice obtained by taking integer linear combinations of these vectors, that is, in
\[
\Lambda=\left\{\sum_{i=1}^d z_i \vb_i \colon z_1,\dots,z_d\in\Z\right\}.
\]  
The geometry and combinatorics of SVP are interesting mathematical problems in their own right, but SVP also underlies the most promising alternatives to the old favorites of integer factorization and discrete logarithm (which are both broken by quantum computers~\cite{Shor97}) as a basis for public-key cryptography; see~\cite{MR08} for a survey on lattice-based cryptography. 
Indeed, much of our future cryptography is premised on the assumption that there is no efficient algorithm for SVP, in fact not even for \emph{approximate} SVP, which is the task of finding a nonzero vector of length at most some small factor $\gamma$ (say some polynomial in $d$) larger than the shortest length. 

The fastest provable classical algorithm for SVP has runtime essentially $2^d$~\cite{ADRS15} on worst-case instances. This was improved to runtime roughly $2^{0.95d}$ on a quantum computer, and even to $2^{0.835d}$ on a quantum computer with a large QCRAM~\cite{ACKS25QSVP}.  For meaningfully breaking cryptosystems, worst-case algorithms are more than needed and it suffices to have a fast \emph{heuristic} algorithm, one that works for most (practical) instances under some plausible, though not quite rigorous, assumptions. 
After all, no one would use a cryptographic system that is known to be breakable with a significant probability. 
The fastest heuristic methods we know for solving approximate SVP still take exponential time $2^{cd}$, but with a constant $c<1$ that is much smaller than for the best known worst-case algorithms. These heuristic methods are based on sieving ideas, which were first introduced in the context of SVP by Ajtai, Kumar, and Sivakumar~\cite{AKS01} and made more practical by Nguyen and Vidick~\cite{NV08}. (In fact, the best non-heuristic, worst-case algorithms such as \cite{ADRS15} can also be viewed as ``sieving'' algorithms.) 

The most basic type of sieving is \emph{2-tuple} sieving. The idea here is to begin by sampling a large list~$L$ of $m$ random vectors from the lattice $\Lambda$, all of roughly the same norm $R$, with $R$ chosen large enough to allow efficient sampling (in fact we may assume $R=1$ by scaling the lattice at the start). We then try to find slightly shorter vectors that are sums or differences of pairs of vectors $\vx,\vy \in L$ (note that $\vx+\vy$ and $\vx-\vy$ are still in $\Lambda$ because a lattice is closed under taking integer linear combinations). 
The use of heuristics in the analysis of this approach can be explained as follows. Given two uniformly random $\vx,\vy \in \R^d$ of the same norm $R$, the probability that $\vx+\vy$ or $\vx-\vy$ has norm $<R$ can be seen to be $p\approx 2^{-0.2075d}$ from basic estimates of the area of a cap on a $d$-dimensional sphere. With $m$ initial vectors, we have $\binom{m}{2}$ 2-tuples of vectors, so if these vectors behave like i.i.d.\ uniformly random unit vectors then the expected number of 2-tuples inducing a shorter vector is $p\binom{m}{2}$. Hence, choosing $m\approx 2/p\approx 2^{0.2075d}$ ensures there will be roughly $m$ 2-tuples that each give a shorter vector. If we can efficiently \emph{find} those $m$ good 2-tuples, then we have generated a new list~$L'$ of $m$ shorter lattice vectors. The heuristic assumption is that, once normalized, the vectors in~$L'$ again behave like i.i.d.\ uniformly random unit vectors (this is a common assumption in lattice-sieving algorithms, and has been extensively verified numerically~\cite{NV08,BLS2016,ADHKPS19}). 
This list $L'$ will form the starting point of the next sieving iteration. Under the heuristic assumption, we can then form a new list~$L''$ of roughly $m$ lattice vectors of even shorter norm by taking sums or differences of 2-tuples of vectors from~$L'$, and so on. 

Usually a relatively small (polynomial in $d$) number of such iterations suffices to find quite short vectors (see~\Cref{sec:explanation-sieving}), so the cost of the overall procedure for solving approximate SVP will be dominated by the cost of one sieving iteration, which is the cost of finding $m$ good 2-tuples among the $\binom{m}{2}$ 2-tuples. How much time does this take?  Nguyen and Vidick~\cite{NV08} used brute-force search over all $\binom{m}{2}$ 2-tuples, which takes time $O(m^2)=O(2^{0.416d})$. This time has since been improved using nearest-neighbor data structures that allow us to quickly find, for a given $\vx\in L$, a small number of $\vy\in L$ that are somewhat close to $\vx$,  and in particular using locality-sensitive filtering techniques that allow us to focus the search for a close $\vy$ to the neighborhoods of shared ``center points''. 
These algorithms have been further improved by using various \emph{quantum} subroutines to speed up the search for good 2-tuples, in particular Grover's search algorithm~\cite{LMP15,Laa16}, quantum walks~\cite{CL21}, and reusable quantum walks~\cite{BCSS22}. Currently, the best classical and quantum runtimes are $2^{0.2925d+o(d)}$~\cite{BDGL16} and $2^{0.2563d+o(d)}$~\cite{BCSS22}, respectively, which are the fastest known  heuristic attacks on SVP. 
These attacks, however, do require at least $2^{0.2075d}$ bits of memory. 

The 2-tuple-sieving approach can be generalized to $k$-tuple sieving for some larger constant $k\geq 3$ in the natural way~\cite{BLS2016}: by looking for $k$-tuples $\vx_1,\dots,\vx_k$ from our current list~$L$ such that $\norm{\vx_1\pm \vx_2\pm\cdots\pm \vx_k}\leq 1$ for some choice of the coefficients $\pm 1$ (note that the number of sign patterns for the coefficients is just $2^{k-1}$, which disappears in the big-$O$ because $k$ is a constant). 
The advantage of going to $k>2$ is that the initial list size $m$ can be smaller while still giving roughly $m$ good $k$-tuples (and hence roughly $m$ shorter lattice vectors for the next sieving iteration), meaning the algorithm requires less memory. 
The disadvantage, on the other hand, is that the time to find $m$ good $k$-tuples increases with $k$, because the set of $k$-tuples has $\binom{m}{k}$ elements, which grows with $k$, even when we take into account that the minimal required list size $m$ itself goes down with~$k$. For example, for 2-tuple sieving and 3-tuple sieving, the required list sizes are $m_2\approx 2^{0.2075 d}$ and  $m_3\approx 2^{0.1887 d}$, respectively, yet $\binom{m_3}{3}\gg \binom{m_2}{2}$ despite the fact that $m_3\ll m_2$.

As in 2-tuple sieving, a relatively small (polynomial in $d$) number of iterations suffices for finding short lattice vectors. Hence, the overall cost of $k$-tuple sieving will be dominated by the cost of finding $m$ good $k$-tuples among the set of all $\binom{m}{k}$ $k$-tuples from the given list of $m$ vectors. In~\Cref{table:k234complexities}, we give the best known classical and quantum upper bounds on the time complexity of $k$-tuple sieving for $k=2,3,4$, as established in previous work. These algorithms use time and memory that is exponential in~$d$, so  the table just gives the constant in the exponent of the time complexity, suppressing $o(1)$ terms. 


\bigskip

\begin{table}[hbt]
\centering
\begin{tabular}{|l|l|l|l|}
\hline 
$k$ & Memory & Classical time & Quantum time\\ \hline
2 & 0.2075 & 0.2925~\cite{BDGL16}& 0.2563~\cite{BCSS22}\\
3 & 0.1887 & 0.3383~\cite{CL23} & 0.3098~\cite{CL23}, \textbf{0.2846 (this work)} \\
4 & 0.1724 & 0.3766~\cite{HKL2017}  & 0.3178~\cite[App.~B.2]{KMPR2019}\\ \hline
\end{tabular}
\caption{Exponents of the best known classical and quantum time complexities for $k$-tuple sieving with minimal memory usage, for small~$k$. Our main result improves the quantum time exponent for $k=3$ to 0.2846, while keeping the memory exponent at $0.1887$.}\label{table:k234complexities}
\end{table}

\subsection{Our results}\label{sec:result}

Our main result is an improvement of the quantum time complexity of 3-tuple sieving, from $2^{0.3098d+o(d)}$ to $2^{0.2846d+o(d)}$. 
Interestingly, this means quantum 3-tuple sieving is now faster than classical 2-tuple sieving which is still the best classical heuristic algorithm for SVP  (and of course 3-tuple sieving uses less memory than 2-tuple sieving, which is the main advantage of considering $k>2$). 
Our improvement is relatively small and will not scare cryptographers who base their constructions on the assumed hardness of SVP. Often when deciding on an appropriately large key size to guarantee a certain level of security for their cryptosystem, they already allow for a quadratic quantum speedup over the best known classical attack (the best speedup one can hope for just using Grover's algorithm or amplitude amplification without exploiting further specific structure of the problem), which is already better than all quantum speedups that we actually know how to achieve.  
In addition, our results include factors of the form $2^{o(d)}$ in the runtime, which are insignificant for asymptotic analysis but may be quite large for the dimensions used in practice. 
However, we remark that our improvement is significantly larger than recent progress on this front.  For comparison, the most recent improvement in the exponent for quantum 2-tuple sieving was from 0.2570~\cite{CL21} to 0.2563~\cite{BCSS22}. 

The core computational problem that we would like to solve, and which dominates the cost of 3-tuple sieving, is the following:

\begin{problem}[Finding many solutions]\label{prob:finding-many-3-tuples} 
    Given a list $L$ of $m$ i.i.d.\ uniform samples from the unit sphere in $\R^d$, find $m$ $3$-tuples of distinct $\vx, \vy, \vz \in L$ such that $\norm{\vx - \vy - \vz} \leq 1$. 
    We will refer to such a tuple $(\vx, \vy, \vz)$ as a \emph{3-tuple solution}.
\end{problem}  
\noindent As argued before, we could also allow all sign patterns $\norm{\vx \pm \vy \pm \vz}$, but asymptotically this does not matter because there are only 4 such patterns.
We assumed here for simplicity that the $m$ initial vectors all have norm exactly~1.\footnote{It may seem odd to start with vectors $\vx,\vy,
\vz$ of norm~1 and then to find new vectors $\vx-\vy-\vz$ whose norm is still (at most)~1. However, the way this is actually used in practice is to aim at finding vectors with norm $\leq 1-\mu$ for some $\mu$ that is inverse-polynomially small in~$d$: big enough to ensure that a polynomially-small number of sieving iterations results in quite short vectors, and small enough to only affect the runtime up to subexponential factors (i.e., $o(1)$ in the exponent). Aiming at norm~$\leq 1$ rather than $\leq  1-\mu$ is just to simplify notation.}  
To ensure that $m$ such tuples exist with high probability over the choice of $L$, the list size $m$ has to be at least roughly $(27/16)^{d/4 + o(d)} \approx 2^{0.1887 d}$~\cite{HK2017}, providing a memory lower bound for 3-tuple sieving.  

Expanding $\norm{\vx - \vy - \vz}^2$ shows that unit vectors $\vx,\vy,\vz$ satisfy $\norm{\vx - \vy - \vz} \leq 1$ if and only if $\innerP{\vx}{\vy} + \innerP{\vx - \vy}{\vz} \geq 1$. This fact lets us replace the condition $\norm{\vx-\vy-\vz} \leq 1$ by two conditions on the inner products $\innerP{\vx}{\vy}$ and $\innerP{\vx - \vy}{\vz}$, motivating us to search for 3-tuple solutions $(\vx,\vy,\vz)$ by first searching for a pair $(\vx,\vy)$ with sufficiently large inner product. 
Our main result is indeed a faster quantum algorithm for a mildly relaxed version of~\Cref{prob:finding-many-3-tuples}. Slightly simplified (by omitting approximations in the inner products), this relaxed problem is: 

\begin{problem}[Finding many solutions with a stronger property]\label{prob:finding-many-3-tuples-relaxed}
    Given a list $L$ of $m$ i.i.d.\ uniform samples from the unit sphere in $\R^d$, find $m$ $3$-tuples of distinct $\vx, \vy, \vz \in L$ such that $\innerP{\vx}{\vy}\geq 1/3$ and $\innerP{\vx-\vy}{\vz}\geq 2/3$ (implying $\norm{\vx - \vy - \vz} \leq 1$). 
\end{problem}

\noindent The latter problem demands something stronger than $\norm{\vx - \vy - \vz} \leq 1$. 
This slightly increases the required list size~$m$ that ensures $m$ triples with the stronger property exist, but only by a factor~$2^{o(d)}$ as a result of the chosen inner-product conditions (e.g., see~\cite{HK2017} or~\Cref{lem:size-of-Tsol-application} below). 
This factor is negligible for our asymptotic purposes, and allows us to solve~\Cref{prob:finding-many-3-tuples} with list size $m \approx 2^{0.1887 d}$ using an algorithm for~\Cref{prob:finding-many-3-tuples-relaxed}. 

We will find those $m$ triples in~\Cref{prob:finding-many-3-tuples-relaxed} one by one. To find one good triple, consider the following approach: 
\begin{enumerate}
\item Create a uniform superposition over all $\vx\in L$, and conditioned on $\vx$ use amplitude amplification to create (in a second register) a superposition over all $\vy\in L$ such that $\innerP{\vx}{\vy}\geq 1/3$. 
\item Starting from the state of the previous step, conditioned on $\vx,\vy$ use amplitude amplification to create (in a third register) a uniform superposition over all $\vz\in L$ such that $\innerP{\vx-\vy}{\vz}\geq 2/3$, and set a ``flag'' qubit to~1 if such a $\vz$ exists. 
\end{enumerate}
Then use amplitude amplification on top of this two-step procedure to amplify the part of the state where the flag is~1. 
Measuring the resulting state gives us one of the triples $(\vx,\vy,\vz)$ that we want. Repeating $\tilde{O}(m)$ times, we will obtain (except with negligibly small error probability) $m$ triples satisfying the stronger property, and hence $m$ new vectors with norm at most~$1$. 
This approach is already better than a basic amplitude amplification on all $m^3$ triples $(\vx,\vy,\vz)$, due to the somewhat surprising properties of ``two-oracle search''~\cite{KLL15twoOracles}. 
However, when executed as stated, this procedure can be shown to yield an overall runtime of roughly $2^{0.3350d}$ for $m \approx 2^{0.1887 d}$, and therefore does not yet outperform the state-of-the-art quantum runtime $2^{0.3098d}$ from~\cite{CL23} (see~\Cref{sec:overview-algo}). 
The costs of steps~1 and~2 can also be seen to be unbalanced, suggesting there may be room for improvement. 

To get a faster algorithm, we improve the search for vectors satisfying the stronger property by locality-sensitive filtering using random product codes (RPCs, following~\cite{BDGL16}). Roughly speaking, the idea here is to choose a certain number of sufficiently random center points, and do a preprocessing step that, for each of the $m$ vectors in $L$, writes down their closest center points. The properties of RPCs allow us to do this efficiently. 
Then for a given $\vx$ in step~1, we can improve the search for a $\vy \in L$ satisfying $\innerP{\vx}{\vy}\geq 1/3$ by restricting the search to those $\vy\in L$ that share a close center point with~$\vx$.  
Specifically, using the data structure prepared during preprocessing, we can quickly create, for a given $\vx$, a superposition over all center points $\vc$ that are close to $\vx$, and subsequently, for each such $\vc$, a superposition over all $\vy\in L$ that are close to $\vc$. 
Since those $\vy$ are relatively close to $\vx$, they are more likely to satisfy $\innerP{\vx}{\vy}\geq 1/3$. Putting amplitude amplification on top of this then allows us to quickly prepare a superposition over pairs $(\vx,\vy)$ with $\innerP{\vx}{\vy}\geq 1/3$ that share a close center point, improving over the cost of the original step~1. 

Similarly, for step~2, it is easier to find a $\vz\in L$ close to $\vx-\vy$ if we focus only on those $\vz \in L$ that share a close center point with $\vx-\vy$. This approach may overlook some close $(\vx,\vy)$ pairs or some $\vz\in L$ close to $\vx-\vy$  (if $\vx$ and $\vy$, respectively $\vx-\vy$ and $\vz$, do not share a close center point). However, we can show that a careful application of these ideas enables us to balance the costs of steps~1 and~2, and that this modified method (repeated $\tilde{O}(m)$ times and choosing a fresh RPC once in a while to ensure no good triple is overlooked all the time) will find $m$ good triples. 
This results in a quantum algorithm that solves~\Cref{prob:finding-many-3-tuples} with list size $m \approx 2^{0.1887d}$ in time $2^{0.2846 d}$, giving a speedup over the previous best known time complexity for this choice of~$m$.  

\subsection{Discussion and future work}

Various ingredients of our algorithm, including the use of RPCs and of course amplitude amplification, have been used before in quantum $k$-tuple sieving algorithms. What distinguishes our algorithm, and enables our speedup, is first our nested use of amplitude amplification (to do two-oracle search in the vein of~\cite{KLL15twoOracles}), and second the way we leverage the preprocessing that we do in advance, where we write down for each lattice vector in $L$ its set of close center points from the RPC. Since we can afford relatively expensive preprocessing in terms of both time and memory, this allows us to balance different terms in the time complexity.
Altogether, this yields a more sophisticated nested amplitude amplification that results in our speedup. 

One of the main messages of our paper is that the toolbox of techniques that have been used for sieving is not exhausted yet, and can still give improved results when combined with a well-chosen preprocessing step.
How much further can this be pushed? Surprisingly, we only used basic amplitude amplification here, not the more sophisticated quantum-walk approaches that are good at finding collisions (and which have been used for sieving~\cite{CL21,BCSS22}). 
In fact, the work that led to our algorithm for~\Cref{prob:finding-many-3-tuples-relaxed} started from a quantum-walk approach, but optimizing the parameters suggested that the current version (which can be viewed as a quantum walk with vertex size $1$) is optimal. 
This might be due to the fact that there are \emph{many} 3-tuple solutions that need to be found --- a setting in which quantum walks may be less helpful (consider that for \emph{element distinctness}, the optimal algorithm uses quantum walks~\cite{Amb04ED}, whereas for its many-solution variant, \emph{collision finding}, there is an optimal quantum algorithm that uses only amplitude amplification~\cite{BHT98}). 
Are there other ways to improve the algorithm further by using quantum walks? Also, can we improve the best quantum algorithm for 2-tuple sieving? This would give the fastest known heuristic attack on SVP, rather than just a faster attack with moderate memory size. Finally, do our techniques lead to improvements for 4-tuple sieving?

\subsection{Organization}

The remainder of this paper is organized as follows. \Cref{sec:prelim} provides preliminaries on the computational model, data structure, and quantum-algorithmic techniques used in this work, as well as on properties of RPCs and the unit sphere relevant to our analysis. In~\Cref{sec:main-quantum-algorithm}, we present our new quantum algorithm, and the details of its application to SVP are given in~\Cref{sec:application}. 

\section{Preliminaries} \label{sec:prelim}

\subsection{Notation}\label{sec:notation}

Throughout the paper, $d$ will always be the dimension of the ambient space $\R^d$, $\log$ without a base means the binary logarithm, $\ln=\log_e$ the natural logarithm, and $\exp(f)=e^f$. 
We write $p_{\theta}$ as shorthand for $(1-\cos^2(\theta))^{d/2}$ (as motivated by~\Cref{sec:sphere}). The set of rotation matrices over $\R^d$ is denoted by $\mathrm{SO}(d)$. 

We write $x \sim D$ to denote that $x$ is a sample from the distribution $D$. For $M\in \Z^+$ and a set $S\subseteq \R^d$, we let $\cU(S,M)$ denote the distribution that samples $C \subseteq S$ by selecting $M$ elements from $S$ independently and uniformly with replacement. Throughout this work, we refer to multisets simply as sets, so the cardinality (or ``size'') $|C|$ corresponds to the total number of sampled elements, counting repetitions. When an argument requires using independence explicitly, we view $C$ as an ordered sequence of random variables. 
We write $x \sim \cU(S)$ as shorthand for $\{x\} \sim \cU(S,1)$.

For $a,b \in \R$ and $\epsapprox > 0$, we write $a \approx_\epsapprox b$ if $|a - b| \leq \epsapprox$. In particular, there is some dependency on~$\epsapprox$ in our analysis, so we will fix $\epsapprox \coloneqq 1/(\log d)^2$ to absorb this dependence into the $2^{o(d)}$ notation. 
We also write $a =_d b$ as shorthand for $b2^{-o(d)} \leq a \leq b2^{o(d)}$, i.e., $a$ and $b$ are equal up to subexponential factors in $d$. We let $a \geq_d b$ denote $a \geq b2^{-o(d)}$, and $a \leq_d b$ denote $a \leq b2^{o(d)}$.  

Whenever we write $\ket{0}$, we mean the all-zero computational basis state over the number of qubits implied by the context, not necessarily a single-qubit state. 
When considering unitaries~$U$ on quantum states in some Hilbert space $H$, we will often only be interested in its application on a subset $X \subseteq H$ of states, and we will sometimes (with slight abuse of notation) write ``$U \colon \ket{\psi} \mapsto \ket{\bot}$ if $\ket{\psi} \in H \setminus X$'', where $\ket{\bot}$ is a substitute for a quantum state that is irrelevant for our analysis.

\subsection{Tail bounds}

\begin{lemma}[Chernoff bound~\cite{chernoff1952Bound}, {\cite[Section~4.2.1]{MU05probability}}]\label{lem: Chernoff}
    Let $X = \sum_{i=1}^m X_i$ be a sum of independent random variables $X_i \in \{0,1\}$ 
    and define $\mu \coloneqq \E[X]$.
    Then 
    \begin{enumerate}[label=(\roman*)]
        \item $\Pr\left[X \geq (1 + \delta) \mu\right] \leq \big( \frac{e^\delta}{(1+\delta)^{1+\delta}}\big)^{\mu} \leq e^{-\frac{\delta^2}{2 + \delta} \mu}$ for all $\delta \geq 0$.\footnote{The second inequality follows because $\delta-(1+\delta)\ln(1+\delta)+\delta^2/(2+\delta)<0$ for all $\delta >0$; for $\delta=0$ it is trivial.}  
        \item $\Pr\left[X \leq (1 - \delta)\mu\right] \leq e^{-\frac{\delta^2}{2} \mu}$ for all $\delta \in (0,1)$.  
        \item $\Pr\left[|X - \mu| \geq \delta \mu\right] \leq 2 e^{-\frac{\delta^2}{3} \mu}$ for all $\delta \in (0,1)$. 
    \end{enumerate} 
\end{lemma}

\noindent We will use the following immediate corollary.  

\begin{corollary}[Simple application of the Chernoff bound]\label{cor:simple-application-of-Chernoff}
    Let $m \colon \N \to \N$. 
    For $d \in \N$, let $X(d) = \sum_{i=1}^{m(d)} X_i(d)$ be a sum of independent random variables $X_i(d) \in \{0,1\}$.  Then $X(d) \leq_d \max\{1, \E[X(d)]\}$, except with probability $e^{-\omega(d)}$. 
    Moreover, if $\E[X(d)] = \omega(d)$, then $X(d) =_d \E[X(d)]$, except with probability $e^{-\omega(d)}$.
\end{corollary}

\begin{proof}
Let $\mu \coloneqq \E[X]$, where $X = X(d)$. 
Applying part~$(i)$ of~\Cref{lem: Chernoff} with $\delta = d^2 \max(1, 1/\mu)$  yields $\Pr[X \geq (1 + \delta) \mu] \leq e^{-\delta \mu / 2} \leq e^{-d^2/2} = e^{-\omega(d)}$, where the last inequality uses that $\delta \geq 2$ for $d \geq 2$. 
Note that $(1 + \delta) \mu \leq_d \max\{1, \mu\}$ by definition of $\delta$. (This can be seen by separating the case $\mu \leq 1$ and $\mu > 1$.) 
To prove the second part, assume that $\mu = \omega(d)$. Applying part~$(iii)$ of~\Cref{lem: Chernoff} with (say) arbitrary constant $\delta \in (0,1)$ yields $\Pr[|X - \mu| \geq \delta \mu] \leq 2 e^{-\delta^2 \mu / 3} = e^{-\omega(d)}$ by assumption on $\mu$, from which the claim follows. 
\end{proof}

\subsection{Computational model and quantum preliminaries}

\subsubsection{Computational model}\label{sec:compmodel}

Our computational model is a classical computer (a classical random-access machine) that can invoke a quantum computer as a subroutine. This classical computer can also write bits to a {quantum-readable classical-writable} classical memory (QCRAM). This memory stores an $n$-bit string $w = w_0 \dots w_{n-1}$, and supports \emph{quantum random access queries} or \emph{QCRAM queries}, which correspond to calls to the unitary
$O_w \colon \ket{i, b} \mapsto \ket{i, b \oplus w_i}$ 
for $i \in \{0, \dots, n-1\}$ and $b \in \{0,1\}$. Note that QCRAM itself (as a memory) remains classical throughout the algorithm: it stores a classical string rather than a quantum superposition over strings.  
The notion of QCRAM is commonly used in quantum algorithm design for efficient quantum queries to classical data, allowing us to query (read) multiple bits of the classical data in superposition. In contrast, write operations (i.e., changes to the stored string $w$) can only be done  classically.

In classical algorithms, allowing random access memory queries with unit or logarithmic cost is standard practice. 
The intuition is that the $n$ bits of memory can be, for example, arranged on the leaves of a binary tree with depth $\lceil \log_2 n\rceil$, and querying the $i$-th bit corresponds to traversing a $\lceil \log_2 n\rceil$-length path from the root to the $i$-th leaf of this binary tree.  For similar reasons, QCRAM queries are often assumed ``cheap'' to execute (i.e., in time $O(\log n)$) once the classical memory is stored in QCRAM.

While the physical implementation of such a device is nontrivial and controversial due to the challenge and expense of doing quantum error-correction on the whole device (whose size will have to be proportional to the number of stored bits rather than to its logarithm, though its depth will still be logarithmic), QCRAM remains conceptually acceptable for theoretical purposes. This is particularly true in theoretical computer science contexts, where classical RAM is likewise assumed to be fast and error-free without a concern for error correction. One may hope that in the future, quantum hardware will be able to implement such memory with comparable reliability and efficiency. In cryptography it is also important to learn the runtime of the best quantum algorithms for breaking cryptography under fairly optimistic assumptions on the hardware, which is exactly what we do in this paper.

In our computational model, we will count one RAM operation or one QCRAM write operation in the  classical machine, or one elementary gate in a quantum circuit, as one ``step''. When we refer to the ``time'' or ``time complexity'' of an algorithm or subroutine, we mean the total number of steps it takes on a worst-case input. We will typically count the number of QCRAM queries separately.

\subsubsection{Amplitude amplification} \label{sec:prelim-AA}

Our main quantum algorithmic tool will be amplitude amplification~\cite{BHMT02}. 
The goal of amplitude amplification is to project an easily generated state $\ket{\psi}$ onto a ``marked'' subspace. 
\emph{Fixed-point} amplitude amplification~\cite{yoder2014FixedPointSearch,gilyen2018QSingValTransf} ensures that, with high probability, the output state is the desired state (assuming it is nonzero), even if we run for more steps than necessary. The version that we give below ensures that the algorithm always stops after a fixed number of steps, even if $\ket{\psi}$ does not overlap the marked space at all. 

\begin{lemma}[Fixed-point amplitude amplification (implicit in~{\cite{yoder2014FixedPointSearch,gilyen2018QSingValTransf}})]\label{lem:AA-without-knowing-norm}
Let {\tt Samp} be a unitary implementable in time $\sf S$, and let $\ket{\psi} \coloneqq {\tt Samp}\ket{0}$. Let $\Pi$ be a projector on the Hilbert space $H$ of $\ket{\psi}$, and let {\tt Check} be a unitary implementable in time ${\sf C}$ such that, for all $\ket{\phi} \in H$, ${\tt Check}\ket{\phi} \ket{0}=\Pi\ket{\phi}\ket{1}+(I-\Pi)\ket{\phi}\ket{0}$. 
For all $r \geq 1$ and $\delta\in(0,1/2]$, there exists a quantum algorithm ${\tt AA}_{r}({\tt Samp},{\tt Check})$ that satisfies the following, except with probability $\delta$: 
\begin{enumerate}
    \item[(i)] If $\norm{\Pi\ket{\psi}} \geq \frac{1}{r}$, then it outputs the state $\ket{\psi_{\mathrm{out}}} = \frac{\Pi\ket{\psi}}{\norm{\Pi\ket{\psi}}} \ket{1}$ in time $O(\log(\frac{1}{\delta}) \norm{\Pi\ket{\psi}}^{-1}({\sf S}+{\sf C}))$. 
    \item[(ii)] If $\norm{\Pi\ket{\psi}} = 0$, then it outputs the state $\ket{\psi_{\mathrm{out}}} = \ket{\psi} \ket{0}$ in time $O(\log(\frac{1}{\delta}) r ({\sf S}+{\sf C}))$. 
\end{enumerate} 
\end{lemma} 

\noindent We call the auxiliary qubit in the output of amplitude amplification its \emph{flag} qubit or \emph{flag} register. 

\begin{proof} 
By~\cite{yoder2014FixedPointSearch,gilyen2018QSingValTransf}, there exists a  universal constant $\eta > 0$ and a quantum algorithm $\cA_{r'} = \cA_{r'}({\tt Samp}, {\tt Check})$ that generates a state $\ket{\psi_{\mathrm{out}}}$ in time $O(r'({\sf S}+{\sf C}))$, satisfying: 
\begin{enumerate}
    \item[(1)] If $\norm{\Pi\ket{\psi}}\neq 0$ and $r' \geq \eta \log(\frac{1}{\delta}) \norm{\Pi\ket{\psi}}^{-1}$, then $\norm{\ket{\psi_{\mathrm{out}}} - \frac{\Pi\ket{\psi}}{\norm{\Pi\ket{\psi}}}\ket{1}} \leq \delta$.
    \item[(2)] If $\norm{\Pi\ket{\psi}}=0$ and $r' \geq 1$, then $\ket{\psi_{\mathrm{out}}}=\ket{\psi}\ket{0}$.
\end{enumerate}
Our lemma extends this result slightly: in part~$(i)$, we require the runtime to be much smaller than in part~$(ii)$ whenever 
the actual (possibly unknown) value of $\norm{\Pi\ket{\psi}}$ is significantly larger than the lower bound $\frac{1}{r}$ that our algorithm has for it.

Specifically, we define the quantum algorithm ${\tt AA}_{r}({\tt Samp}, {\tt Check})$ as follows. 
It first runs $\cA_{r'}$ with $r' = 1$ and measures the auxiliary qubit. If the measurement outcome is 1, then it stops and outputs the resulting state. Otherwise, it doubles $r'$ (i.e., replaces $r'$ by $2r'$) and repeats. This continues until $r' \geq 2 \eta \log(\frac{2}{\delta}) r$, at which point it outputs the current state. 

By construction, ${\tt AA}_{r}({\tt Samp}, {\tt Check})$ always terminates, and outputs a state $\ket{\psi_{\mathrm{out}}}$. 
Note that if ${\tt AA}_{r}({\tt Samp}, {\tt Check})$ terminates after $\ell$ repetitions, then the aforementioned result implies that the total runtime is $O(2^{\ell} ({\sf S}+{\sf C}))$, because $\sum_{i=1}^\ell 2^i = O(2^{\ell})$.   

Suppose $\norm{\Pi\ket{\psi}} \geq \frac{1}{r}$. Let $\ell$ be the smallest integer such that $2^\ell \geq \eta \log(\frac{2}{\delta}) \norm{\Pi\ket{\psi}}^{-1}$. 
By statement~(1), the inner product $\alpha$ between $\frac{\Pi\ket{\psi}}{\norm{\Pi\ket{\psi}}}\ket{1}$ and the output state $\ket{\psi_{\mathrm{out}}}$ of $\cA_{2^\ell}$ satisfies $|\alpha| \geq 1 - \delta/2$, so the probability that measuring the ancilla qubit yields outcome 1 is at least $(1 - \delta/2)^2 \geq 1 - \delta$.  
In other words, with probability at least $1 - \delta$, ${\tt AA}_{r}({\tt Samp}, {\tt Check})$ generates the state $\ket{\psi_{\mathrm{out}}} = \frac{\Pi\ket{\psi}}{\norm{\Pi\ket{\psi}}} \ket{1}$ after at most $\ell$ repetitions, from which part~$(i)$ follows.

On the other hand, if $\norm{\Pi\ket{\psi}} = 0$, statement~(2) implies that ${\tt AA}_{r}({\tt Samp}, {\tt Check})$ always generates $\ket{\psi_{\mathrm{out}}}=\ket{\psi}\ket{0}$. The complexity claim follows because, by construction of the algorithm, it terminates after $\ell$ repetitions, for some $\ell$ satisfying $2^{\ell} = O(\log(\frac{1}{\delta}) r)$. 
\end{proof}

\begin{remark}[Choice of $\delta$]\label{rem:AA-choice-delta}
When applying~\Cref{lem:AA-without-knowing-norm}, we will always consider a superexponentially small $\delta$, say $\delta= \exp(-2^{\sqrt{d}})$, and will not state its value explicitly. This implicit choice ensures that the error probability of ${\tt AA}_{r}({\tt Samp},{\tt Check})$ is $2^{-\omega(d)}$, and incurs a cost of only a factor $\log(1/\delta)\leq 2^{o(d)}$ in the overall complexity of ${\tt AA}_{r}({\tt Samp},{\tt Check})$, a cost that is negligible for our purposes. 
\end{remark}

\begin{remark}[Unstructured search as a special case]\label{rem:AA-special-case-search}
Amplitude amplification allows us to search in a finite set $M_0$ for elements of a (possibly empty) subset $M_1$, and create a uniform superposition $\ket{M_1}$ over these elements if they exist. Namely, let {\tt Samp} be a quantum algorithm that maps $\ket{0}$ to a uniform superposition $\ket{\psi}$ over the elements of $M_0$, and {\tt Check} a quantum algorithm that maps $\ket{x}\ket{0} \mapsto \ket{x}\ket{1}$ if $x \in M_1$ and that acts as identity for $x \in M_0 \setminus M_1$.  
Taking $r \coloneqq |M_0|$ (assuming $|M_0|$ is known or encoded in {\tt Samp}), \Cref{lem:AA-without-knowing-norm} implies that, with probability at least $1- \delta$, ${\tt AA}_{r}({\tt Samp},{\tt Check})$ generates a state $\ket{\psi_{\mathrm{out}}}$ in time $O(\sqrt{|M_0|}\log\frac{1}{\delta})$, which satisfies $\ket{\psi_{\mathrm{out}}} = \ket{M_1}\ket{1}$ if $M_1 \neq \emptyset$, and  $\ket{\psi_{\mathrm{out}}} = \ket{\psi}\ket{0}$ otherwise. 
\end{remark}

\subsubsection{Two-oracle search} \label{sec:prelim-two-oracle-search} 

One of the most important quantum algorithmic primitives is Grover's algorithm~\cite{Gro96}, which solves the problem of \emph{oracle search}. Consider a set $M_0$ of elements that are easy to ``sample'' --- meaning that we can generate some superposition over the elements of $M_0$, where we think of the squared amplitudes as the sampling probabilities. 
For a marked subset $M_1\subseteq M_0$, consider the problem of searching over $M_0$ for an element of $M_1$. If we can sample an element of $M_0$ in complexity ${\sf S}$ and we can check membership in $M_1$ in complexity ${\sf C}$, then there is a quantum algorithm~\cite{BHMT02} (see~\Cref{lem:AA-without-knowing-norm} for a slightly different ``fixed-point'' version) that finds an element of $M_1$ in complexity (neglecting constants)
\[\frac{1}{\sqrt{\varepsilon}}\left({\sf S}+{\sf C}\right)\]
where $\varepsilon$ is the probability that an element sampled from $M_0$ is in $M_1$. 
In our applications, the sampling complexity ${\sf S}$ is negligible, resulting in complexity $\frac{1}{\sqrt{\varepsilon}}{\sf C}$.

Oracle search is often called \emph{unstructured search}, since the oracle abstracts away any potential structure of the problem that an algorithm might be able to take advantage of. Although this makes it quite general, we can in some cases take advantage of a small amount of structure. A simple case is \emph{two-oracle search}, first studied in~\cite{KLL15twoOracles} for the case of a unique marked element. In this problem, one wants to find an element of $M_2\subseteq M_0$, but in addition to having access to an oracle for checking membership in $M_2$ (at cost ${\sf C}_2$), one also has access to an oracle for checking membership in a set $M_1$ such that $M_2\subseteq M_1\subseteq M_0$ (at cost ${\sf C}_1$). See~\Cref{fig:high-level-sketch} for a depiction. 
Assuming ${\sf C}_1\ll {\sf C}_2$, this additional structure gives an advantage, intuitively because one can always first cheaply check membership in $M_1$, and for nonmembers, not waste time on a more expensive  membership check in $M_2$. This significantly reduces the number of times we check membership in~$M_2$. By variable-time search~\cite{Amb10VariableTime}, if $\varepsilon_1$ is the probability that an element sampled from $M_0$ is in $M_1$, and $\varepsilon_2\leq \varepsilon_1$ is the probability that an element sampled from $M_0$ is in $M_2$, then a quantum algorithm can find an element of $M_2$ in complexity 
\[
    \sqrt{\frac{\varepsilon_1}{\varepsilon_2}} \left(\frac{1}{\sqrt{\varepsilon_1}}{\sf C}_1+{\sf C}_2\right) = \frac{1}{\sqrt{\varepsilon_2}}{\sf C}_1+\sqrt{\frac{\varepsilon_1}{\varepsilon_2}}{\sf C}_2
\]
where we neglect logarithmic factors and assume the cost ${\sf S}$ of sampling from $M_0$ is negligible.

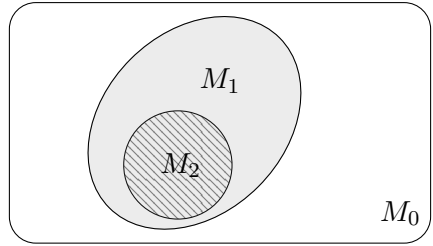
\begin{figure}[H]
    \centering
    \vspace*{1cm}
    \begin{tikzpicture}[scale=0.8]
    \draw[rotate=0, pattern=north west lines, pattern color=white, rounded corners=10pt] (-3.5,-2) rectangle (3.5,2);
    \draw[pattern=north west lines, pattern color=gray, opacity=1] (-0.7,-0.7) ellipse (0.9cm and 0.9cm);
    \draw[fill=gray, opacity=0.15, rotate=45] (-0.3,0.3) ellipse (2cm and 1.5cm);
    \draw[draw=black, rotate=45] (-0.3,0.3) ellipse (2cm and 1.5cm);
    \node at (3,-1.5) (A) {$M_0$};
    \node at (0,0.7) (B) {$M_1$};
    \node at (-0.65,-0.7) (C) {$M_2$};
    \end{tikzpicture}
    \vspace*{0.4cm}
    \caption{Illustration of the two-oracle search setup, where the task is to search for elements in the search space $M_0$ that belong to a marked subset $M_2 \subseteq M_0$, given the ability to check membership in both $M_2$ and some subset $M_1$ satisfying $M_2 \subseteq M_1 \subseteq M_0$.}
    \label{fig:high-level-sketch}
\end{figure}

\subsubsection{Finding all elements from a set with unknown size}\label{sec:find-all}

Suppose you aim to find all elements of a finite set $X$ by querying an algorithm {\tt Samp} that samples (say) uniformly from $X$. By the coupon collector's argument, $|X|^{1 + o(d)}$ calls to {\tt Samp} suffice to succeed with overwhelming probability. 
While this might seem to require knowing $|X|$ to decide when to stop, it actually suffices to keep sampling until no new element has been observed for some time --- provided a suitable (possibly loose) upper bound on $|X|$ is known. 

\clearpage 
\begin{lemma}[Folklore]\label{lem:find-all}
    Let $X$ be a nonempty finite set and let ${\tt Samp}$ be a (classical, resp.\ quantum) algorithm that samples from $X$ in time ${\sf S}$. 
    Suppose there exists $0 < \varepsilon \leq 1$ such that, for all $x \in X$, 
    \[
        \Pr[{\tt Samp} \text{ outputs } x] \geq \frac{\varepsilon}{|X|} 
    \]
    where the probability is over the internal randomness of ${\tt Samp}$. 
    For all $\delta \in (0,1)$, there exists a (classical, resp.\ quantum) algorithm that, given a positive integer $\searchthreshold \geq \frac{1}{\varepsilon}\ln(|X|/\delta)$, outputs all elements of $X$ in time $O(\searchthreshold |X| \ln(|X|){\sf S})$ using $O(\searchthreshold |X| \ln(|X|))$ applications of ${\tt Samp}$, except with probability $\delta$. The algorithm uses $O(|X|)$ classical bits beyond the memory used by {\tt Samp}. 
    (The algorithm does not require prior knowledge of $|X|$.)  
\end{lemma}

\noindent Several arguments for this lemma are possible; the one presented below is inspired by~\cite{Iiridayn2020}. 

\begin{proof}
Let $\cA({\tt Samp}, \searchthreshold)$ be the algorithm that repeatedly calls $\tt Samp$ and maintains the set $S \subseteq X$ of distinct elements seen so far, and that outputs $S$ once no new element has been observed in the last $(|S| + 1) \searchthreshold$ calls. (This stopping rule is inspired by~\cite{Iiridayn2020}.) 
The runtime of $\cA({\tt Samp}, \searchthreshold)$ is determined by the number of calls to ${\tt Samp}$, and its additional memory usage beyond that of ${\tt Samp}$ is $O(|X|)$. 

We claim that, except with probability $\delta$, the following holds for all $s \in \{0, \dots, |X| -1\}$: immediately after the set $S$ constructed during $\cA({\tt Samp}, \searchthreshold)$ is of size $s$, a new element of $X$ is found within the next $\lceil |X| \searchthreshold / (|X| - s) \rceil$ calls to $\tt Samp$. 
Since $\lceil |X| \searchthreshold /(|X| - s) \rceil \leq (s + 1)\searchthreshold$, this implies that the stopping rule is satisfied if and only if all elements of $X$ have been found, meaning that $\cA({\tt Samp}, \searchthreshold)$ outputs all elements of $X$. 
Moreover, it implies that the total number of calls to ${\tt Samp}$ is upper bounded by 
\[\sum_{s=0}^{|X| - 1} \left\lceil \frac{|X| \searchthreshold}{|X| - s} \right\rceil + (|X| + 1)\searchthreshold \leq |X| \searchthreshold \sum_{j = 1}^{|X|} \frac{1}{j} + O(|X|\searchthreshold) = O(\searchthreshold |X| \ln(|X|))
\] 
since $\sum_{j = 1}^{|X|} \tfrac{1}{j} \leq \ln(|X|) + 1$. 

It remains to prove the claim. Fix $s \in \{0, \dots, |X| -1\}$ and consider the call of ${\tt Samp}$ (during $\cA({\tt Samp}, \searchthreshold)$) that returns the $s$-th new element of $X$. 
By assumption on $\varepsilon$, the probability that a single run of ${\tt Samp}$ returns a new element of $X$ is at least $(|X| - s)\varepsilon/|X|$, so the probability that the next $\lceil |X| \searchthreshold /(|X| - s) \rceil$ (independent) calls to $\tt Samp$ yield no new element is at most \[\left(1 - (|X| - s)\frac{\varepsilon}{|X|}\right)^{\left\lceil \frac{|X| \searchthreshold}{|X| - s} \right\rceil} \leq e^{- \varepsilon \searchthreshold}.\] 
The claim now follows from a union bound over all $s \in \{0, \dots, |X| - 1\}$, because $|X| e^{- \varepsilon \searchthreshold} \leq \delta$ by assumption on $\searchthreshold$. 
\end{proof}

\subsubsection{Relations and associated data structures}\label{sec:relation-and-data-structure-requirement}

We will facilitate the application of two-oracle search using carefully constructed relations and associated data structures. 

\begin{definition}\label{def:R-collision}
A \emph{relation} $R$ on $X\times Y$ is a subset $R\subseteq X\times Y$, equivalently viewed as a function $R:X\times Y\rightarrow\{0,1\}$. 
For every $x\in X$, we let $R(x)=\{y\in Y:(x,y)\in R\}$.
We let $R^{\transpose}\subseteq Y\times X$ be the relation defined by $(y,x)\in R^{\transpose}$ if and only if $(x,y)\in R$. 
A pair $x,x'\in X$ such that $R(x)\cap R(x')\neq\emptyset$ is called an \emph{$R$-collision}. 
\end{definition}

\noindent Note that the latter is a natural generalization (from functions to relations) of the concept of a collision. 

The standard way of accessing a function $f:X\rightarrow Y$ is through queries, meaning an algorithm is given a description of $f$ that allows the ability to efficiently compute, for any $x\in X$, the value $f(x)$. If $f$ is a 1-to-1 function, then a more powerful type of access allows the efficient computation of $f^{-1}(y)$ for any $y\in Y$. This is not always possible from a simple description of $f$, but is, for example, possible if all function values of $f$, $(x,f(x))$, are stored in a data structure, perhaps constructed during some preprocessing step. 

For relations, we can distinguish three types of (quantum) access. First, the simplest type, \emph{query access}, means that for any $(x,y)\in X\times Y$, it is possible to efficiently check if $(x,y)\in R$. 
A second type that is more analogous to evaluation of a function $f$ is \emph{forward superposition query access} to $R$, meaning we can query an oracle $\cO_{R}$ that acts, for all $x\in X$, as:  
\begin{align*}
    \ket{x}\ket{0} &\mapsto \left\{\begin{array}{ll}
    \ket{x}\sum_{y\in R(x)}\frac{1}{\sqrt{|R(x)|}}\ket{y} &\quad \mbox{if }R(x)\neq \emptyset\\
    \ket{x}\ket{\bot} &\quad \mbox{otherwise.}
    \end{array}\right. 
\intertext{The third type of access to $R$, analogous to having standard and inverse query access to $f$, is to have query access to both $\cO_{R}$ and $\cO_{R^{\transpose}}$. That is, one can not only implement a forward superposition query, but also its inverse: }  
    \ket{y}\ket{0} &\mapsto \left\{\begin{array}{ll}
    \ket{y}\sum_{x\in R^{\transpose}(y)}\frac{1}{\sqrt{|R^{\transpose}(y)|}}\ket{x} &\quad \mbox{if }R^{\transpose}(y)\neq \emptyset\\
    \ket{y}\ket{\bot} &\quad \mbox{otherwise.}
    \end{array}\right.
\end{align*}
We will always implement this third type of access by working with a data structure $D(R)$ that stores $R$ in QCRAM. 

\paragraph{Data structures for relations.} 
Specifically, we will store a relation $R\subseteq X\times Y$ in a classical data structure, denoted $D(R)$, in a way such that the following operations can be performed using $O(\log|X|+\log|Y|)$ time and classical memory: 
\begin{itemize}
    \item \textbf{Insert:} For any $(x,y)\in (X\times Y)\setminus R$, add $(x,y)$ to $D(R)$.
    \item \textbf{Lookup by $x$:} For any $x\in X$, return a pointer to an array containing all $y$ such that $(x,y)\in R$, and its size $|R(x)|$. 
    \item \textbf{Lookup by $y$:} For any $y\in Y$, return a pointer to an array containing all $x$ such that $(x,y)\in R$, and its size $|R^{\transpose}(y)|$.
\end{itemize}
To accomplish this, we store the elements $(x,y)\in R$ in two different ways: once in a keyed data structure with $x$ as the ``key'' and $y$ as an associated ``value'', and once in another keyed data structure, with $y$ as the key, and $x$ as the value. 

By storing $D(R)$ in QCRAM, we can use the ability to access the QCRAM in superposition to perform lookups in superposition.
We then call $D(R)$ a \textit{QCRAM data structure}. 

\clearpage 
\begin{lemma}\label{lem:DSlemma}
    Let $D(R)$ be a QCRAM data structure for a relation $R \subseteq X \times Y$. Then the following operations can be performed using $O(\log|X|+\log|Y|)$ time and QCRAM queries:
\begin{itemize}
    \item \textbf{\emph{Insert:}} For any $(x,y)\in (X\times Y)\setminus R$, add $(x,y)$ to $D(R)$. 
    \item \textbf{\emph{Lookup by $x$ in superposition:}} For any $x\in X$, map $$\ket{x}\mapsto \left\{\begin{array}{ll}
    \ket{x}\sum_{i=1}^{|R(x)|}\frac{1}{\sqrt{|R(x)|}}\ket{i}\ket{y_i} &\quad \mbox{if }R(x)\neq\emptyset\\
    \ket{x}\ket{\bot} &\quad \mbox{otherwise}
    \end{array}\right.$$ 
    where $\{y_1,\dots,y_{|R(x)|}\}=R(x)$.
    \item \textbf{\emph{Lookup by $y$ in superposition:}} For any $y\in Y$, map $$\ket{y}\mapsto \left\{\begin{array}{ll}
    \ket{y}\sum_{i=1}^{|R^{\transpose}(y)|}\frac{1}{\sqrt{|R^{\transpose}(y)|}}\ket{i}\ket{x_i}&\quad \mbox{if }R^{\transpose}(y)\neq\emptyset\\
    \ket{y}\ket{\bot} &\quad \mbox{otherwise}
    \end{array}\right.$$  
    where $\{x_1,\dots,x_{|R^{\transpose}(y)|}\}=R^{\transpose}(y)$.
\end{itemize}\end{lemma}

\noindent In later sections, we will abuse notation by letting $\ket{x,y}$ denote $\ket{x,i,y}$, where $i$ is the index of $y$ in the array storing $R(x)$, and sometimes we will even let $\ket{x,y,x'}$ denote $\ket{x,i,y,j,x'}$ where $i$ is as above, and $j$ is the index of $x'$ in the array storing $R^{\transpose}(y)$. This is not a problem, as all we require from the specific encoding of $\ket{x,y}$ is that we can do computations on both~$x$ and~$y$ --- it does not matter if there is superfluous information in the encoding, as long as it is not too large. 

\begin{proof}
The insertion is directly inherited from the data structure. We describe a superposition lookup of $x$ (the case for $y$ is virtually identical). The classical lookup by $x$ returns a number $N=|R(x)|$, and a pointer to an array storing $R(x) = \{y_1, \dots,y_N\}$. If $N\neq 0$, then generate $\sum_{i=1}^N\frac{1}{\sqrt{N}}\ket{i}\ket{0}$. Use a QCRAM query to access the entries of the array, to map $\ket{i}\ket{0}\mapsto \ket{i}\ket{y_i}$.
The time and memory complexities follow immediately from the properties of $D(R)$. 
\end{proof}

\subsection{Geometric properties of the unit sphere}\label{sec:sphere} 


We denote the sphere of the $d$-dimensional unit ball by $\cS^{d-1}\coloneqq\{\vx \in\R^d \colon \norm{\vx} = 1\}$. A unit vector $\vx\in\cS^{d-1}$ and angle $\alpha \in [0,\pi/2]$ define the set $\cH_{\vx,\alpha}\coloneqq\{\vc \in\cS^{d-1} \colon \innerP{\vx}{\vc} \geq \cos(\alpha)\}$, called the \textit{spherical cap} of center $\vx$ and angle $\alpha$. 
Note that $\cH_{\vx,0} = \{\vx\}$ and that $\cH_{\vx,\pi/2}$ is a hemisphere. 
The intersection of two spherical caps forms a \textit{spherical wedge}, and is denoted by $\cW_{\vx,\alpha,\vy,\beta}\coloneqq \cH_{\vx,\alpha} \cap \cH_{\vy,\beta}$ for $\vx,\vy\in\cS^{d-1}$ and $\alpha, \beta \in [0,\pi/2]$. 
If $\innerP{\vx}{\vy}=\cos(\theta)$ for $\theta \in [0, \pi/2]$, then  
\begin{equation}
\begin{split}
    \cW_{\vx,\alpha,\vy,\beta} = \begin{cases}
            \cH_{\vx, \alpha} &\quad\text{if  $\theta \leq \beta - \alpha$} \\ 
            \cH_{\vy, \beta} &\quad\text{if  $\theta \leq \alpha - \beta$} \\ 
            \emptyset &\quad\text{if $\theta > \alpha + \beta$}
        \end{cases}
        \label{eq:wedge-edge-cases}
\end{split}
\end{equation} 
and in the remaining case ($|\alpha - \beta| < \theta \leq \alpha + \beta$) the wedge is nonempty and properly contained in both $\cH_{\vx, \alpha}$ and $\cH_{\vy, \beta}$. 

The volume (or hyperarea) of a spherical cap or wedge allows us to quantify the probability that a random unit vector is close to one or two given vectors, respectively.   
Given $\vx,\vy\in\cS^{d-1}$ with $\innerP{\vx}{\vy}=\cos(\theta)$, we denote the ratio of the volume of $\cW_{\vx,\alpha,\vy,\beta}$ to that of $\cS^{d-1}$ by $\cW(\alpha,\beta\mid\theta)$, as it depends on $\vx,\vy$ only through their mutual angle $\theta$. 
This normalized wedge volume equals the probability that, for a fixed pair $(\vx,\vy)$ of unit vectors satisfying $\innerP{\vx}{\vy}=\cos(\theta)$, a uniformly random unit vector lies in the wedge $\cW_{\vx,\alpha,\vy,\beta}$.  

The normalized volume of a spherical cap is thus of the form $\cW(\alpha, \alpha \mid 0)$, and can be estimated using the following lemma from~\cite{BDGL16}. 

\begin{lemma}[Volume of a spherical cap {\cite[Lemma~2.1]{BDGL16}}]\label{lem: cap volume} 
    Let $\alpha \in (0,\pi/2)$. 
    For all $\vx \in \cS^{d-1}$, 
    \begin{equation*}
        \cW(\alpha, \alpha \mid 0) \coloneqq \Pr_{\vc \sim \cU(\cS^{d-1})}\left[\innerP{\vx}{\vc} \geq  \cos(\alpha)\right] = d^{\,\Theta(1)}  \underbrace{\left(1 - \cos^2(\alpha)\right)^{d/2}}_{\eqqcolon\, p_{\alpha}}.
    \end{equation*} 
\end{lemma} 

\noindent In~\cite{BDGL16}, an estimate for the normalized volume of a spherical wedge was also given, but it was missing the condition $|\alpha - \beta| \leq \theta \leq \alpha + \beta$, which is necessary for their formula to hold. 
Note that the other cases follow directly from~\Cref{eq:wedge-edge-cases} and~\Cref{lem: cap volume}: $\cW(\alpha, \beta \mid \theta) = 0$ if $\theta > \alpha + \beta$, and $\cW(\alpha, \beta \mid \theta) =_d \min\{p_\alpha, p_\beta\}$ if $\theta \geq |\alpha - \beta|$. (When $\theta = |\alpha - \beta|$, \Cref{lem: cap volume} and~\Cref{lem: wedge volume} coincide.) More detailed formulas for the wedge volume can be found in~\cite{techreport:LK2014}.  

\begin{lemma}[Volume of a spherical wedge (correction of~{\cite[Lemma~2.2]{BDGL16}})]\label{lem: wedge volume} 
    Let $\alpha, \beta, \theta \in (0,\pi/2)$ satisfy $|\alpha - \beta| \leq \theta \leq \alpha + \beta$.  
    For all $\vx, \vy \in \cS^{d-1}$ with $\innerP{\vx}{\vy} = \cos(\theta)$, 
    \begin{align*}
        \cW(\alpha, \beta \mid \theta) &\coloneqq \Pr_{\vc \sim \cU(\cS^{d-1})}\left[\innerP{\vx}{\vc} \geq \cos(\alpha) \text{ and } \innerP{\vy}{\vc} \geq  \cos(\beta)\right] = d^{\,\Theta(1)}  (1 - \gamma^2)^{d/2}
        \intertext{where $\gamma^2 \coloneqq \frac{\cos^2(\alpha) + \cos^2(\beta) - 2\cos(\alpha)\cos(\beta)\cos(\theta)}{\sin^2(\theta)}$. 
        In particular, if $\alpha = \beta$,}
        \cW(\alpha, \alpha \mid \theta) &\coloneqq \Pr_{\vc \sim \cU(\cS^{d-1})}\left[\innerP{\vx}{\vc} \geq \cos(\alpha) \text{ and } \innerP{\vy}{\vc} \geq  \cos(\alpha)\right] = d^{\,\Theta(1)}  \left(1 - \frac{2\cos^2(\alpha)}{1 + \cos(\theta)}\right)^{d/2}.
    \end{align*}
\end{lemma}

\noindent The previous two lemmas allow us to estimate the probability that a random unit vector $\vc$ satisfies both $\innerP{\vx}{\vc} \geq \cos(\alpha)$ and $\innerP{\vy}{\vc} \geq  \cos(\beta)$, for fixed $\vx,\vy \in \cS^{d-1}$. In our analysis, we instead deal with events where those inner products are equal to $\cos(\alpha)$ (respectively,  $\cos(\beta)$) up to an additive~$\epsapprox > 0$.   
We therefore consider the following variants of~\Cref{lem: cap volume} and~\Cref{lem: wedge volume}. These results are folklore and follow from the previous lemmas, provided that~$\epsapprox$ is chosen sufficiently small (see~\Cref{app:proofs-approx-lemmas} for the proofs). Throughout this work, we globally fix $\epsapprox = 1/(\log d)^2$. 

\begin{lemma}[Volume of a spherical cap, approximate version]\label{lem: approx cap volume} 
    Let $\alpha \in (0, \pi/2)$ satisfy $\alpha = \Omega(1)$. 
    For all $\vx \in \cS^{d-1}$, 
    \begin{equation*}
        \Pr_{\vc \sim \cU(\cS^{d-1})}\left[\innerP{\vx}{\vc} \approx_\epsapprox  \cos(\alpha)\right] =_d p_{\alpha}.
    \end{equation*}
\end{lemma}

\begin{lemma}[Volume of a spherical wedge, approximate version]\label{lem: approx wedge volume}  
    Let $\alpha, \beta, \theta \in (0, \pi/2)$ satisfy $|\alpha - \beta| + \Omega(1) \leq \theta \leq \alpha + \beta - \Omega(1)$. 
    For all $\vx, \vy \in \cS^{d-1}$ with $\innerP{\vx}{\vy} \approx_\epsapprox \cos(\theta)$, 
    \begin{equation*}
        \Pr_{\vc \sim \cU(\cS^{d-1})}\left[\innerP{\vx}{\vc} \approx_\epsapprox \cos(\alpha) \text{ and } \innerP{\vy}{\vc} \approx_\epsapprox \cos(\beta)\right] =_d \cW(\alpha, \beta \mid \theta).
    \end{equation*} 
\end{lemma}

\noindent The~$\Omega(1)$ terms in these lemmas hide a (small) absolute constant $c>0$ that is independent of~$d$. They ensure that the non-approximate lemmas remain applicable as $d \to \infty$, which we use to prove~\Cref{lem: approx cap volume} and~\Cref{lem: approx wedge volume}. 
In this work, we mainly consider normalized wedge volumes of the form $\cW(\alpha, \alpha \mid \theta)$ or $\cW(\theta, \alpha \mid \alpha)$, in which case the condition in~\Cref{lem: approx wedge volume} simplifies to $\min\{\theta, 2\alpha - \theta\} = \Omega(1)$. 

\subsection{Random product codes and their induced relations}\label{sec:RPC}

The main tasks in our quantum algorithm for finding 3-tuple solutions can be formulated as searching for pairs $(\vx,\vy)$ of unit vectors that are somewhat ``close'' to each other, in the sense that we can bound their inner product. We simplify these tasks by only searching for pairs that form an $R$-collision under carefully constructed relations of the form $R = R_{(C,\alpha)}$, where 
\begin{equation*}
    R_{(C,\alpha)} \coloneqq \{(\vx,\vc) \in \cS^{d-1} \times C \colon \innerP{\vx}{\vc} \approx_\epsapprox \cos(\alpha)\}
\end{equation*}
for a subset $C \subseteq \cS^{d-1}$ and $\alpha \in (0,\pi/2)$. 
Consequently, if $(\vx,\vy)$ form an $R_{(C,\alpha)}$-collision (i.e., $R_{(C,\alpha)}(\vx) \cap R_{(C,\alpha)}(\vy) \neq \emptyset$), then there is a point $\vc \in C$ that is close to both $\vx$ and $\vy$, implying that $\vx$ and $\vy$ are also close to each other, where closeness is quantified by the parameter $\alpha$. When the dependencies on $C$ and $\alpha$ are clear from context, we often just write $R$. 
Note also that those relations $R_{(C,\alpha)}$ are infinite objects, but we will usually restrict them to a finite subset in the first coordinate, for instance the vectors in our list~$L$, making the relation finite. 

The family of subsets $C$ that we will work with are based on random product codes~\cite{BDGL16}. 

\begin{definition}[Random product code]\label{defn:RPC}
    We define $\cC(d,b,M)$ as the family of sets $C \subseteq \cS^{d-1}$ that can be written as $$C = \mQ(C^{(1)} \times \dots \times C^{(b)})$$ 
    where $\mQ \in \mathrm{SO}(d)$, and $C^{(i)} \subseteq \frac{1}{\sqrt{b}}\cS^{d/b - 1}$ with $|C^{(i)}| = M^{1/b}$ for each $i \in \{1,\dots,b\}$. 
    Any such tuple $(\mQ, C^{(1)}, \dots, C^{(b)})$ is called a \emph{description} of $C$.

    A \emph{random product code (RPC)} is a random set $C \in \cC(d,b,M)$ obtained by sampling a uniformly random description $(\mQ, C^{(1)}, \dots, C^{(b)})$ from the set of all valid descriptions. 
    We write $\mathrm{RPC}(d,b,M)$ for the resulting distribution over $\cC(d,b,M)$. 
\end{definition} 

\noindent Random product codes~$C$ have two very useful properties. For a parameter $\alpha$ and induced relation $R = R_{(C,\alpha)}$, we have: 
\begin{enumerate}[label=(\arabic*)]
    \item \textbf{Efficient decodability:} In certain parameter regimes (in particular, if $b$ is not too small), there is an algorithm that, on input $\vx \in \cS^{d-1}$, computes the set $R(\vx)$ in time roughly equal to its size $|R(\vx)|$. 
    See~\Cref{lem:-efficient-decodability-RPCs}. Note that this algorithm gives forward superposition query access to $R$ (see \Cref{sec:relation-and-data-structure-requirement}). 
    \item \textbf{Random behavior:} In certain parameter regimes (in particular, if $b$ is not too large), a random product code $C$ behaves like a uniformly random subset of $\cS^{d-1}$ in the following sense: for all $\vx, \vy \in \cS^{d-1}$ satisfying $\innerP{\vx}{\vy} \approx_\epsapprox \cos(\theta)$, the probability that there exists $\vc \in C$ satisfying $\innerP{\vx}{\vc} \approx_\epsapprox \cos(\alpha)$ and $\innerP{\vy}{\vc} \approx_\epsapprox \cos(\alpha)$ (meaning that $(\vx,\vy)$ is an $R$-collision) is the same,  up to subexponential factors, as in the case that each element of $C$ was independently sampled from $\cU(\cS^{d-1})$. See~\Cref{lem:random-behavior-RPCs}.
\end{enumerate}
In other words, RPCs give us sufficiently good random behavior, while still allowing for efficient decodability, which is the main reason we work with RPCs instead of uniformly random subsets. 

The next two lemmas\footnote{The proofs in~\cite{BDGL16} consider $R_{(C,\alpha)}(\vx)$ defined as $\{\vc \in C \colon \innerP{\vx}{\vc} \geq \cos(\alpha)\}$, but can easily be seen to work for our definition (with $\approx_\epsapprox$ instead of $\geq$).} show that it suffices to take $b = \log d$ for both properties to be satisfied, so we fix this choice of $b$ in the remainder of this paper.  

\begin{lemma}[Efficient decodability (implicit in {\cite[Lemma 5.1]{BDGL16}})]\label{lem:-efficient-decodability-RPCs}  
    There exists a classical algorithm that, given a description $(\mQ, C^{(1)}, \dots, C^{(b)})$ of $C \in \cC(d,b,M)$ and a target vector $\vx \in \cS^{d-1}$, returns the set $R_{(C,\alpha)}(\vx) \coloneqq \{\vc \in C \colon \innerP{\vx}{\vc} \approx_\epsapprox \cos(\alpha)\}$ in time $O(d^2 M^{1/b} + dM^{1/b} \log M +  bd  |R_{(C,\alpha)}(\vx)|)$.  
    In particular, if $b = \omega(1)$ and $M = 2^{O(d)}$, then the runtime is $O(bd  |R_{(C,\alpha)}(\vx)|) + 2^{o(d)}$.  
\end{lemma}

\noindent Next, \Cref{lem:random-behavior-RPCs} provides sufficiently tight bounds on the probability that two unit vectors form a collision under the relation $R_{(C,\alpha)}$ induced by a sample $C \sim \mathrm{RPC}(d,b,M)$. Specifically, it identifies parameter regimes where RPCs behave similarly to uniformly random subsets with respect to collision probabilities, as $\Pr_{C \sim \cU(\cS^{d-1},M)}[R_{(C,\alpha)}(\vx) \cap R_{(C,\beta)}(\vy) \neq \emptyset] = \Theta(\min\{1, M \cW(\alpha, \beta \mid \theta)\})$.

\begin{lemma}[Random behavior {\cite[Theorem 5.1]{BDGL16}}]\label{lem:random-behavior-RPCs}  
    Let $\alpha, \beta \in (0, \pi/2)$ and $\theta \in [0, \pi/2)$ satisfy $|\alpha - \beta| + \Omega(1) \leq \theta \leq \alpha + \beta - \Omega(1)$ if $\theta \neq 0$, and $\alpha = \beta = \Omega(1)$ otherwise.  
    Let $b = O(\log d)$ and let $M$ be such that $M \cdot \cW(\alpha, \beta \mid \theta) \geq 2^{-O(d)}$. 
    For all $\vx, \vy \in \cS^{d-1}$ with $\innerP{\vx}{\vy} \approx_\epsapprox \cos(\theta)$, 
    \begin{align*}
        \Pr_{C \sim \mathrm{RPC}(d,b,M)}\left[R_{(C,\alpha)}(\vx) \cap R_{(C,\beta)}(\vy) \neq \emptyset\right] =_d \min\{1, M \cdot \cW(\alpha, \beta \mid \theta)\}.
    \end{align*} 
\end{lemma}

\noindent Note that the case $\theta = 0$ deals with the probability that $R_{(C,\alpha)}(\vx)$ is nonempty (meaning there exists $\vc \in C$ that is ``close'' to $\vx = \vy$).

\section{Quantum algorithm for finding many 3-tuple solutions}\label{sec:main-quantum-algorithm}

In this section, we present a quantum algorithm for~\Cref{prob:finding-many-3-tuples} from the introduction: given a list $L$ of $m$ i.i.d.\ uniform samples from $\cS^{d-1}$, this algorithm returns $m$ triples $(\vx, \vy, \vz) \in L^3$ satisfying $\norm{\vx - \vy - \vz} \leq 1$. 
We are specifically interested in instances with $m = (27/16)^{d/4 +o(d)}$. For sufficiently large ${o(d)}$, this is the minimal list size to ensure that with high probability over the choice of $L$ there exist $m$ 3-tuple solutions~\cite[Theorem~3]{HK2017}, and hence corresponds to the minimal memory regime of~\Cref{prob:finding-many-3-tuples}. 
Our work is motivated by the observation that, for a list size $m$ that is slightly larger, but still asymptotically $m = (27/16)^{d/4 +o(d)}$, there exist in fact $m$ 3-tuple solutions $(\vx,\vy,\vz) \in L^3$ for which $\innerP{\vx}{\vy}$ is essentially $1/3$ and $\innerP{\vx-\vy}{\vz}$ is essentially $2/3$. This allows us to reduce our search problem to (a less simplified version of)~\Cref{prob:finding-many-3-tuples-relaxed} from the introduction.
More precisely, as proven in~\Cref{lem:size-of-Tsol-application}, there exist $\theta,\theta' \in (0,\pi/2)$ such that  
\begin{equation} 
    \cT_{\mathrm{sol}}(L, \theta,\theta') \coloneqq \left\{(\vx, \vy, \vz)  \in  L^3 \colon \innerP{\vx}{\vy} \approx_\epsapprox \cos(\theta),  \innerP{\tfrac{\vx-\vy}{\norm{\vx-\vy}}}{\vz} \approx_\epsapprox \cos(\theta')\right\} \label{eq:T-sol}
\end{equation} 
consists only of 3-tuple solutions and (with high probability over the choice of $L$) has size at least $m$.
We therefore design a quantum algorithm that finds $m$ elements of $\cT_{\mathrm{sol}}(L, \theta,\theta')$, thereby solving~\Cref{prob:finding-many-3-tuples}. 
As we present the algorithm for a fixed choice\footnote{For~\Cref{prob:finding-many-3-tuples}, choosing $\cos(\theta) = 1/3$ and $\cos(\theta') = \epsapprox + \sqrt{1/3 + \epsapprox/2}$ is optimal in the sense that it yields the smallest possible list size $|L|$ for which $\cT_{\mathrm{sol}}(L, \theta,\theta')$ consists of $|L|$ 3-tuple solutions. However, our quantum algorithm may also be of interest in settings where other choices are more suitable, such as the generalization of~\Cref{prob:finding-many-3-tuples} that searches for triples satisfying $\norm{\vx - \vy - \vz} \leq t$ for $t \neq 1$ (e.g., see~\cite{HK2017}). We therefore present most results for a larger range of $\theta,\theta'\in (0, \pi/2)$.} of $(\theta,\theta')$, we usually write $\cT_{\mathrm{sol}}$ as shorthand for $\cT_{\mathrm{sol}}(L, \theta,\theta')$. 

\begin{remark}\label{rem:elems-Tsol-are-distinct}
    There is a slight subtlety in the definition of $\cT_{\mathrm{sol}}(L, \theta,\theta')$: it does not explicitly require each tuple $(\vx,\vy,\vz) \in \cT_{\mathrm{sol}}(L, \theta,\theta')$ to consist of distinct vectors, which is required by~\Cref{prob:finding-many-3-tuples}. (Note that the latter problem is trivial otherwise: for instance, consider the $|L|$ 3-tuples $(\vx,\vx,\vx)$ for $\vx \in L$.) 
    While we could explicitly impose that $\vx,\vy,\vz$ be distinct in~\Cref{eq:T-sol}, this property is already implied by~\Cref{eq:T-sol} in our setting of interest. Namely, if $\theta = \Omega(1)$ and $\cos(\theta') = \epsapprox + \sqrt{(1 - \cos(\theta) + \epsapprox)/2}$ (as motivated by~\Cref{lem:size-of-Tsol-application}), then each $(\vx,\vy,\vz) \in \cT_{\mathrm{sol}}(L, \theta,\theta')$ that is not composed of distinct vectors must have $\innerP{\vx}{\vy} = \cos(\theta) - \epsapprox$. As the latter event has measure zero for $\vy \sim \cU(\cS^{d-1})$, we may, without loss of generality, assume that in the setting of interest each element of $\cT_{\mathrm{sol}}(L, \theta,\theta')$ consists of distinct vectors. 
\end{remark}

\subsection{High-level overview of our quantum algorithm and main result}\label{sec:overview-algo}

Our quantum algorithm consists of several steps of amplitude amplification (\Cref{lem:AA-without-knowing-norm}), carefully nested together, and preceded by preprocessing the list~$L$ into a useful data structure. We describe the high-level ideas here, and defer the details to the following subsections. 
As outlined in the introduction, a natural strategy for finding an element of $\cT_{\mathrm{sol}}$ is to search for a pair $(\vx,\vy) \in L^2$ satisfying $\innerP{\vx}{\vy} \approx_\epsapprox \cos(\theta)$, and then attempt to find $\vz \in L$ such that $(\vx,\vy,\vz) \in \cT_{\mathrm{sol}}$. 
This approach can be viewed as a two-oracle search (\Cref{sec:prelim-two-oracle-search}) with sets $\cT_2 \subseteq \cT_1 \subseteq \cT_0$ defined by  
\begin{align*} 
    \cT_0 &\coloneqq L^2, \\
    \cT_1 &\coloneqq \{(\vx, \vy)  \in \cT_0 \colon \innerP{\vx}{\vy} \approx_\epsapprox \cos(\theta)\}, \\
    \cT_2 &\coloneqq \{(\vx, \vy)  \in \cT_1 \colon \exists \vz \in L, \innerP{\tfrac{\vx-\vy}{\norm{\vx-\vy}}}{\vz} \approx_\epsapprox \cos(\theta')\}.
\end{align*}   
Whereas checking membership in $\cT_1$ has negligible time complexity ${\sf C}_1 = 2^{o(d)}$, checking membership in $\cT_2$ is a more involved search problem with nontrivial time complexity ${\sf C}_2$. For instance, using amplitude amplification we obtain ${\sf C}_2 = \sqrt{|L|} \,  2^{o(d)}$ when $|L| = 2^{O(d)}$. 
Assuming that, given $(\vx,\vy) \in \cT_2$, we can also \textit{find} $\vz$ such that $(\vx,\vy,\vz) \in \cT_{\mathrm{sol}}$ in time ${\sf C}_2$ (which is true for amplitude amplification), this results in a quantum algorithm for finding an element of $\cT_{\mathrm{sol}}$ in time 
\[\sqrt{\frac{\varepsilon_1}{\varepsilon_2}} \left(\frac{1}{\sqrt{\varepsilon_1}}{\sf C}_1+{\sf C}_2\right) 2^{o(d)}\]
where $\varepsilon_1$ and $\varepsilon_2$ are the probabilities that an element sampled uniformly at random from $\cT_0$ lies in $\cT_1$ and $\cT_2$, respectively. 
Repeating the algorithm about $|L|$ times then hopefully solves~\Cref{prob:finding-many-3-tuples}. Unfortunately, this naive strategy is too expensive.\footnote{One can show  $\varepsilon_1 =_d p_{\theta}$ and $\varepsilon_2 =_d  p_{\theta} \min\{1, |L| p_{\theta'}\}$. For $|L| = (27/16)^{d/4 +o(d)}$, taking $\cos(\theta) \approx 1/3$ and $\cos(\theta') \approx 1/\sqrt{3}$ is necessary and sufficient to ensure $\cT_{\mathrm{sol}}(L,\theta,\theta')$ has size $\geq |L|$ with high probability (see~\Cref{lem:size-of-Tsol-application}). This implies that the naive approach (using amplitude amplification to check membership in $\cT_2$) takes time at least $2^{0.3350d + o(d)}$ to find $|L|$ solutions. Moreover, it yields $\frac{1}{\sqrt{\varepsilon_1}}{\sf C}_1 \ll {\sf C}_2$, suggesting that this approach is suboptimal.} 

Inspired by the locality-sensitive filtering technique, which has been used in state-of-the-art lattice sieving algorithms since \cite{BDGL16}, our key idea to improve this naive strategy is to search more \textit{locally}. The sets $\cT_0, \cT_1, \cT_2$ will be replaced by random subsets of them that only include those pairs of vectors that lie in the same ``local'' region of $\cS^{d-1}$ (formally defined as forming a collision under some suitable relation), meaning that these vectors are not too far apart. 
While those subsets may not cover all of $\cT_{\mathrm{sol}}$, their randomness will ensure that repeating this approach sufficiently many times for different random subsets allows us to find all elements of $\cT_{\mathrm{sol}}$.
Altogether, this modified strategy results in a better trade-off between the different cost components of two-oracle search. 

Specifically, we restrict $\cT_0 = L^2$ to those pairs of vectors in $L$ that are both ``close'' to some vector in a fixed subset $C \subseteq \cS^{d-1}$, where closeness is measured by a parameter~$\alpha$. Letting $R$ be the relation on $\cS^{d-1} \times C$ including exactly those pairs $(\vx,\vc)$ that satisfy $\innerP{\vx}{\vc}\approx_\epsapprox\cos(\alpha)$, we define 
\begin{equation*}
    \cT_0(R) \coloneqq \{(\vx,\vy) \in \cT_0 \colon R(\vx) \cap R(\vy) \neq \emptyset\}
\end{equation*}
as the set of $R$-collisions in $\cT_0 = L^2$. 
For each $(\vx,\vy) \in \cT_0(R)$, we are now guaranteed that both $\vx$ and $\vy$ are close to some $\vc \in C$, so they are also somewhat close to each other, and have a higher chance of satisfying $\innerP{\vx}{\vy} \approx_\epsapprox \cos(\theta)$. We will therefore replace the role of $\cT_1$ by a suitable subset 
\begin{equation*}
    \cT_1(R) \subseteq  \{(\vx, \vy)  \in \cT_0(R) \colon \innerP{\vx}{\vy} \approx_\epsapprox \cos(\theta) \}
\end{equation*} 
consisting of $R$-collisions in $\cT_1$. The probability $\varepsilon'_1$ that an element sampled from $\cT_0(R)$ is in $\cT_1(R)$ could now be larger than $\varepsilon_1$, while the cost ${\sf C}'_1$ of checking membership in $\cT_1(R)$ remains negligible, so we seem to have reduced one component of the two-oracle search cost.

However, there is a caveat: in order to project onto $\cT_1(R)$ using amplitude amplification, we need a unitary that creates a superposition over this restricted subset $\cT_0(R)$ of $L^2$, and it is not immediately clear that finding $R$-collisions in $L^2$ is easy. This will be resolved by adding a \textit{preprocessing} phase during which the algorithm prepares a data structure in QCRAM that stores the finite relation $R_L \coloneqq R \vert_{L \times C}$, allowing us to efficiently construct a superposition over the elements of $\cT_0(R)$ at any later stage of the algorithm. 
By taking $C$ to be a random product code (RPC, \Cref{defn:RPC}), we can prepare such a data structure at reasonable cost. 

Next, given a superposition over $\cT_1(R)$, the goal is to check whether there exists $\vz\in L$ satisfying $\innerP{\tfrac{\vx-\vy}{\norm{\vx-\vy}}}{\vz}\approx_\epsapprox\cos(\theta')$ for a given pair $(\vx,\vy) \in \cT_1(R)$. 
Again, we will restrict our search to achieve this more efficiently than performing amplitude amplification over $L$: given $(\vx,\vy) \in \cT_1(R)$, we only consider those $\vz \in L$ that collide with the normalization of $\vx-\vy$ under the relation $R' \coloneqq \{(\vx, \vc) \in \cS^{d-1} \times C' \colon \innerP{\vx}{\vc}\approx_\epsapprox\cos(\alpha')\}$ defined by another subset $C' \subseteq \cS^{d-1}$ and parameter $\alpha'$. 
That is, the search is restricted to a subset 
\begin{equation*}
    \cT_2(R,R') \subseteq \{(\vx,\vy) \in \cT_1(R) \colon \exists \vz \in L, \innerP{\tfrac{\vx-\vy}{\norm{\vx-\vy}}}{\vz} \approx_\epsapprox \cos(\theta'), R'(\tfrac{\vx-\vy}{\norm{\vx-\vy}}) \cap R'(\vz) \neq \emptyset\}. 
\end{equation*}
This ``local'' search is again facilitated by letting $C'$ be a random product code and by preparing a data structure for $R'_L \coloneqq R' \vert_{L \times C'}$. 
If restricting to $R'$-collisions reduces the search for $\vz$ to a much smaller subset of $L$, then the cost ${\sf C}'_2$ of checking membership in $\cT_2(R,R')$ could be significantly less than the cost ${\sf C}_2$ in the naive approach, possibly resulting in an improved overall time complexity.

Indeed, this \textit{local} two-oracle search algorithm finds elements of $\cT_2(R,R') \subseteq \cT_2$ more efficiently than the naive strategy. 
By sampling the subsets $C, C'$ randomly, using the RPC distribution from~\Cref{defn:RPC}, we can ensure that this approach finds a sufficiently random subset of $\cT_{\mathrm{sol}}$, so repeating the local two-oracle search algorithm for sufficiently many random pairs $(C,C')$ allows us to find $|L|$ elements of $\cT_{\mathrm{sol}}$, thereby solving~\Cref{prob:finding-many-3-tuples}. 

We summarize the resulting quantum algorithm, called ${\tt 3List}$, in~\Cref{algo:main-3List}.\footnote{The algorithm {\tt SolutionSearch} in the \textbf{Search} phase is essentially the aforementioned local two-oracle search algorithm, and samples from $\cT_2(R,R')$ for fixed $(R,R')$. By choosing the parameter $\searchthreshold$ carefully, the \textbf{Search} phase finds all elements of $\cT_2(R,R')$ using $|\cT_2(R,R')| 2^{o(d)}$ calls to {\tt SolutionSearch}, even if $|\cT_2(R,R')|$ is not known a priori.}  
In the following sections, we explain and analyze the individual phases of {\tt 3List} in more detail, allowing us to then prove our main result, \Cref{thm:main-3List}, in~\Cref{sec:proof-of-main-thm}. 
We remark that we cannot use the result of~\cite{KLL15twoOracles} or~\cite{Amb10VariableTime} directly: as our goal is to find all elements of $\cT_2(R,R')$, we want to sample a single element with sufficient randomness so that many samples are likely to correspond to many distinct elements. We therefore give our own algorithm and analysis, for our particular setting. 

\bigskip

\begin{algorithm}[H]
\caption{${\tt 3List}$}\label{algo:main-3List}
\vskip3pt
\makebox[5.7em][l]{Input:} A list $L \subseteq \cS^{d-1}$;  \\[1.5pt] 
\makebox[5.7em][l]{ } Angles $\theta,\theta' \in (0, \pi/2)$ \\[4pt] 
\makebox[5.7em][l]{Parameters:} Angles $\alpha,\alpha' \in (0,\pi/2)$; \\[1.5pt] 
\makebox[5.7em][l]{ } Number of repetitions $\nrep$; \\[1.5pt] 
\makebox[5.7em][l]{ } Search threshold $\searchthreshold$  \\[4pt] 
\makebox[5.7em][l]{Output:} A list $L'$ consisting of 3-tuple solutions in $\cT_{\mathrm{sol}}(L,\theta,\theta')$
\vskip5pt
\hrule
\vskip5pt
\begin{enumerate}
    \item $L'\leftarrow \emptyset$
    \item Repeat for $\nrep$ times: 
    \begin{enumerate}
        \item \textbf{Sample:} Sample $C \sim \mathrm{RPC}(d,\log d,1/p_{\alpha})$ and $C' \sim \mathrm{RPC}(d,\log d,1/p_{\alpha'})$ 
        \item \textbf{Preprocess:} $(D,D')\leftarrow{\tt Preprocess}(L, C, C')$ (\Cref{algo:preprocess})
        \item \textbf{Search:} Use  $(D, D')$ to find 3-tuple solutions: 
        \begin{enumerate}
            \item $S \leftarrow \emptyset$, $\textsc{count} \leftarrow 0$ 
            \item While $\textsc{count} < (|S|+1) \searchthreshold$: 
            \begin{enumerate}
                \item $\vs \leftarrow {\tt SolutionSearch}(D,D')$ (\Cref{algo:SolutionSearch}) 
                \item If $\vs \in L^3 \setminus  S$, then add $\vs$ to $S$ and set $\textsc{count} \leftarrow 0$ 
                \item Else, set $\textsc{count} \leftarrow \textsc{count} + 1$
            \end{enumerate}
            \item Add $S$ to $L'$
        \end{enumerate}  
    \end{enumerate}
    \item Return $L'$  
\end{enumerate}
\end{algorithm}

\begin{theorem}\label{thm:main-3List} 
There exists a quantum algorithm that, with probability $1 - 2^{-\Omega(d)}$, solves~\Cref{prob:finding-many-3-tuples} with list size $m = (27/16)^{d/4 +o(d)}$ in time $2^{0.284551d + o(d)}$. 
This algorithm uses $m 2^{o(d)}$ classical memory and QCRAM bits, and $2^{o(d)}$ qubits.  
\end{theorem}

\noindent As explained in~\Cref{sec:application}, ~\Cref{thm:main-3List} implies the existence of a quantum algorithm that \textit{heuristically} solves the Shortest Vector Problem with the stated time and memory complexity.

\subsection{The Sampling phase}\label{sec:alg-sample}

The \textbf{Sampling} phase of~\Cref{algo:main-3List} samples two RPCs (\Cref{defn:RPC}) $C$ and $C'$, and stores their descriptions. This phase can be completed using $2^{o(d)}$ time and classical memory.
The sizes of $C$ and $C'$ depend on the parameters $\alpha, \alpha' \in (0,\pi/2)$, which affect the complexity of {\tt 3List} and will be chosen to optimize performance. 
The sampled RPCs $C,C'$ and angles $\alpha,\alpha'$ induce the relations
\begin{equation*}
    R \coloneqq \{(\vx,\vc) \in \cS^{d-1} \times C \colon \innerP{\vx}{\vc}\approx_\epsapprox\cos(\alpha)\} \quad
    \text{ and } \quad
    R' \coloneqq \{(\vx,\vc) \in \cS^{d-1} \times C' \colon \innerP{\vx}{\vc}\approx_\epsapprox\cos(\alpha')\}.
\end{equation*} 

\begin{remark}[On the RPC parameter setting]
The parameter setting of the distribution of $C \sim \mathrm{RPC}(d,\log d,1/p_{\alpha})$ ensures that, for a given $\vx \in \cS^{d-1}$, the set $R(\vx) = \{\vc \in C \colon \innerP{\vx}{\vc} \approx_\epsapprox \cos(\alpha)\}$ has expected size $2^{o(d)}$ (taken over the choice of $C$) by~\Cref{lem: approx cap volume}.
The set $R(\vx)$ can therefore be computed in expected time $2^{o(d)}$ by~\Cref{lem:-efficient-decodability-RPCs}. The same is true for $C'$, and those properties are utilized during the \textbf{Preprocessing} and \textbf{Search} phases of {\tt 3List}. 
\end{remark}

\noindent To simplify our general analysis of {\tt 3List}, we impose the following conditions on the relations $(R,R')$. In the proof of \Cref{lem:main-3List-numberofreps} (the core lemma underlying \Cref{thm:main-3List}), we show that  with overwhelming probability over the input list $L \sim \cU(\cS^{d-1},m)$, all relations encountered during {\tt 3List} satisfy these conditions. 
\begin{definition}[Well-balanced $(R,R')$]\label{defn:well-balanced}
    For $L, C, C' \subseteq \cS^{d-1}$, let $R \subseteq \cS^{d-1} \times C$ and $R' \subseteq \cS^{d-1} \times C'$ be relations. 
    We say that $(R,R')$ is \emph{well-balanced on $L$} if the following conditions hold: 
    \begin{enumerate}
        \item[$(i)$] $|R_L| =_d |L|$ and $|(R_L)^{\!\transpose}(\vc)| =_d \frac{|R_L|}{|C|}$ for all $\vc \in C$, where $R_L \coloneqq R \vert_{L \times C}$. 
        \item[$(ii)$] $|R'_L| =_d |L|$ and $|(R'_L)^{\!\transpose}(\vc')| =_d \frac{|R'_L|}{|C'|}$ for all $\vc' \in C'$, where $R'_L \coloneqq R' \vert_{L \times C'}$. 
        \item[$(iii)$] For all $(\vx,\vy,\vz) \in \cT_{\mathrm{sol}}(R,R')$ (\Cref{eq:defn-T-relations}), $|\{\vz' \in L \colon (\vx,\vy,\vz') \in \cT_{\mathrm{sol}}(R,R')\}| =_d 1$. 
    \end{enumerate}
\end{definition}

\subsection{The Preprocessing phase}\label{sec:alg-preprocess}

After having sampled two subsets $C,C' \subseteq  \cS^{d-1}$, we proceed to the \textbf{Preprocessing} phase, which uses~\Cref{algo:preprocess} to construct QCRAM data structures $D$ and $D'$ (as defined in~\Cref{sec:relation-and-data-structure-requirement}) for the relations induced by these subsets and the parameters $\alpha$ and $\alpha'$. 

By the end of this classical algorithm, $D$ stores the relation $R_L \subseteq L \times C$, which relates each $\vx \in L$ to all ``close'' vectors in $C$, quantified by the parameter $\alpha$. 
Similarly, $D'$ stores the relation $R'_L \subseteq L \times C'$, which relates each $\vx\in L$ to all close vectors in $C'$, this time quantified by $\alpha'$. 

By construction of $C$ and $C'$, the \textbf{Preprocessing} phase can be completed in time that is optimal up to subexponential factors. 

\begin{lemma}[Preprocessing cost]\label{lem:preprocessing-cost}
Given as input a list $L \subseteq \cS^{d-1}$ and the descriptions of sets  $C \in \cC(d,\log d,1/p_{\alpha})$ and $C' \in \cC(d,\log d,1/p_{\alpha'})$, 
${\tt Preprocess}(L, C, C')$ (\Cref{algo:preprocess}) returns the data structures $D(R_L)$ and $D(R'_L)$ using $(|R_L| + |R'_L|)2^{o(d)}$ time and classical memory. 
In particular, if $(R,R')$ is well-balanced on $L$ (\Cref{defn:well-balanced}), it uses $|L| 2^{o(d)}$ time and classical memory. 
\end{lemma}

\begin{proof} 
By \Cref{lem:-efficient-decodability-RPCs}, for each $\vx \in L$, step 2(a) and step 2(c) take time $|R_L(\vx)|  2^{o(d)}$ and $|R'_L(\vx)|  2^{o(d)}$, respectively, so the cost of step 2 is $\sum_{\vx \in L} (|R_L(\vx)| + |R'_L(\vx)|) 2^{o(d)} = (|R_L| + |R'_L|) 2^{o(d)}$. As the complexity of {\tt Preprocess} is determined by step 2, this proves the first part of the lemma.
The second part follows from conditions $(i)$ and $(ii)$ in~\Cref{defn:well-balanced}. 
\end{proof}

\begin{algorithm}[H]
\caption{${\tt Preprocess}(L, C, C')$}\label{algo:preprocess}
\vskip3pt
\makebox[3.7em][l]{Input:} A list $L \subseteq \cS^{d-1}$; \\[1.5pt]
\makebox[3.7em][l]{ } Descriptions of $C \in \cC(d,\log d,1/p_{\alpha})$ and $C' \in \cC(d,\log d,1/p_{\alpha'})$ \\[4pt]
\makebox[3.7em][l]{Output:} QCRAM data structures $D=D(R_L)$ and $D'=D(R'_L)$ for the relations
\begin{equation*}
    R_L=\{(\vx,\vc)\in L \times C:\innerP{\vx}{\vc}\approx_\epsapprox \cos(\alpha)\} \quad
    \text{ and } \quad
    R'_L=\{(\vx,\vc)\in L \times C':\innerP{\vx}{\vc}\approx_\epsapprox \cos(\alpha')\}
\end{equation*}
\vskip5pt
\hrule
\vskip5pt
\begin{enumerate}
    \item Initialize a pair of empty data structures $D$ and $D'$
    \item For each $\vx\in L$:
\begin{enumerate}
    \item Compute $R_L(\vx)$ using \Cref{lem:-efficient-decodability-RPCs}
    \item For each $\vc\in R_L(\vx)$:
    \begin{enumerate}
        \item Insert $(\vx,\vc)$ into $D$
    \end{enumerate}
    \item Compute $R'_L(\vx)$ using \Cref{lem:-efficient-decodability-RPCs}
    \item For each $\vc\in R'_L(\vx)$:
    \begin{enumerate}
        \item Insert $(\vx,\vc)$ into $D'$
    \end{enumerate}
\end{enumerate}
    \item Return $D$ and $D'$
\end{enumerate}
\end{algorithm}

\noindent While most steps in the \textbf{Search} phase of~\Cref{algo:main-3List} can be achieved using the data structures $D(R_L)$ and $D(R'_L)$ that are prepared during the \textbf{Preprocessing} phase, we actually need one more tool. Namely, we would like to be able to efficiently create a uniform superposition over 
$$
R'(\vx) = \{\vc' \in C' \colon \innerP{\vx}{\vc'}\approx_\epsapprox \cos(\alpha')\}
$$ 
for a large number of vectors $\vx \in \cS^{d-1}$ that the algorithm will encounter (in superposition). That is, we want forward superposition query access to $R'$ (recall~\Cref{sec:relation-and-data-structure-requirement}) for those vectors.   
These $\vx$ are not vectors from our list $L$ itself, but rather (normalized) differences of two vectors from $L$. As the number of those $\vx$ will be rather large ($\gg m$), it is too expensive for us to store each $R'(\vx)$ in a data structure.  
To ensure that we can create the superposition over $R'(\vx)$ in subexponential time (i.e., at negligible cost for us), we therefore consider a ``decoding'' subroutine ${\tt Dec}(C')$ that is guaranteed to run in time $2^{o(d)}$, and achieves the desired mapping for all $\vx \in \cS^{d-1}$ satisfying $|R'(\vx)| \leq 2^{d/\!\log d}$, which suffices for our purposes. 

\begin{lemma}\label{lem:R-beta-oracle}
    For $\alpha' \in (0,\pi/2)$ and $C' \in \cC(d,\log d,1/p_{\alpha'})$, define the relation
    \[
    R' \coloneqq \{(\vx,\vc')\in \cS^{d-1} \times C': \innerP{\vx}{\vc'}\approx_\epsapprox \cos(\alpha')\}.
    \]     
    There is a quantum algorithm ${\tt Dec}(C')$ that, given a description of $C'$, implements the map 
    \begin{align*}
        \widetilde{\cO}_{R'} \colon \ket{\vx}\ket{0} \mapsto 
        \begin{cases}
            \displaystyle \ket{\vx}\sum_{\vc'\in R'(\vx)}\frac{1}{\sqrt{|R'(\vx)|}}\ket{\vc'}  &\quad\text{if } |R'(\vx)| \in [1, 2^{d/\!\log d}] \\
            \displaystyle \ket{\vx}\ket{\bot}  &\quad\text{otherwise} 
        \end{cases} 
    \end{align*}
    in time $2^{o(d)}$ using an auxiliary register of $2^{o(d)}$ qubits. 
    In particular, ${\tt Dec}(C')$ correctly implements $\cO_{R'}$ for all $\vx \in \cS^{d-1}$ satisfying $|R'(\vx)| \leq 2^{d/\!\log d}$. 
\end{lemma}

\begin{proof}
    The subroutine starts by enumerating the elements of $R'(\vx)$ using~\Cref{lem:-efficient-decodability-RPCs} in an auxiliary register, but stops as soon as it has found $2^{d/\!\log d} + 1$ elements, unless it has already finished earlier. (Note that the algorithm from~\Cref{lem:-efficient-decodability-RPCs} finds elements of $R'(\vx)$ one by one.) If it finds either zero or $2^{d/\!\log d} + 1$ elements, which means $|R'(\vx)| \notin [1, 2^{d/\!\log d}]$, then it undoes the computation and maps the second register to $\ket{\bot}$. Otherwise, when $|R'(\vx)| \in [1, 2^{d/\!\log d}]$, the subroutine prepares a uniform superposition over the found elements in the second register and then uncomputes the auxiliary register. 
    Altogether this takes time $2^{o(d)}$ and uses $2^{o(d)}$ auxiliary qubits. 
    (Note that, just as in \Cref{lem:DSlemma}, $\ket{\vx,\vc'}$ will be encoded by $\ket{\vx,i,\vc'}$ for some index $i$ that depends on $\vx$, $\vc'$, and $C'$, but this is not an issue.) 
\end{proof}

\subsection{The Search phase}\label{sec:alg-search}

The \textbf{Search} phase of~\Cref{algo:main-3List} repeatedly invokes a quantum algorithm called {\tt SolutionSearch}, which will be presented in~\Cref{algo:SolutionSearch}. 
As mentioned before, this algorithm searches for elements of $\cT_{\mathrm{sol}}$ by carefully nesting two layers of amplitude amplification  (\Cref{lem:AA-without-knowing-norm}) and by searching for collisions under the relations $R,R'$ defined by the sets $C,C'\subseteq \cS^{d-1}$ obtained during the \textbf{Sampling} phase (and by the angles $\alpha, \alpha'$). This is achieved by leveraging the data structures $D(R_L)$ and $D(R'_L)$ that were prepared during the \textbf{Preprocessing} phase for the finite relations $R_L \coloneqq R \vert_{L \times C}$ and $R'_L \coloneqq R' \vert_{L \times C'}$. 

More precisely, {\tt SolutionSearch} samples 3-tuple solutions from the set  
\begin{equation}
\begin{split}
    \cT_{\mathrm{sol}}(R,R') \coloneqq \bigl\{(\vx,\vy, \vz) \in \cT_{\mathrm{sol}} \colon &~R(\vx) \cap R(\vy) \neq \emptyset, |R(\vx)| \leq 2^{d /\!\log d}, \\ 
    &~R'(\tfrac{\vx-\vy}{\norm{\vx-\vy}}) \cap R'(\vz) \neq \emptyset, |R'(\tfrac{\vx-\vy}{\norm{\vx-\vy}})| \leq 2^{d /\!\log d}\bigr\} 
    \label{eq:defn-T-relations}
\end{split}
\end{equation}
by gradually refining the search through the sets $\cT_0(R) \supseteq \cT_1(R) \supseteq \cT_2(R,R')$ defined as  
\begin{alignat}{2}
    &\cT_0(R) &&\coloneqq \bigl\{(\vx,\vy) \in L^2 \colon R(\vx) \cap R(\vy) \neq \emptyset \bigr\}, \notag \\
    &\cT_1(R) &&\coloneqq \bigl\{(\vx, \vy)  \in \cT_0(R) \colon \innerP{\vx}{\vy} \approx_\epsapprox \cos(\theta), |R(\vx)| \leq 2^{d /\!\log d} \bigr\}, \label{eq:defn-Ti-R} \\ 
    &\cT_2(R,R') &&\coloneqq \bigl\{(\vx, \vy)  \in \cT_1(R) \colon \exists \vz \in L, (\vx,\vy,\vz) \in  \cT_{\mathrm{sol}}(R,R') \vphantom{2^{d /\!\log d}}\bigr\}. \notag
\end{alignat} 
The constraints that $R(\vx)$ and $R'(\tfrac{\vx-\vy}{\norm{\vx-\vy}})$ are of size $2^{o(d)}$ are imposed for technical reasons.\footnote{In our applications, these sets will be of subexponential size for most $(\vx,\vy)$. Excluding the rare cases where this is not the case ensures that the sampling probabilities over $\cT_{\mathrm{sol}}(R,R')$ remain approximately uniform.} 

The core of {\tt SolutionSearch} is a subroutine, {\tt TupleSamp} (presented in~\Cref{algo:xyz-sampler}), which generates a superposition over a subset of $\cT_1(R) \times L$ and ``flags'' those tuples that belong to $\cT_{\mathrm{sol}}(R,R')$. 
{\tt SolutionSearch} applies amplitude amplification on top of {\tt TupleSamp} to keep only those flagged tuples, and then measures the resulting state, yielding an element of $\cT_{\mathrm{sol}}(R,R') \subseteq \cT_{\mathrm{sol}}$. 

The sampling subroutine $\tt TupleSamp$ first generates a superposition over $(\vx,\vy)\in \cT_1(R)$ by applying amplitude amplification to a subroutine $\tt RCollisionSamp$ (presented in~\Cref{algo:xy-sampler}) that uses $D(R_L)$ to efficiently generate a superposition over $\cT_0(R)$. 
$\tt TupleSamp$ then proceeds by searching for $\vz$ such that $(\vx,\vy,\vz) \in \cT_{\mathrm{sol}}$, setting a flag if such a $\vz$ is found. Using $D(R'_L)$, we save computation effort by restricting this search to those $\vz$ that form an $R'$-collision with $\tfrac{\vx-\vy}{\norm{\vx-\vy}}$. 

The structure of the main algorithm {\tt 3List} and its subroutines, including the \textbf{Search} phase described in this section, is illustrated in~\Cref{fig:alg-structure}. We now give precise definitions of the algorithms forming the \textbf{Search} phase, and analyze them in full detail.

\begin{figure}[H]
    \centering
    \vspace*{0.6cm}
    \begin{tikzpicture}
    \tikzset{
    boxed/.style={
      draw,
      thick,
      rectangle,
      rounded corners=3pt,
      fill=darkblue!4,
      inner sep=6pt,
      align=center
    }
    }
        \node[boxed] (listbox) at (0,2) {{\tt 3List}\\ \Cref{algo:main-3List}};
        \node[boxed] (sambox) at (3.6,3.7) {{\tt Sample}};
        \node[boxed] (prebox) at (4,2) {{\tt Preprocess}\\ \Cref{algo:preprocess}};
        \node[boxed] (solbox) at (4.4,0) {{\tt SolutionSearch}\\ \Cref{algo:SolutionSearch}};
        \node[boxed] (xyzbox) at (8.5,1) {{\tt TupleSamp}\\ \Cref{algo:xyz-sampler}};
        \node at (8.5,-.85) {$(\vx,\vy,\vz) \overset{?}{\in} \cT_{\mathrm{sol}}$};
        \node[boxed] (xybox) at (12.6,2) {{\tt RCollisionSamp}\\ \Cref{algo:xy-sampler}};
        \node at (12.3,.15) {$\innerP{\vx}{\vy}\overset{?}{\approx}_{\epsapprox}\cos(\theta)$};

        \draw[<-] (listbox.east) -- (prebox.west); \node at (1.65,2.25) {\footnotesize loop};
        \draw (2,2) -- (2,3.7) -- (sambox.west);
        \draw (2,2) -- (2,0) -- (solbox.west); \node at (2.4,.25) {\footnotesize loop};
        \draw[<-] (solbox.east) -- (6.75,0); \node at (6.4,.25) {\footnotesize {\tt AA}};
        \draw (xyzbox.west) -- (6.75,1) -- (6.75,-1) -- (7.25,-1);
        \draw[<-] (xyzbox.east) -- (10.45,1); \node at (10.15,1.25) {\footnotesize {\tt AA}\,};
        \draw (xybox.west) -- (10.45,2) -- (10.45,0) -- (10.95,0);
    \end{tikzpicture}
    \vspace*{0.2cm}
    \caption{Structure of the algorithm {\tt 3List}, which repeats the following. First, the \textbf{Sampling} phase produces a pair $(R,R')$ of random relations that are stored in a data structure during the \textbf{Preprocessing} phase. The algorithm then repeatedly calls {\tt SolutionSearch}, which is instructed to find an element of $\cT_{\mathrm{sol}}$ using a nested amplitude amplification ({\tt AA}). 
    Given a subroutine {\tt RCollisionSamp} that creates a superposition over $R$-collisions $(\vx,\vy)\in L^2$, the first {\tt AA} amplifies those that satisfy $\innerP{\vx}{\vy} \approx_{\epsapprox}\cos(\theta)$.  
    Next, {\tt TupleSamp} extends them to triples $(\vx,\vy,\vz)$ such that $(\tfrac{\vx-\vy}{\norm{\vx-\vy}}, \vz)$ forms an $R'$-collision (if such a $\vz$ exists), and the final ${\tt AA}$ amplifies those triples that belong to $\cT_{\mathrm{sol}}$.}\label{fig:alg-structure}
\end{figure}
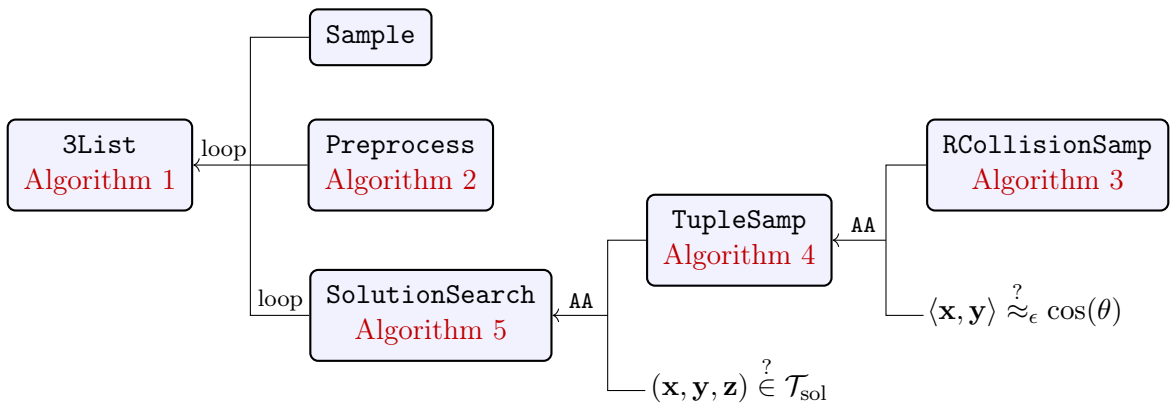

\subsubsection{Sampling an $R$-collision}
We start with the bottom layer, {\tt RCollisionSamp} (\Cref{algo:xy-sampler}). For arbitrary sets $L$ and $C$, this algorithm takes as input a QCRAM data structure $D(R)$ storing a relation $R \subseteq L \times C$ (see~\Cref{sec:relation-and-data-structure-requirement}), and outputs a superposition over $R$-collisions, specifically over $(\vx,\vc,\vy)\in L \times C \times L$ such that $\vc \in R(\vx) \cap R(\vy)$, as visualized in part~(a) of~\Cref{fig:rel-structure}. While {\tt RCollisionSamp} works for any relation, we will apply it to $L, C \subseteq \cS^{d-1}$ (where $L$ is an instance of~\Cref{prob:finding-many-3-tuples} and $C$ an RPC), so we use vector notation such as $\vx$ and $\vc$ for elements of these sets. 

{\tt RCollisionSamp} starts by taking a uniform superposition over all $\vx\in L$. Then, for all $\vx$ such that $R(\vx) \neq \emptyset$, it creates a superposition over all $\vc \in R(\vx)$ in the second register, and subsequently over all $\vy \in R^{\transpose}(\vc)$ in the third register. These steps are easy using the data structure $D(R)$.  
In case $R(\vx) = \emptyset$, the second and third register are mapped to $\ket{\bot}\ket{\bot}$, but in our applications this typically accounts for a small fraction of the state, so it is easily suppressed using amplitude amplification when~\Cref{algo:xy-sampler} is called by {\tt TupleSamp} (\Cref{algo:xyz-sampler}).

\begin{algorithm}[H]
\caption{${\tt RCollisionSamp}(D(R))$}\label{algo:xy-sampler}
\vskip3pt
\makebox[3.7em][l]{Input:} $\ket{0,0,0}$; \\[1.5pt]
\makebox[3.7em][l]{ } A QCRAM data structure $D(R)$ for $R \subseteq L \times C$, where $L$ and $C$ are finite sets \\[4pt]
\makebox[3.7em][l]{Output:} A superposition $\ket{\psi}$ over $(\vx,\vc,\vy) \in L \times (C\cup\bot)\times (L\cup\bot)$  
\vskip5pt
\hrule
\vskip5pt 
\begin{enumerate}
    \item Generate a uniform superposition over $L$ in the first register: $\ket{0,0,0}\mapsto\frac{1}{\sqrt{|L|}} \displaystyle\sum_{\substack{\vx \in L}} \ket{\vx}\ket{0}\ket{0}$.
    \item Controlled on $\vx$ in the first register, generate a uniform superposition over the set $R(\vx)$ in the second register, or map the second register to $\ket{\bot}$ if $R(\vx)$ is empty: 
    {\setlength{\abovedisplayskip}{6pt}\setlength{\belowdisplayskip}{4pt}%
    \[
    \ket{\vx}\ket{0}\ket{0}\mapsto 
    \begin{cases}
        \displaystyle \ket{\vx} \Bigg( \frac{1}{\sqrt{|R(\vx)|}} \displaystyle\sum_{\substack{\vc \in R(\vx)}} \ket{\vc}\Bigg)\ket{0} &\quad\text{if } R(\vx)\neq\emptyset \\[20pt]
        \displaystyle \ket{\vx}\ket{\bot}\ket{0} &\quad\text{if } R(\vx)=\emptyset
    \end{cases}
    \]}%
    \item Controlled on $\vc\neq\bot$ in the second register, generate a superposition over the set $R^{\transpose}(\vc)$ in the third register, and map $\ket{\bot}\ket{0}$ to $\ket{\bot} \ket{\bot}$: 
    {\setlength{\abovedisplayskip}{6pt}\setlength{\belowdisplayskip}{4pt}%
    \[
    \ket{\vx}\ket{\vc}\ket{0}\mapsto \begin{cases}
        \displaystyle \ket{\vx}\ket{\vc} \Bigg( \frac{1}{\sqrt{|R^{\transpose}(\vc)|}} \sum_{\substack{\vy \in R^{\transpose}(\vc)}} \ket{\vy} \Bigg) &\quad\text{if } \vc\neq\bot \\[20pt]
        \displaystyle \ket{\vx}\ket{\vc}\ket{\bot} &\quad\text{if } \vc=\bot
    \end{cases}
    \]}%
    \item Return the final quantum state 
\end{enumerate}
\end{algorithm}

\begin{lemma}[Analysis of {\tt RCollisionSamp}]\label{lem:easy-analysis-xy-sampler}
For finite sets $L$ and $C$, let $R \subseteq L \times C$.
Let $D(R)$ be a QCRAM data structure for $R$. ${\tt RCollisionSamp}(D(R))$ (\Cref{algo:xy-sampler}) outputs a state $\ket{\psi}$ such that, for all $(\vx,\vc,\vy)\in L \times C \times L$ satisfying $\vc \in R(\vx) \cap R(\vy)$,  we have
\[ 
\braket{\vx,\vc,\vy | \psi} = \frac{1}{\sqrt{|L| \, |R(\vx)| \, |R^{\transpose}(\vc)|}}.\]  
The algorithm uses $O(\log|L|+\log|C|)$ time and QCRAM queries, and $O(\log|L|+\log|C|)$ qubits. 
\end{lemma}

\begin{proof}
The operations used by the subroutine (taking a uniform superposition over $L$, taking a uniform superposition over $R(\vx)$ for any $\vx \in L$, and taking a uniform superposition over $R^{\transpose}(\vc)$ for any $\vc \in C$) can all be done using $O(\log|L|+\log|C|)$ time and QCRAM queries, by \Cref{lem:DSlemma}. Since~\Cref{algo:xy-sampler} uses $O(\log|L|+\log|C|)$ qubits, the claim on the time and memory complexities follows. 
The output state of~\Cref{algo:xy-sampler} is 
\begin{equation*}
    \ket{\psi} \coloneqq \frac{1}{\sqrt{|L|}} \sum_{\substack{\vx \in L \colon \\ R(\vx) \neq \emptyset}}  \frac{1}{\sqrt{|R(\vx)|}} \displaystyle\sum_{\substack{\vc \in R(\vx)}} \frac{1}{\sqrt{|R^{\transpose}(\vc)|}} \displaystyle\sum_{\vy \in R^{\transpose}(\vc)} \ket{\vx}\ket{\vc}\ket{\vy} 
    + 
    \frac{1}{\sqrt{|L|}} \sum_{\substack{\vx \in L \colon \\ R(\vx) = \emptyset}} \ket{\vx}\ket{\bot}\ket{\bot}. 
\end{equation*}
In particular, $\braket{\vx,\vc,\vy | \psi} = 1 / \sqrt{|L| \, |R(\vx)| \, |R^{\transpose}(\vc)|}$ whenever $\vc\in R(\vx)$ and $\vy\in R^{\transpose}(\vc)$. 
\end{proof}

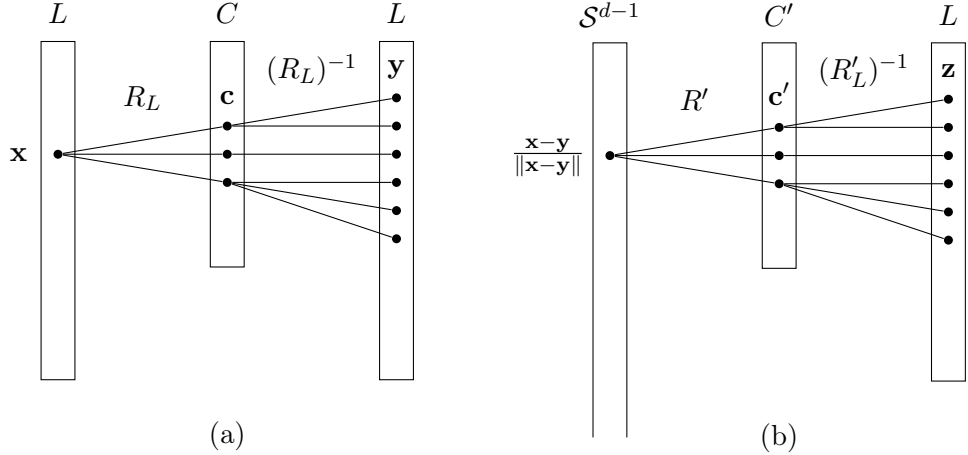
\begin{figure}[H]
\centering
    \begin{tikzpicture}
        \node at (0,0){\begin{tikzpicture}[scale = 1.5]
            \draw (0,0) rectangle (.3,3); \node at (.15,3.25) {$L$};
            \node[circle, fill, inner sep = 1.2pt] (x) at (.15,2) {}; \node at (-0.2,2) {$\vx$};

            \draw (1.5,1) rectangle (1.8,3); \node at (1.65,3.25) {$C$};
            \node[circle, fill, inner sep = 1.2pt] (c1) at (1.65,2.25) {}; \node at (1.65,2.5) {$\vc$};
            \node[circle, fill, inner sep = 1.2pt] (c2) at (1.65,2) {};
            \node[circle, fill, inner sep = 1.2pt] (c3) at (1.65,1.75) {};

            \draw (x)--(c1);
            \draw (x)--(c2);
            \draw (x)--(c3);

            \node at (.9,2.5) {$R_L$};

            \draw (3,0) rectangle (3.3,3); \node at (3.15,3.25) {$L$};
            \node[circle, fill, inner sep = 1.2pt] (y1) at (3.15,2.5) {}; \node at (3.15,2.75) {$\vy$};
            \node[circle, fill, inner sep = 1.2pt] (y2) at (3.15,2.25) {};
            \node[circle, fill, inner sep = 1.2pt] (y3) at (3.15,2) {};
            \node[circle, fill, inner sep = 1.2pt] (y4) at (3.15,1.75) {};
            \node[circle, fill, inner sep = 1.2pt] (y5) at (3.15,1.5) {};
            \node[circle, fill, inner sep = 1.2pt] (y6) at (3.15,1.25) {};

            \draw (c1)--(y1);
            \draw (c1)--(y2);
            \draw (c2)--(y3);
            \draw (c3)--(y4);
            \draw (c3)--(y5);
            \draw (c3)--(y6);

            \node at (2.4,2.75) {$(R_L)^{\!\transpose}$}; 

        \node at (1.65,-.5) {(a)};    
        \end{tikzpicture}};
        
        \node at (8,0){\begin{tikzpicture}[scale = 1.5]
            \draw (0,-0.5) -- (0,3) -- (.3,3) -- (.3,-0.5); \node at (.15,3.25) {$\cS^{d-1}$};
            \node[circle, fill, inner sep = 1.2pt] (x) at (.15,2) {}; \node at (-.4,2) {$\tfrac{\vx-\vy}{\norm{\vx-\vy}}$};

            \draw (1.5,1) rectangle (1.8,3); \node at (1.65,3.25) {$C'$};
            \node[circle, fill, inner sep = 1.2pt] (c1) at (1.65,2.25) {}; \node at (1.65,2.55) {$\vc'$};
            \node[circle, fill, inner sep = 1.2pt] (c2) at (1.65,2) {};
            \node[circle, fill, inner sep = 1.2pt] (c3) at (1.65,1.75) {};

            \draw (x)--(c1);
            \draw (x)--(c2);
            \draw (x)--(c3);

            \node at (.9,2.5) {$R'$};

            \draw (3,0) rectangle (3.3,3); \node at (3.15,3.25) {$L$};
            \node[circle, fill, inner sep = 1.2pt] (y1) at (3.15,2.5) {}; \node at (3.15,2.75) {$\vz$};
            \node[circle, fill, inner sep = 1.2pt] (y2) at (3.15,2.25) {};
            \node[circle, fill, inner sep = 1.2pt] (y3) at (3.15,2) {};
            \node[circle, fill, inner sep = 1.2pt] (y4) at (3.15,1.75) {};
            \node[circle, fill, inner sep = 1.2pt] (y5) at (3.15,1.5) {};
            \node[circle, fill, inner sep = 1.2pt] (y6) at (3.15,1.25) {};

            \draw (c1)--(y1);
            \draw (c1)--(y2);
            \draw (c2)--(y3);
            \draw (c3)--(y4);
            \draw (c3)--(y5);
            \draw (c3)--(y6);

            \node at (2.4,2.75) {$(R'_L)^{\!\transpose}$};

            \node at (1.65,-.5) {(b)};
        \end{tikzpicture}};
    \end{tikzpicture}
    \caption{Summary of the relations between vectors encountered during the \textbf{Search} phase. 
    Part~(a) visualizes the relations during the subroutine {\tt RCollisionSamp}, which first creates a superposition over all $\vx\in L$, followed by taking, for each such $\vx$, a superposition over all $\vc$ such that $(\vx,\vc) \in R_L$, and then, for each such $\vc$, over all $\vy$ such that $(\vy,\vc) \in R_L$. 
    This results in a superposition over all $R_L$-collisions (that is, all $R$-collisions in $L^2$). 
    Part~(b) visualizes what happens after the first step of {\tt TupleSamp}, which amplifies those $R_L$-collisions $(\vx,\vy)$ that satisfy $\innerP{\vx}{\vy} \approx_\epsapprox \cos(\theta)$. Namely, the second step creates, for any such $(\vx,\vy)$, a superposition over $\vc' \in C'$ satisfying $(\tfrac{\vx-\vy}{\norm{\vx-\vy}}, \vc') \in R'$, and, for each such $\vc'$, the third step creates a superposition over all $\vz$ such that $(\vz,\vc') \in R_L'$, as visualized in the figure, and amplifies those $\vz$ satisfying $\innerP{\tfrac{\vx-\vy}{\norm{\vx-\vy}}}{\vz}\approx_\epsapprox\cos(\theta')$. 
    As for most $(\vx,\vy)$ no such $\vz$ exists, {\tt SolutionSearch} applies amplitude amplification on top of {\tt TupleSamp} to amplify exactly those $(\vx,\vy)$ where such a $\vz$ does exist.}\label{fig:rel-structure} 
\end{figure}

\subsubsection{TupleSamp}

The next layer is the subroutine {\tt TupleSamp}, presented in~\Cref{algo:xyz-sampler}, which considers relations on $\cS^{d-1}$. Its goal is to construct a quantum state $\ket{\psi'}$ that has sufficiently large overlap with $\cT_{\mathrm{sol}}(R,R')$ (\Cref{eq:defn-T-relations}) so that putting amplitude amplification on top (as will be done by {\tt SolutionSearch}) allows to sample from $\cT_{\mathrm{sol}}(R,R')$ at not too high cost. In fact, for our applications, we want that \textit{each} element in $\cT_{\mathrm{sol}}(R,R')$ has rather large overlap with $\ket{\psi'}$, ensuring that repeatedly calling {\tt SolutionSearch} allows for finding \emph{all} elements of $\cT_{\mathrm{sol}}(R,R')$, and not just the same element over and over again.

{\tt TupleSamp} takes as input two relations $R \subseteq \cS^{d-1} \times C$ and $R' \subseteq \cS^{d-1} \times C'$, or rather their restrictions to $L$, given as data structures $D(R_L)$ and $D(R'_L)$. 
It also assumes the existence of an oracle that creates a uniform superposition over $R'(\tfrac{\vx-\vy}{\norm{\vx-\vy}})$ for vectors $(\vx,\vy) \in L^2$, which we will implement using~\Cref{lem:R-beta-oracle}. 
{\tt TupleSamp} first creates a superposition over $R$-collisions $(\vx,\vy)\in L^2$ that belong to $\cT_1(R)$ (implying that $\vx$ and $\vy$ are relatively ``close''), and then searches for a ``close'' $\vz \in L$ among those that form an $R'$-collision with $\tfrac{\vx-\vy}{\norm{\vx-\vy}}$. Such a $\vz$ might not exist for all $(\vx,\vy)$, so we use a flag qubit that is set to $\ket{1}$ whenever such a $\vz$ was found. 

In particular, step 1 of {\tt TupleSamp} creates a superposition over $\cT_1(R)$ by applying amplitude amplification to {\tt RCollisionSamp} (\Cref{algo:xy-sampler}), amplifying those $R$-collisions $(\vx,\vy) \in L^2$ that satisfy $\innerP{\vx}{\vy}\approx_\epsapprox\cos(\theta)$ and $|R(\vx)| \leq 2^{d/\!\log d}$.\footnote{Note that the amplification step also checks whether $\vy \neq \bot$, because the superposition constructed by {\tt RCollisionSamp} may have nonzero amplitude on pairs $(\vx,\vy)$ with $\vy = \bot$, namely if $R(\vx) = \emptyset$.}

\begin{algorithm}[tp]
\caption{${\tt TupleSamp}(D(R_L),D(R'_L))$}\label{algo:xyz-sampler} 
\vskip3pt
\makebox[3.7em][l]{Input:} $\ket{0,0,0,0,0}\ket{0,0}$, where the last two qubits form the flag register; \\[1.5pt] 
\makebox[3.7em][l]{ } QCRAM data structures $D(R_L)$ and $D(R'_L)$ for the relations 
\begin{equation*}
    R_L \coloneqq R \vert_{L \times C} \quad \text{ and } \quad R'_L \coloneqq R' \vert_{L \times C'}
\end{equation*} 
\makebox[3.7em][l]{ } where $L, C, C' \subseteq \cS^{d-1}$ are finite, $R \subseteq \cS^{d-1} \times C$, and $R' \subseteq \cS^{d-1} \times C'$ \\[4pt] 
\makebox[3.7em][l]{Output:} A superposition $\ket{\psi'}$ over $(\vx,\vc,\vy,\vc',\vz,b^{\phantom{\prime}}_F,b^{\prime}_F)$ such that all basis states in its support \\
\makebox[3.7em][l]{ } with $b^{\prime}_F = 1$ satisfy $(\vx,\vy,\vz) \in \cT_{\mathrm{sol}}(R,R')$
\vskip5pt
\hrule
\vskip5pt
\begin{enumerate}
    \item Apply ${\tt AA}_r({\tt RCollisionSamp}(D(R_L)),{\tt RCollisionCheck})$ to the first three registers and first flag qubit, where:  
    \begin{enumerate} 
        \item[--] $r = \max\{1, \sqrt{|\cT_0(R)|}\} \, 2^{o(d)}$ (computed using $D(R_L)$) 
        \item[--] {\tt RCollisionSamp} (\Cref{algo:xy-sampler}) generates a superposition $\ket{\psi}$ over $(\vx,\vc,\vy)$ such that either $\vc \in R(\vx) \cap R(\vy)$ or $\vy = \bot$
        \item[--] {\tt RCollisionCheck} checks if $\vy\neq\bot$, $\innerP{\vx}{\vy} \approx_\epsapprox \cos(\theta)$, and $|R(\vx)| \leq 2^{d/\!\log d}$ 
    \end{enumerate} 

    {\it This step sets the first flag qubit to $\ket{1}$ if and only if $\cT_1(R) \neq \emptyset$.} 
    
    \item Controlled on the first three registers being in state $\ket{\vx,\vc,\vy}$ and the first flag qubit in $\ket{1}$, apply the following unitary map to the fourth register: 
    \begin{align*}
        \ket{0} \mapsto 
        \begin{cases} 
        \displaystyle \frac{1}{\sqrt{|R'(\tfrac{\vx-\vy}{\norm{\vx-\vy}})|}} \displaystyle\sum_{\vc'\in R'(\tfrac{\vx-\vy}{\norm{\vx-\vy}})} \ket{\vc'} &\quad\text{if } |R'(\tfrac{\vx-\vy}{\norm{\vx-\vy}})| \in [1, 2^{d/\!\log d}] \\[20pt] 
        \displaystyle \ket{\bot}  &\quad\text{otherwise}
        \end{cases}
    \end{align*} 
    
    \item Controlled on the first four registers being in state $\ket{\vx,\vc,\vy,\vc'}$ and the first flag qubit in $\ket{1}$, apply ${\tt AA}_{r'}({\tt zSamp}(\vc'),{\tt zCheck}(\vx,\vy))$ to the fifth register and second flag qubit,  
    where:  
    \begin{enumerate} 
        \item[--] $r' \coloneqq \sqrt{|(R'_L)^{\!\transpose}(\vc')|}$ (computed using $D(R'_L)$)
        \item[--] {\tt zSamp}$(\vc')$ uses $D(R'_L)$ to apply the following map to the fifth register: 
        \begin{align*} 
            \ket{0}\mapsto 
            \begin{cases}
            \displaystyle \frac{1}{\sqrt{|(R'_L)^{\!\transpose}(\vc')|}} \displaystyle\sum_{\vz \in (R'_L)^{\!\transpose}(\vc')} \ket{\vz}  &\quad\text{if } \vc'\neq \bot \text{ and } (R'_L)^{\!\transpose}(\vc') \neq \emptyset \\[20pt]
            \displaystyle \ket{\bot} &\quad\text{otherwise} 
            \end{cases}
        \end{align*}
        \item[--] {\tt zCheck}$(\vx,\vy)$ checks if $\vz \neq \bot$ and $\innerP{\tfrac{\vx-\vy}{\norm{\vx-\vy}}}{\vz}\approx_\epsapprox\cos(\theta')$ 
    \end{enumerate}

    {\it This step ensures the second flag qubit has support on $\ket{1}$ if and only if $\cT_{\mathrm{sol}}(R,R') \neq \emptyset$.}
    \item Return the final quantum state 
\end{enumerate}
\end{algorithm}

\begin{lemma}[Analysis of {\tt TupleSamp}]\label{lem:analysis-xyz-sampler}
    For $L, C, C' \subseteq \cS^{d-1}$ of size $2^{O(d)}$, let $R \subseteq \cS^{d-1} \times C$ and $R' \subseteq \cS^{d-1} \times C'$ be such that $(R,R')$ is well-balanced on $L$ (\Cref{defn:well-balanced}). 
    Let $D(R_L)$ and $D(R'_L)$ be QCRAM data structures for $R_L \coloneqq R \vert_{L \times C}$ and $R'_L \coloneqq R' \vert_{L \times C'}$.
    With probability at least $1 - 2^{-\omega(d)}$, ${\tt TupleSamp}(D(R_L),D(R'_L))$ (\Cref{algo:xyz-sampler}) generates a quantum state $\ket{\psi'}$ and satisfies: 
    \begin{enumerate}
        \item[(1)] \text{Complexity:} If $\cT_1(R)\neq \emptyset$, then it uses  
        \begin{align*} 
            \left(\sqrt{\frac{|\cT_0(R)|}{|\cT_1(R)|}} + \sqrt{\frac{|L|}{|C'|}} \right) 2^{o(d)} 
        \end{align*}
        time and QCRAM queries, and  $\max\{1, \sqrt{|\cT_0(R)|}\} \, 2^{o(d)}$ otherwise. It uses $2^{o(d)}$ qubits.    
        \item[(2)] \text{Near-uniform sampling:} The probability that measuring $\ket{\psi'}$ yields outcome $(\vx,\vy,\vz, b^{\prime}_F)$ with $b^{\prime}_F = 1$ satisfies
        \begin{align*}
        \Pr[(\vx,\vy,\vz,1)] =_d 
        \begin{cases}
            \displaystyle \frac{1}{|\cT_1(R)|} &\quad\text{if } (\vx,\vy,\vz) \in \cT_{\mathrm{sol}}(R,R') \\[12pt]
            \displaystyle 0 &\quad\text{otherwise.}
        \end{cases} 
        \end{align*}
    \end{enumerate}
\end{lemma}

\begin{proof} 
Suppose $(R,R')$ is well-balanced on $L$.  
We will show that there exists an implementation of ${\tt TupleSamp}(D(R_L),D(R'_L))$ (\Cref{algo:xyz-sampler}) with the desired properties. 

Let $\Pi$ be the orthogonal projector onto the subspace marked by {\tt RCollisionCheck}. 
\Cref{lem:easy-analysis-xy-sampler} implies that ${\tt RCollisionSamp}(D(R_L))$ has complexity $2^{o(d)}$ and generates a state $\ket{\psi}$ such that 
\begin{align*}
    \norm{\Pi \ket{\psi}}^2 
    &= \sum_{\substack{(\vx,\vy)\in \cT_1(R), \\ \vc\in R(\vx) \cap R(\vy)}} \frac{1}{|L| \, |R(\vx)| \, |(R_L)^{\!\transpose}(\vc)|}  =_d \frac{|\cT_1(R)|}{|\cT_0(R)|}
\end{align*} 
if $\cT_1(R) \neq \emptyset$, and $\norm{\Pi \ket{\psi}} = 0$ otherwise. 
Here, we use that $|R(\vx)| \in [1,2^{o(d)}]$ for all $(\vx,\vy)\in \cT_1(R)$, and that condition $(i)$ of~\Cref{defn:well-balanced} implies $|L| \, |(R_L)^{\!\transpose}(\vc)| =_d |L| \, |R_L| / |C| =_d |\cT_0(R)|$ for all $\vc \in C$. 
Moreover, all $r =_d \max\{1, |L|/\sqrt{|C|}\}$ satisfy $r =_d \max\{1, \sqrt{|\cT_0(R)|}\}$. Choosing $r = \max\{1, |L|/\sqrt{|C|}\} \, 2^{o(d)}$ sufficiently large (note that $|L|$ and $|C|$ are known) ensures $\norm{\Pi \ket{\psi}} \geq 1/r$ whenever $\norm{\Pi \ket{\psi}} \neq 0$ (the latter holds if and only if $\cT_1(R) \neq \emptyset$). 
Next, note that {\tt RCollisionCheck} can be implemented in time $2^{o(d)}$, using one query to $D(R_L)$ to compute $|R(\vx)|$.  
Therefore, \Cref{lem:AA-without-knowing-norm} implies that, with probability $1 - 2^{-\omega(d)}$, step 1 of ${\tt TupleSamp}$ can be implemented using $t_1 \leq_d \sqrt{\max\{1, |\cT_0(R)|\} / \max\{1, |\cT_1(R)|\}}$ time and QCRAM queries and using one auxiliary qubit, and maps \[\ket{0,0,0}\ket{0} \mapsto \frac{\Pi\ket{\psi}}{\norm{\Pi\ket{\psi}}} \ket{1}\] if $\cT_1(R) \neq \emptyset$, and $\ket{0,0,0}\ket{0} \to \ket{\psi}\ket{0}$ otherwise.  
In the remainder of this proof, we assume step~1 was successful, so the first flag qubit is either fully supported on $\ket{0}$ or fully supported on $\ket{1}$. Since steps~2 and~3 are only applied in the latter case, in the remainder of the proof we omit writing this first flag qubit. 

Step 2 can be implemented using the quantum algorithm ${\tt Dec}(C')$ from~\Cref{lem:R-beta-oracle} in time $t_2 = 2^{o(d)}$ using $2^{o(d)}$ auxiliary qubits.\footnote{We remark that ${\tt Dec}(C')$ takes as input a description of $C'$. We assume without loss of generality that such a description (which is of size $2^{o(d)}$) is stored in $D(R'_L)$ during the \textbf{Preprocessing} phase.} 
For step 3, we use the data structure $D(R'_L)$ to implement ${\tt zSamp}(\vc')$ using $2^{o(d)}$ time and QCRAM queries, as established by~\Cref{lem:DSlemma}. Note that $D(R'_L)$ stores the size $|(R'_L)^{\!\transpose}(\vc')|$, so step 3 (controlled on $(\vx,\vc,\vy,\vc')$) can be implemented using the algorithm ${\tt AA}_{r'}$ from~\Cref{lem:AA-without-knowing-norm} with $r' \coloneqq \sqrt{|(R'_L)^{\!\transpose}(\vc')|}$.  
In addition, ${\tt zCheck}(\vx,\vy)$ is a straightforward check that takes time at most $2^{o(d)}$. 
By~\Cref{lem:AA-without-knowing-norm} (see also~\Cref{rem:AA-special-case-search}), this implements step 3 with probability $1 - 2^{-\omega(d)}$ in time $t_3 = \sqrt{\max_{\vc' \in C'} |(R'_L)^{\!\transpose}(\vc')|} \, 2^{o(d)}$, and maps 
\begin{align*}
    \ket{\vx,\vc,\vy,\vc'}\ket{0}\ket{1,0} \mapsto \begin{cases}
        \displaystyle \ket{\vx,\vc,\vy,\vc'}\frac{1}{\sqrt{|L(\vx,\vy,\vc')|}} \sum_{\vz \in L(\vx,\vy,\vc')} \ket{\vz} \ket{1,1} &\text{if } \vc' \neq \bot \text{ and } L(\vx,\vy,\vc') \neq \emptyset \\
        \displaystyle \ket{\vx,\vc,\vy,\vc'}\ket{\bot} \ket{1,0} &\text{otherwise}
        \end{cases}
\end{align*}
where $L(\vx,\vy,\vc') \coloneqq \{ \vz \in L \colon \innerP{\tfrac{\vx-\vy}{\norm{\vx-\vy}}}{\vz} \approx_\epsapprox \cos(\theta'), \vc' \in R'(\tfrac{\vx-\vy}{\norm{\vx-\vy}}) \cap R'(\vz)\}$. 
Since $t_3 \leq_d |L|/|C'|$ by condition $(ii)$ of being well-balanced, ${\tt TupleSamp}$ uses 
\[ 
    t_1 + t_2 + t_3 \leq_d \sqrt{\frac{|\cT_0(R)|}{|\cT_1(R)|}} + \sqrt{\frac{|L|}{|C'|}}
\] 
time and QCRAM queries if $\cT_1(R) \neq \emptyset$, and $t_1 + t_2 + t_3 \leq_d \max\{1, \sqrt{|\cT_0(R)|}\}$ otherwise. 
The claim on the memory complexity follows immediately from~\Cref{lem:R-beta-oracle}, \Cref{lem:easy-analysis-xy-sampler}, and the construction of the algorithm. This proves  statement (1). 

By construction of {\tt TupleSamp}, its output state $\ket{\psi'}$ has nonzero overlap with $\ket{\vx,\vy,\vz} \ket{1}$ (the latter register corresponding to the second flag qubit) if and only if $(\vx,\vy,\vz) \in \cT_{\mathrm{sol}}(R,R')$, so fix an arbitrary $(\vx,\vy,\vz) \in \cT_{\mathrm{sol}}(R,R')$. Note that this implies $\cT_1(R) \neq \emptyset$, so the first flag qubit is fully supported on $\ket{1}$. By the above, $\ket{\psi'}$ satisfies
\begin{align*}
    \braket{\vx,\vc,\vy,\vc',\vz, 1, 1 | \psi'} = \frac{1}{\sqrt{\,\norm{\Pi \ket{\psi}}^{2}  |L| \, |R(\vx)| \, |(R_L)^{\!\transpose}(\vc)| \, |R'(\tfrac{\vx-\vy}{\norm{\vx-\vy}})| \, |L(\vx,\vy,\vc')|}} 
\end{align*}    
if $\vc \in R(\vx)\cap R(\vy)$ and $\vc' \in R'(\tfrac{\vx-\vy}{\norm{\vx-\vy}}) \cap R'(\vz)$; and $\braket{\vx,\vc,\vy,\vc',\vz, 1, 1 | \psi'} = 0$ otherwise. 
Using condition $(i)$ and $(iii)$ of being well-balanced and the fact that $|R(\vx)| =_d |R'(\tfrac{\vx-\vy}{\norm{\vx-\vy}})| =_d 1$ for $(\vx,\vy,\vz) \in \cT_{\mathrm{sol}}(R,R')$, we conclude that a measurement of $\ket{\psi'}$ yields outcome $(\vx,\vy,\vz, 1)$ with probability $1/|\cT_1(R)|$, up to subexponential factors in $d$.  
\end{proof}

\noindent Using the ingredients of step~1 of {\tt TupleSamp} (\Cref{algo:xyz-sampler}), we can also estimate $|\cT_1(R)|$ up to a subexponential factor in $d$, which will be used in~\Cref{algo:SolutionSearch}. 

\begin{lemma}[Estimating $|\cT_1(R)|$]\label{lem:estimate-size-T1-R}
    For $L,C \subseteq \cS^{d-1}$ of size $2^{O(d)}$, let $R \subseteq \cS^{d-1} \times C$ satisfy condition~$(i)$ in~\Cref{defn:well-balanced}, and let $D(R_L)$ be a QCRAM data structure for $R_L \coloneqq R \vert_{L \times C}$.
    There exists a quantum algorithm that, with probability $1 - 2^{-\omega(d)}$, returns an integer $\tilde{\tau}$ satisfying $\tilde{\tau} = 0$ if $\cT_1(R) = \emptyset$ and  $\tilde{\tau} =_d |\cT_1(R)|$ otherwise, using $\sqrt{\max\{1, |\cT_0(R)|\}/\max\{1, |\cT_1(R)|\}}2^{o(d)}$ time and QCRAM queries, and using $2^{o(d)}$ qubits.  
\end{lemma} 

\begin{proof}  
    Define {\tt RCollisionSamp} and {\tt RCollisionCheck} as in step~1 of~\Cref{algo:xyz-sampler}. Let $\ket{\psi}$ be the state generated by {\tt RCollisionSamp}, and $\Pi$ the orthogonal projector onto the subspace marked by {\tt RCollisionCheck}. 
    Note that $\cT_1(R) = \emptyset$ if and only if $\norm{\Pi \ket{\psi}} = 0$. Moreover, by the proof of~\Cref{lem:analysis-xyz-sampler} (using the assumption on $R$), if $\cT_1(R) \neq \emptyset$, then $\norm{\Pi \ket{\psi}}^2 =_d |\cT_1(R)|/|\cT_0(R)|$, where $|\cT_0(R)|$ is known up to a subexponential factor in $d$. 
    To prove the lemma, it thus suffices to show we can decide whether $\norm{\Pi \ket{\psi}} = 0$, and compute an estimate $\tilde{a} =_d \norm{\Pi \ket{\psi}}$, in the desired complexity and with success probability $1 - 2^{-\omega(d)}$.  

    To decide whether $\norm{\Pi \ket{\psi}} = 0$, we can apply step~1 of~{\tt TupleSamp} (\Cref{algo:xyz-sampler}) to the state $\ket{0,0,0}\ket{0}$, and measure the last (flag) qubit.  By the proof of~\Cref{lem:analysis-xyz-sampler}, with probability $1-2^{-\omega(d)}$, the measurement outcome is 0 if and only if $\norm{\Pi \ket{\psi}} = 0$, and the complexity is within the desired bounds. 

    Moreover, $\norm{\Pi \ket{\psi}}$ can be estimated using a variant of amplitude estimation. For instance, by~\cite[Thm.~3]{AR20}, there is a quantum algorithm that, with probability $1 - 2^{-\omega(d)}$, returns an integer $\tilde{a} = \Theta(\norm{\Pi \ket{\psi}})$ if $0 < \norm{\Pi \ket{\psi}} < 1$, using at most $\norm{\Pi \ket{\psi}}^{-1} 2^{o(d)} \leq \sqrt{|\cT_0(R)|/|\cT_1(R)|}2^{o(d)}$ applications of {\tt RCollisionSamp} and {\tt RCollisionCheck}, and using $2^{o(d)}$ qubits. 
    The possibility that $\norm{\Pi \ket{\psi}} = 1$ can be dealt with in affordable time using standard methods. 
\end{proof}

\subsubsection{SolutionSearch}
Finally, the outermost layer of the \textbf{Search} phase is \Cref{algo:SolutionSearch}, which uses amplitude amplification to project the output state $\ket{\psi'}$ of {\tt TupleSamp} (\Cref{algo:xyz-sampler}) onto $\cT_{\mathrm{sol}}(R,R')$. If $\cT_{\mathrm{sol}}(R,R')$ is nonempty, then this state $\ket{\psi'}$ is solely supported on $(\vx,\vy,\vz, b^{\prime}_F)$ with $(\vx,\vy)\in \cT_1(R)$, and with the flag bit set to~$b^{\prime}_F=1$ only if $(\vx,\vy,\vz)\in \cT_{\mathrm{sol}}(R,R')$.  
Thus, as long as we perform sufficiently many rounds of amplitude amplification, \Cref{algo:SolutionSearch} successfully outputs an element of $\cT_{\mathrm{sol}}(R,R')$. 

\begin{algorithm}[ht!]
\caption{${\tt SolutionSearch}(D(R_L),D(R'_L))$}\label{algo:SolutionSearch}
\vskip3pt
\makebox[3.7em][l]{Input:} QCRAM data structures $D(R_L)$ and $D(R'_L)$ for the relations 
\begin{equation*}
    R_L \coloneqq R \vert_{L \times C} \quad \text{ and } \quad R'_L \coloneqq R' \vert_{L \times C'}
\end{equation*} 
\makebox[3.7em][l]{ } where $L, C, C' \subseteq \cS^{d-1}$ are finite, $R \subseteq \cS^{d-1} \times C$, and $R' \subseteq \cS^{d-1} \times C'$ \\[4pt] 
\makebox[3.7em][l]{Output:} An element $(\vx,\vy,\vz)\in \cT_{\mathrm{sol}}(R,R')$ if $\cT_{\mathrm{sol}}(R,R') \neq \emptyset$, and ``no solution'' otherwise  
\vskip5pt
\hrule
\vskip5pt
\begin{enumerate}
    \item Initialize $\ket{0,0,0,0,0}\ket{0,0}$ 
    \item If $\cT_1(R) \neq \emptyset$ (decided using~\Cref{lem:estimate-size-T1-R}), then apply 
    ${\tt AA}_{r}({\tt TupleSamp}(D(R_L),D(R'_L)), \linebreak {\tt TupleCheck})$, where: 
    \begin{enumerate}
        \item[--] $r = \sqrt{|\cT_1(R)|} \, 2^{o(d)}$ (computed using~\Cref{lem:estimate-size-T1-R}) 
        \item[--] {\tt TupleSamp} (\Cref{algo:xyz-sampler}) generates a superposition over $\ket{\vx,\vc,\vy,\vc',\vz}\ket{b^{\phantom{\prime}}_F,b^{\prime}_F}$ such that either $b^{\prime}_F=0$ or $(\vx,\vy,\vz)\in \cT_{\mathrm{sol}}(R,R')$  
        \item[--] {\tt TupleCheck} checks if $b^{\prime}_F=1$  \end{enumerate}
    \item Measure. From outcome $(\vx,\vy,\vz,b^{\prime}_F)$, output $(\vx,\vy,\vz)$ if $b^{\prime}_F=1$, and ``no solution'' otherwise. 
\end{enumerate}
\end{algorithm}

\noindent 
The relations between the vectors encountered during {\tt SolutionSearch} are visualized in~\Cref{fig:rel-structure}. 

\begin{lemma}[Analysis of {\tt SolutionSearch}]\label{lem:analysis-SolutionSearch} 
    For $L, C, C' \subseteq \cS^{d-1}$ of size $2^{O(d)}$, let $R \subseteq \cS^{d-1} \times C$ and $R' \subseteq \cS^{d-1} \times C'$ be such that $(R,R')$ is well-balanced on $L$ (\Cref{defn:well-balanced}).  
    Let $D(R_L)$ and $D(R'_L)$ be QCRAM data structures for $R_L \coloneqq R \vert_{L \times C}$ and $R'_L \coloneqq R' \vert_{L \times C'}$. 
    With probability $1 - 2^{-\omega(d)}$, ${\tt SolutionSearch}(D(R_L),D(R'_L))$ (\Cref{algo:SolutionSearch}) satisfies the following: 
    \begin{enumerate}
        \item[(1)] \text{Complexity:} 
        If $\cT_{\mathrm{sol}}(R,R') \neq \emptyset$, then it outputs $(\vx,\vy,\vz) \in \cT_{\mathrm{sol}}(R,R')$ using
        \begin{align*}
            \sqrt{\frac{|\cT_1(R)|}{|\cT_{\mathrm{sol}}(R,R')|}} \left(\sqrt{\frac{|\cT_0(R)|}{|\cT_1(R)|}} + \sqrt{\frac{|L|}{|C'|}} \right) 2^{o(d)}
        \end{align*}
        time and QCRAM queries; otherwise, it outputs ``no solution'' using
        \begin{align*}
            \left(\max\left\{1, \sqrt{|\cT_0(R)|}\right\} + \sqrt{|\cT_1(R)|\frac{|L|}{|C'|}} \right) 2^{o(d)}
        \end{align*}
        time and QCRAM queries. 
        It uses $2^{o(d)}$ qubits. 
        \item[(2)] \text{Near-uniform sampling:} 
        For all $(\vx,\vy,\vz) \in \cT_{\mathrm{sol}}(R,R')$, 
        \begin{align*}
            \Pr\left[\text{${\tt SolutionSearch}(D(R_L),D(R'_L))$ outputs $(\vx,\vy,\vz)$}\right] =_d  \frac{1}{|\cT_{\mathrm{sol}}(R,R')|} 
        \end{align*} 
        where the probability is over the internal randomness of {\tt SolutionSearch}. 
    \end{enumerate}
\end{lemma}

\begin{proof} 
For step~2 of {\tt SolutionSearch}, we invoke the algorithm from~\Cref{lem:estimate-size-T1-R} to decide whether $\cT_1(R) \neq \emptyset$ and to estimate $|\cT_1(R)|$ up to a subexponential factor in $d$. This uses $2^{o(d)}$ auxiliary qubits. If $\cT_1(R) = \emptyset$, then the lemma statement follows immediately, so suppose  $\cT_1(R) \neq \emptyset$. 
By~\Cref{lem:analysis-xyz-sampler}, with probability at least $1 - 2^{-\omega(d)}$, 
the output state $\ket{\psi'}$ of {\tt TupleSamp} (\Cref{algo:xyz-sampler})  satisfies 
\begin{equation*}
\begin{split}
    \norm{\Pi' \ket{\psi'}}^2 =_d \begin{cases}
        \displaystyle \frac{|\cT_{\mathrm{sol}}(R,R')|}{|\cT_1(R)|} &\quad\text{if } \cT_{\mathrm{sol}}(R,R') \neq \emptyset \\ 
        \displaystyle 0 &\quad\text{otherwise}
    \end{cases}
\end{split}
\end{equation*} 
where $\Pi'$ denotes the orthogonal projector onto the subspace where $b^{\prime}_F = 1$. If $\cT_{\mathrm{sol}}(R,R') \neq \emptyset$, then $\norm{\Pi' \ket{\psi'}} \geq_d 1/\sqrt{|\cT_1(R)|}$. Choosing $r = \sqrt{|\cT_1(R)|} \, 2^{o(d)}$ sufficiently large (using the computed estimate for $|\cT_1(R)|$) thus ensures $\norm{\Pi' \ket{\psi'}} \geq 1/r$, allowing us to invoke algorithm ${\tt AA}_{r}$ from~\Cref{lem:AA-without-knowing-norm} to implement step~2. 
Since ${\tt TupleCheck}$ can be implemented in time $2^{o(d)}$, statements (1) and (2) now follow from~\Cref{lem:analysis-xyz-sampler}. 
\end{proof}

\subsubsection{Overall analysis of the Search phase} 
We now complete our analysis of the \textbf{Search} phase of {\tt 3List} (\Cref{algo:main-3List}). By applying a coupon-collector argument (see~\Cref{lem:find-all}), it follows that for a sufficiently large search threshold $\searchthreshold = 2^{o(d)}$, the \textbf{Search} phase finds, with overwhelming probability, all elements of $\cT_{\mathrm{sol}}(R,R')$ using at most $|\cT_{\mathrm{sol}}(R,R')| 2^{o(d)}$ calls to {\tt SolutionSearch}, even if $|\cT_{\mathrm{sol}}(R,R')|$ is not known to the algorithm a priori. This establishes the overall complexity of the \textbf{Search} phase.

\begin{lemma}[Analysis of the \textbf{Search} phase of {\tt 3List}]\label{lem:analysis-Search-phase} 
Let $\searchthreshold = 2^{o(d)}$ be sufficiently large.\footnote{By ``sufficiently large'', we mean that there exists some $f(d) \leq 2^{o(d)}$, determined by the proof of~\Cref{lem:analysis-Search-phase}, such that the lemma holds for all $\searchthreshold = 2^{o(d)}$ with $\searchthreshold \geq f(d)$.\label{footnote:sufficiently-large}} 
Consider a single repetition of the \textbf{Search} phase of {\tt 3List} (\Cref{algo:main-3List}) for $L, C,C' \subseteq \cS^{d-1}$ of size $2^{O(d)}$, given the QCRAM data structures $D(R_L)$ and $D(R'_L)$ constructed in the \textbf{Preprocessing} phase. 
Suppose $(R,R')$ is well-balanced on $L$ (\Cref{defn:well-balanced}). With probability at least $1 - 2^{-\omega(d)}$, the \textbf{Search} phase constructs a list consisting of all elements of $\cT_{\mathrm{sol}}(R,R')$ and uses 
\begin{align*}
     \sqrt{\max\{1,|\cT_{\mathrm{sol}}(R,R')|\}} \left(\max\{1,\sqrt{|\cT_0(R)|}\} + \sqrt{|\cT_1(R)| \frac{|L|}{|C'|}} \right) 2^{o(d)}
\end{align*}   
time and QCRAM queries, $2^{o(d)}$ qubits, and $\max\{1, |\cT_{\mathrm{sol}}(R,R')|\} \, 2^{o(d)}$ additional classical memory. 
\end{lemma}

\begin{proof} 
The \textbf{Search} phase of {\tt 3List} is exactly the algorithm $\cA({\tt Samp}, \searchthreshold)$ in the proof of~\Cref{lem:find-all} with $X = \cT_{\mathrm{sol}}(R,R')$ and ${\tt Samp} = {\tt SolutionSearch}$. 
If $\cT_{\mathrm{sol}}(R,R') = \emptyset$, then the \textbf{Search} phase will finish after at most $\searchthreshold$ repetitions of step~$ii$, correctly yielding an empty set $S$. Otherwise, by~\Cref{lem:analysis-SolutionSearch}, with probability $1 - 2^{-\omega(d)}$, there exists $\varepsilon \geq 2^{-o(d)}$ such that \[\Pr[{\tt SolutionSearch} \text{ outputs } (\vx,\vy,\vz)] \geq \varepsilon / |\cT_{\mathrm{sol}}(R,R')|\] for all $(\vx,\vy,\vz) \in \cT_{\mathrm{sol}}(R,R')$.  
Therefore, for all $\searchthreshold \geq \frac{1}{\varepsilon} \ln(|\cT_{\mathrm{sol}}(R,R')|) + \omega(d)$, ~\Cref{lem:find-all} implies that the \textbf{Search} phase outputs all elements of $\cT_{\mathrm{sol}}(R,R')$ in time $\searchthreshold \, {\sf S} \, |\cT_{\mathrm{sol}}(R,R')|^{1+o(d)}$, where ${\sf S}$ denotes the time complexity of ${\tt SolutionSearch}$. Since $|\cT_{\mathrm{sol}}(R,R')| \leq |L|^3 \leq 2^{O(d)}$, it suffices to take a sufficiently large $\searchthreshold = 2^{o(d)}$. The complexity of the \textbf{Search} phase then follows from~\Cref{lem:analysis-SolutionSearch}. 
\end{proof}

\subsection{Final ingredients}

This section presents a few lemmas that will be used to prove our main theorem, \Cref{thm:main-3List}. 

\begin{lemma}\label{lem:bucket-size} 
    Let $m$ be a positive integer and let $\alpha \in (0, \pi/2)$ satisfy $\alpha = \Omega(1)$ and $m p_{\alpha} = \omega(d)$. 
    Let $C \subseteq \cS^{d-1}$ be of size $2^{O(d)}$. 
    With probability $1 - 2^{-\omega(d)}$ over $L \sim \cU(\cS^{d-1},m)$, $|(R_L)^{\!\transpose}(\vc)| =_d m  p_{\alpha}$ for all $\vc \in C$, and thus $|R_L| =_d |C| m p_{\alpha}$, where 
    $R_L \coloneqq \{(\vx,\vc) \in L \times C \colon \innerP{\vx}{\vc} \approx_\epsapprox \cos(\alpha)\}$.
\end{lemma}

\begin{proof}
Fix $\vc \in C$. View $L \sim \cU(\cS^{d-1},m)$ as an ordered sequence $(\vx_1,\dots,\vx_m)$ of i.i.d.\ uniformly random unit vectors, and consider the sum $X = \sum_{i \in [m]} X_i$ of i.i.d.\ random variables $X_i \in \{0,1\}$ defined by $X_i = 1$ if and only if $\innerP{\vx_i}{\vc} \approx_\epsapprox \cos(\alpha)$. Note that $X = |(R_L)^{\!\transpose}(\vc)|$. 
By~\Cref{lem: approx cap volume} and by assumption, $\E_L[X] =_d m p_{\alpha} = \omega(d)$, so~\Cref{cor:simple-application-of-Chernoff} implies $|(R_L)^{\!\transpose}(\vc)| =_d m p_{\alpha}$ except with probability at most $\exp(-\omega(d))$.  
Hence, applying a union bound over all $2^{O(d)}$ vectors $\vc\in C$ yields $|(R_L)^{\!\transpose}(\vc)| =_d m p_{\alpha}$ for all $\vc \in C$, except with probability $|C| \exp(-\omega(d)) = \exp(-\omega(d))$.
\end{proof}

\begin{lemma}\label{lem:number-of-z-under-distribution-L}
    Let $m = 2^{O(d)}$ and let $\theta' \in (0,\pi/2)$ satisfy $\theta' = \Omega(1)$. 
    With probability $1 - 2^{-\omega(d)}$ over $L \sim \cU(\cS^{d-1},m)$, $|\{\vz \in L \colon \innerP{\tfrac{\vx-\vy}{\norm{\vx-\vy}}}{\vz} \approx_\epsapprox \cos(\theta') \}| \leq_d \max\{1, m p_{\theta'}\}$ for all $(\vx, \vy) \in L^2$.  
\end{lemma}

\begin{proof} 
View $L$ as an ordered sequence $(\vx_1,\dots,\vx_m)$ and let $(i,j) \in [m]^2$. For all $k \in [m-2]$, define the random variable $X_k \in \{0,1\}$ by $X_k = 1$ if and only if $\innerP{(\vx_i-\vx_j)/\norm{\vx_i-\vx_j}}{\vx_k} \approx_\epsapprox \cos(\theta')$. If  $L \sim \cU(\cS^{d-1},m)$, then the $X_k$ are i.i.d.\ and $\E_L[\sum_{k \in [m-2]} X_k] =_d m p_{\theta'}$ by~\Cref{lem: approx cap volume}. Therefore, \Cref{cor:simple-application-of-Chernoff} implies $\sum_{k \in [m-2]} X_k \leq_d \max\{1, m p_{\theta'}\}$ except with probability at most $\exp(-\omega(d))$. 
The lemma statement now follows by applying a union bound over all $(i,j) \in [m]^2$, using that $m = 2^{O(d)}$ and that including $\vx_i$ or $\vx_j$ in the count does not affect the asymptotic upper bound.  
\end{proof}

\begin{lemma}[Size of $\cT_1(R)$]\label{lem:ub-on-T1-R} 
    Let $m = 2^{O(d)}$ and let $\theta, \alpha \in (0, \pi/2)$ satisfy $\min\{\theta, 2\alpha - \theta\} = \Omega(1)$ and $m\cW(\theta, \alpha \mid \alpha) = \omega(d)$.
    Let $C \subseteq \cS^{d-1}$ be of size $2^{O(d)}$, $R \coloneqq \{(\vx,\vc) \in \cS^{d-1} \times C \colon \innerP{\vx}{\vc} \approx_\epsapprox \cos(\alpha)\}$, and $R_L \coloneqq R \vert_{L \times C}$. 
    With probability $1 - 2^{-\omega(d)}$ over $L \sim \cU(\cS^{d-1},m)$, 
    $|\cT_1(R)| \leq_d |R_L| m \cW(\theta, \alpha \mid \alpha)$,  
    where $\cT_1(R)$ is defined in~\Cref{eq:defn-Ti-R}.
\end{lemma}

\begin{proof}
Viewing $L$ as an ordered sequence $(\vx_1,\dots,\vx_m)$ of unit vectors, let $i \in [m]$ and $\vc \in R(\vx_i)$. Note that $L \sim \cU(\cS^{d-1},m)$ implies $L_{-i} \sim \cU(\cS^{d-1},m-1)$, where $L_{-i} \coloneqq (\vx_1, \dots, \vx_{i-1}, \vx_{i+1}, \dots, \vx_m)$. Therefore, by~\Cref{lem: approx wedge volume}, $\Pr_{\vy \sim \cU(L_{-i})}[\innerP{\vx_i}{\vy} \approx_\epsapprox \cos(\theta), \innerP{\vy}{\vc} \approx_\epsapprox \cos(\alpha)] =_d \cW(\theta,\alpha \mid \alpha)$, giving 
\[
    \E_{L_{-i}}[|\{\vy \in L_{-i} \colon \innerP{\vx_i}{\vy} \approx_\epsapprox \cos(\theta), \innerP{\vy}{\vc} \approx_\epsapprox \cos(\alpha)\}|] =_d (m - 1) \cW(\theta,\alpha \mid \alpha) =_d m\cW(\theta,\alpha \mid \alpha)
\]
which is $\omega(d)$. 
Thus, $|\{\vy \in L_{-i} \colon \innerP{\vx_i}{\vy} \approx_\epsapprox \cos(\theta), \innerP{\vy}{\vc} \approx_\epsapprox \cos(\alpha)\}| =_d m \cW(\theta,\alpha \mid \alpha)$ except with probability $\exp(-\omega(d))$ by~\Cref{cor:simple-application-of-Chernoff}. 
Using a union bound over all $i \in [m]$ and $\vc \in R(\vx_i)$ yields $\sum_{i \in [m]} \sum_{\vc \in R(\vx_i)}  |\{\vy \in L_{-i} \colon \innerP{\vx_i}{\vy} \approx_\epsapprox \cos(\theta), \innerP{\vy}{\vc} \approx_\epsapprox \cos(\alpha)\}| \leq_d |R_L| m \cW(\theta,\alpha \mid \alpha)$, except with probability $m |C| \exp(-\omega(d)) = \exp(-\omega(d))$, so the statement follows by definition of $\cT_1(R)$.  
\end{proof}

\begin{lemma}[Size of $\cT_{\mathrm{sol}}(L, \theta,\theta')$]\label{lem:size-of-Tsol}
    Let $m = 2^{\Omega(d)}$ and let $\theta, \theta' \in (0,\pi/2)$ satisfy $\min\{\theta, \theta'\} = \Omega(1)$ and $\min\{m^2 p_{\theta}, m^2 p_{\theta'}, m^3 p_{\theta} p_{\theta'}\} = 2^{\Omega(d)}$.  
    With probability $1 - 2^{-\Omega(d)}$ over $L \sim \cU(\cS^{d-1}, m)$, we have $|\{(\vx,\vy,\vz) \in \cT_{\mathrm{sol}}(L, \theta,\theta') \colon \text{$\vx,\vy,\vz$ are distinct}\}| =_d m^3 p_{\theta} p_{\theta'}$, where $\cT_{\mathrm{sol}}(L, \theta,\theta')$ is defined in~\Cref{eq:T-sol}. 
\end{lemma}

\begin{proof}
We think of $L$ as an ordered sequence $(\vx_1,\dots,\vx_m)$ of i.i.d.\ uniformly random unit vectors. Let $I$ denote the set of all $m(m-1)(m-2)$  triples $(i,j,k)$ of distinct indices from $[m]$. For each $(i,j,k) \in I$, let $Z_{ijk}$ be the indicator random variable defined by $Z_{ijk} = 1$ if and only if both $\innerP{\vx_{i}}{\vx_{j}} \approx_\epsapprox \cos(\theta)$ and $\innerP{\tfrac{\vx_{i}-\vx_{j}}{\norm{\vx_{i}-\vx_{j}}}}{\vx_{k}} \approx_\epsapprox \cos(\theta')$. 
Then $Z \coloneqq \sum_{(i,j,k) \in I}Z_{ijk}$ counts the number of elements in $\cT_{\mathrm{sol}}(L, \theta,\theta')$ that consist of distinct vectors.  
By linearity of expectation, $\mu \coloneqq \E[Z]$ satisfies $\mu=m(m-1)(m-2)p p'$, where 
\begin{equation*}
    p \coloneqq \Pr[\innerP{\vx}{\vy} \approx_\epsapprox \cos(\theta)] \quad \text{ and } \quad p' \coloneqq \Pr[\innerP{\tfrac{\vx-\vy}{\norm{\vx-\vy}}}{\vz} \approx_\epsapprox \cos(\theta')]
\end{equation*}
for independent samples $\vx,\vy,\vz \sim \cU(\cS^{d-1})$. 
By~\Cref{lem: approx cap volume}, $p =_d p_\theta$ and $p' =_d p_{\theta'}$, so it suffices to show $Z =_d \mu$. 
We will now argue that the variance of $Z$ is 
$\sigma^2 \leq \mu^2 2^{-\Omega(d)}$. Chebyshev's inequality (which gives $\Pr[|Z-\mu|\geq k\sigma]\leq 1/k^2$ for any real $k > 0$) then implies tight concentration: the probability that $|Z - \mu| \geq \mu/2$ is at most $4\sigma^2/\mu^2$, so $Z =_d \mu$ except with probability $2^{-\Omega(d)}$. 

To upper bound $\sigma^2=\E[Z^2]-\mu^2$, we note $\E[Z^2]=\sum_{(i,j,k) \in I} \sum_{(i',j',k') \in I} \E[Z_{ijk} Z_{i'j'k'}]$.
Fix $(i,j,k) \in I$ and partition $I$ into four disjoint sets $I_0, I_1, I_2, I_3$, where $I_\ell \coloneqq \{(i',j',k') \in I \colon |\{i,j,k\} \cap \{i',j',k'\}| = \ell\}$.  
By independence and by definition of the sets $I_\ell$, we obtain \begin{align*}
    \sum_{(i',j',k') \in I_\ell}\Pr[Z_{i'j'k'} = 1 \mid Z_{ijk} = 1] \leq \begin{cases}
        (m-3)(m-4)(m-5) pp' &\quad\text{if $\ell = 0$} \\ 
        3(m-3)(m-4) pp' &\quad\text{if $\ell = 1$} \\ 
        3(m-3) \max\{p,p'\} &\quad\text{if $\ell = 2$} \\ 
        1 &\quad\text{if $\ell = 3$.}
    \end{cases}
\end{align*} 
Since $\mu = m(m-1)(m-2)pp'$, we obtain $\E[Z^2] \leq \mu^2 + O(\mu^2/\min\{m, m^2 p, m^2 p', \mu\})$. By assumption, $\min\{m, m^2 p, m^2 p', \mu\} =_d \min\{m, m^2 p_\theta, m^2 p_{\theta'}, m^3 p_\theta p_{\theta'}\} = 2^{\Omega(d)}$, which implies $\sigma^2 \leq \mu^2 2^{-\Omega(d)}$. As explained before, the lemma statement now follows from Chebyshev's inequality. 
\end{proof} 

\begin{lemma}\label{lem:probability-being-covered-by-RPC}
Let $\theta, \alpha \in (0, \pi/2)$ satisfy $\min\{\theta, 2\alpha - \theta\} = \Omega(1)$. 
For all $\vx,\vy\in \cS^{d-1}$ that satisfy $\innerP{\vx}{\vy} \approx_\epsapprox \cos(\theta)$, 
\begin{align*}
    \Pr_{C \sim \mathrm{RPC}(d,\log d,1/p_{\alpha})}\!\big[R(\vx) \cap R(\vy) \neq \emptyset \text{ and } |R(\vx)| \leq 2^{d /\!\log d}\big] =_d \frac{\cW(\alpha, \alpha \mid \theta)}{p_\alpha} 
\end{align*}
where $R \coloneqq \{(\vx,\vc) \in \cS^{d-1} \times C \colon \innerP{\vx}{\vc} \approx_\epsapprox \cos(\alpha)\}$.
\end{lemma}

\begin{proof} 
Define $b \coloneqq \log d$. 
The upper bound follows from~\Cref{lem: approx wedge volume}, because $\innerP{\vx}{\vy} \approx_\epsapprox \cos(\theta)$. 
Since $\Pr_{C}\!\big[|R(\vx)| \in [1, 2^{d /\!\log d}]\big] \geq_d 1$ by~\Cref{lem:random-behavior-RPCs} (applied with $\theta = 0$) and a careful application of Markov's inequality, we obtain 
\begin{align*}
    \Pr_{C}\!\big[R(\vx) \cap R(\vy) \neq \emptyset \text{ and } |R(\vx)| \leq 2^{d /\!\log d}\big] &= \Pr_{C}\!\big[R(\vx) \cap R(\vy) \neq \emptyset \text{ and } |R(\vx)| \in [1, 2^{d /\!\log d}]\big] \\ 
    &\geq_d \Pr_{C}\!\big[R(\vx) \cap R(\vy) \neq \emptyset \, \big\vert \, |R(\vx)| \in [1, 2^{d /\!\log d}]\big]. 
\end{align*} 
Therefore, it remains to lower bound the latter, i.e., the probability of the event $E$ that there is a center point  $\vc \in C$  that lands in the region 
$$\widetilde{\cW}_{\vx,\alpha,\vy,\alpha} = \{\vs \in \cS^{d-1} \colon \innerP{\vx}{\vs} \approx_\epsapprox \cos(\alpha), \innerP{\vy}{\vs} \approx_\epsapprox \cos(\alpha)\},$$ 
conditioned on the event $E'$ that the number of center points $\vc \in C$ that lie in the ``approximate'' spherical cap $\widetilde{\cH}_{\vx,\alpha} = \{\vs \in \cS^{d-1} \colon \innerP{\vx}{\vs} \approx_\epsapprox \cos(\alpha)\} \supseteq \widetilde{\cW}_{\vx,\alpha,\vy,\alpha}$ is between $1$ and $2^{d /\!\log d}$. 
Thus, suppose event $E'$ holds. Then event $E$ holds if one of those center points $\vc^*$ that lie in $\widetilde{\cH}_{\vx,\alpha}$ (let's say $\vc^*$ is the first such center point, assuming without loss of generality some ordering on $C$) also satisfies $\innerP{\vy}{\vc^*} \approx_\epsapprox \cos(\alpha)$, i.e., $\vc^* \in \widetilde{\cW}_{\vx,\alpha,\vy,\alpha}$. 

We recall that the distribution $C \sim \mathrm{RPC}(d,b,1/p_{\alpha})$ is defined as taking $C = \mQ(C^{(1)} \times \dots \times C^{(b)})$, where $\mQ \sim \cU(\mathrm{SO}(d))$ and $C^{(i)} \sim \cU(\frac{1}{\sqrt{b}}\cS^{d/b -1}, (1/p_{\alpha})^{1/b})$ for $i \in \{1,\dots,b\}$. 
Therefore, sampling $C \sim \mathrm{RPC}(d,b,1/p_{\alpha})$ conditional on $|R(\vx)| \in [1, 2^{d /\!\log d}]$ is equivalent to sampling $C$ as follows: 
\begin{enumerate}
    \item Sample $\vc^* \coloneqq \mQ(\vv^*_1,\dots,\vv^*_b)$, by first sampling $\mQ \sim \cU(\mathrm{SO}(d))$ and then sampling the tuple $(\vv^*_1,\dots,\vv^*_b) \sim \cU(\frac{1}{\sqrt{b}}\cS^{d/b -1}) \times \dots \times\cU(\frac{1}{\sqrt{b}}\cS^{d/b -1})$ conditional on $\innerP{\vx}{\vc^*} \approx_\epsapprox \cos(\alpha)$.
    \item Sample $\widetilde{C}^{(i)} \sim \cU(\frac{1}{\sqrt{b}} \cS^{d/b-1}, (1/p_{\alpha})^{1/b} - 1)$ for all $i \in \{1,\dots,b\}$, conditional on $|C \cap \widetilde{\cH}_{\vx,\alpha}| \in [1,2^{d /\!\log d}]$ where $C \coloneqq \mQ(C^{(1)}\times \dots \times C^{(b)})$ for $C^{(i)} \coloneqq \{\vv^*_i\} \cup \widetilde{C}^{(i)}$. 
\end{enumerate} 
The second step does not influence the probability that the vector $\vc^*$ sampled in the first step lands in $\widetilde{\cW}_{\vx,\alpha,\vy,\alpha}$. Moreover, the distribution of the vector $\vc^*$ sampled in the first step is the same as sampling $\vc^* \sim \cU(\cS^{d-1})$ conditional on $\innerP{\vx}{\vc^*} \approx_\epsapprox \cos(\alpha)$. 
In other words, we obtain
\begin{align*}
    \Pr_{C}\!\big[R(\vx) \cap R(\vy) \neq \emptyset \mid |R(\vx)| \in [1, 2^{d /\!\log d}]\big]  
    &\geq \Pr_{C}\!\big[ \innerP{\vy}{\vc^*} \approx_\epsapprox \cos(\alpha) \, \big\vert \, |R(\vx)| \in [1, 2^{d /\!\log d}] \big]  \\ 
    &= \Pr_{\vc^* \sim \cU(\cS^{d-1})}\!\big[ \innerP{\vy}{\vc^*} \approx_\epsapprox \cos(\alpha) \, \big\vert \, \innerP{\vx}{\vc^*} \approx_\epsapprox \cos(\alpha) \big] \\
    &=_d \frac{\cW(\alpha, \alpha \mid \theta)}{p_{\alpha}}. 
\end{align*} 
by~\Cref{lem: approx cap volume} and~\Cref{lem: approx wedge volume} (again, using $\innerP{\vx}{\vy} \approx_\epsapprox \cos(\theta)$). 
\end{proof}

\subsection{Final analysis of our quantum algorithm}\label{sec:proof-of-main-thm}

We now show that {\tt 3List} (i.e., \Cref{algo:main-3List}) solves~\Cref{prob:finding-many-3-tuples} for a random list $L \subseteq \cS^{d-1}$ of size $m = (27/16)^{d/4 +o(d)}$ in complexity $2^{0.284551d + o(d)}$, thereby proving~\Cref{thm:main-3List}. 
First, we formalize the claim we made in the introduction of~\Cref{sec:main-quantum-algorithm} regarding $\cT_{\mathrm{sol}}(L, \theta,\theta')$.  

\begin{lemma}\label{lem:size-of-Tsol-application} 
    Let $L \subseteq \cS^{d-1}$ and let $\theta, \theta' \in (0,\pi/2)$ satisfy $\cos(\theta) = \frac{1}{3}$ and $\cos(\theta') = \epsapprox + \sqrt{\frac{1}{3} + \frac{\epsapprox}{2}}$. 
    Then $\norm{\vx - \vy - \vz} \leq 1$ for all $(\vx,\vy,\vz) \in \cT_{\mathrm{sol}}(L, \theta,\theta')$, where $\cT_{\mathrm{sol}}(L, \theta,\theta')$ is defined in~\Cref{eq:T-sol}. 
    
    \noindent Moreover, there exists $m' = (27/16)^{d/4 + o(d)}$ such that for all $m \geq m'$: 
    \begin{enumerate}[label=(\roman*)]
        \item $\E_{L \sim \cU(\cS^{d-1},m)}[|\cT_{\mathrm{sol}}(L, \theta,\theta')|] \geq m$. 
        \item With probability $1 - 2^{-\Omega(d)}$ over $L \sim \cU(\cS^{d-1},m)$, $m \leq |\cT_{\mathrm{sol}}(L, \theta,\theta')| \leq_d m^3 p_\theta p_{\theta'}$ and every $(\vx,\vy,\vz) \in \cT_{\mathrm{sol}}(L, \theta,\theta')$ satisfies $\vx \neq \vy$, $\vy \neq \vz$, and $\vz \neq \vx$. 
    \end{enumerate}
\end{lemma} 

\begin{proof} 
For all $\vx,\vy,\vz \in \cS^{d-1}$, we have $\norm{\vx - \vy - \vz}^2 \leq 1$ if and only if $\innerP{\tfrac{\vx-\vy}{\norm{\vx-\vy}}}{\vz} \geq \sqrt{(1 - \innerP{\vx}{\vy})/2}$. 
Consider arbitrary~$\theta \in (0, \pi/2)$, and define $\theta'$ by $\cos(\theta') = \epsapprox + \sqrt{(1 - \cos(\theta) + \epsapprox)/2}$. 
Then all $(\vx,\vy,\vz) \in \cT_{\mathrm{sol}}(L, \theta,\theta')$ 
satisfy 
\[
    \sqrt{\frac{1 - \innerP{\vx}{\vy}}{2}} \leq \sqrt{\frac{1 - \cos(\theta) + \epsapprox}{2}} = \cos(\theta') - \epsapprox \leq \innerP{\frac{\vx-\vy}{\norm{\vx-\vy}}}{\vz} 
\] 
and thus $\norm{\vx - \vy - \vz} \leq 1$. 
We will now show that the lemma statement holds for $\cos(\theta) = 1/3$ (and thus $\cos(\theta') = \epsapprox + \sqrt{1/3 + \epsapprox/2}$).

Note that $p_{\theta'} = (1 - \cos^2(\theta'))^{d/2} =_d (1 - (\frac{1 - \cos(\theta)}{2}))^{d/2} = (\frac{1 + \cos(\theta)}{2})^{d/2} = \cos(\tfrac{\theta}{2})^{d}$ by the fixed choice of $\epsapprox = 1/(\log d)^2$. Therefore, if $\theta = \Omega(1)$ and $\theta' = \Omega(1)$, we obtain (for $L \sim \cU(\cS^{d-1},m)$)
\begin{align*}
    \E_{L}\big[|\cT_{\mathrm{sol}}(L, \theta,\theta')|\big] &= \sum_{(\vx,\vy,\vz) \in L^3} \Pr_{L}\!\big[\innerP{\vx}{\vy} \approx_\epsapprox \cos(\theta)\big] \cdot \Pr_{L}\!\big[\innerP{\tfrac{\vx-\vy}{\norm{\vx-\vy}}}{\vz} \approx_\epsapprox \cos(\theta') \, \big\vert \, \innerP{\vx}{\vy} \approx_\epsapprox \cos(\theta)\big]  \\
    &=_d m^3 p_{\theta} p_{\theta'} \\ 
    &=_d m^3 (\sin(\theta)\cos(\tfrac{\theta}{2}))^{d}
\end{align*}  
by~\Cref{lem: approx cap volume}.
In particular, there exists $m' = (\sin(\theta)\cos(\tfrac{\theta}{2}))^{-d/2}  2^{o(d)}$ such that, whenever the list size $|L|=m$ is $\geq m'$, the expected size of $\cT_{\mathrm{sol}}(L, \theta,\theta')$ is at least $m$. 
As $(\sin(\theta)\cos(\tfrac{\theta}{2}))^{-d/2}$ is minimized for $\cos(\theta) = 1/3$, we choose $\cos(\theta) = 1/3$ and choose $\theta'$ accordingly. 
The corresponding $m'$ then satisfies $m' = (27/16)^{d/4 + o(d)}$. 

It remains to prove part~$(ii)$. For all sufficiently large $d$, with probability $1$ over $L \sim \cU(\cS^{d-1},m)$, every $(\vx,\vy,\vz) \in \cT_{\mathrm{sol}}(L, \theta,\theta')$ satisfies  $\vx \neq \vy$, $\vy \neq \vz$, and $\vz \neq \vx$ (recall~\Cref{rem:elems-Tsol-are-distinct}). Indeed, consider $(\vx,\vy,\vz) \in \cT_{\mathrm{sol}}(L, \theta,\theta')$, which implies $\norm{\vx-\vy-\vz} \leq 1$. By definition, $\innerP{\vx}{\vy} \approx_\epsapprox \cos(\theta)$, so $\vx \neq \vy$ for sufficiently large $d$. Next, $\vy \neq \vz$, because otherwise $\|\vx - 2 \vy\| > 1$, a contradiction. Finally, if $\vx \neq \vy$ while $\vx = \vz$, then $\innerP{\vx}{\vy} = \cos(\theta) - \epsapprox$ (by definition of $\cT_{\mathrm{sol}}(L, \theta,\theta')$ and $\theta'$), which occurs with measure 0 for independent $\vx, \vy \sim \cU(\cS^{d-1})$. 
Therefore, by choosing $m' = (27/16)^{d/4 + o(d)}$ sufficiently large, part~$(ii)$ follows from~\Cref{lem:size-of-Tsol} (note that $\theta$ and $\theta'$ satisfy the constraints). 
\end{proof}

\noindent The following lemma implies that if $|\cT_{\mathrm{sol}}(L, \theta,\theta')| \geq m$ and the parameters $(\alpha, \alpha',\nrep, \searchthreshold)$ of ${\tt 3List}$ are chosen appropriately, then, with high probability, ${\tt 3List}$ outputs $m$ distinct elements of $\cT_{\mathrm{sol}}(L, \theta,\theta')$. 
To prove~\Cref{thm:main-3List}, it then remains to show that, for $\theta,\theta' \in (0,\pi/2)$ as given in~\Cref{lem:size-of-Tsol-application}, there exists a suitable choice of $(\alpha, \alpha')$ such that the complexity is as desired. 
 
\begin{lemma}\label{lem:main-3List-numberofreps} 
    Let $m = 2^{\Theta(d)}$ and let~$\theta, \theta', \alpha, \alpha' \in (0, \pi/2)$ satisfy $\min\{\theta, \theta', 2\alpha - \theta, 2\alpha' - \theta'\} = \Omega(1)$, $m p_{\theta'} \leq 2^{o(d)}$, and $\min \{m p_{\alpha}, m p_{\alpha'}, m \, \cW(\theta, \alpha \mid \alpha), m^2 \, \cW(\theta', \alpha' \mid \alpha')\} = \omega(d)$. 
    Let $\nrep = \frac{p_{\alpha}}{\cW(\alpha, \alpha \mid \theta)}  \frac{p_{\alpha'}}{\cW(\alpha', \alpha' \mid \theta')} 2^{o(d)}$ and $\searchthreshold = 2^{o(d)}$ be sufficiently large.\footnote{Similarly to~\Cref{lem:analysis-Search-phase}, ``sufficiently large'' means that there exists some $f(d) \leq 2^{o(d)}$ such that~\Cref{lem:main-3List-numberofreps} holds for all $\nrep = \frac{p_{\alpha}}{\cW(\alpha, \alpha \mid \theta)}  \frac{p_{\alpha'}}{\cW(\alpha', \alpha' \mid \theta')} 2^{o(d)}$ and $\searchthreshold = 2^{o(d)}$ with $\nrep \geq \frac{p_{\alpha}}{\cW(\alpha, \alpha \mid \theta)}  \frac{p_{\alpha'}}{\cW(\alpha', \alpha' \mid \theta')} f(d)$ and $\searchthreshold \geq f(d)$.} 
    With probability $1 - 2^{-\Omega(d)}$ over $L \sim \cU(\cS^{d-1},m)$ and the internal randomness of the algorithm, ${\tt 3List}(L, \theta,\theta')$ (\Cref{algo:main-3List}) with parameters $(\alpha,\alpha', \nrep, \searchthreshold)$ outputs a list containing $m$ elements of $\cT_{\mathrm{sol}}(L, \theta,\theta')$, if $m$ elements exist, in time 
    \begin{align}
        \nrep \left( m  + \sqrt{\frac{|\cT_{\mathrm{sol}}(L, \theta,\theta')|}{\nrep}} \left(\sqrt{m^2 p_\alpha} + \sqrt{m^3 \cW(\theta, \alpha \mid \alpha) p_{\alpha'}} \right)\right) 2^{o(d)}
    \label{eq:main-3List-time-complexity}
\end{align}
    using $\max\{m, |\cT_{\mathrm{sol}}(L, \theta,\theta')|\} \, 2^{o(d)}$ classical memory and QCRAM bits, and $2^{o(d)}$ qubits. 
\end{lemma} 

\begin{proof}  
On input $(L, \theta, \theta')$ and with parameters $(\alpha, \alpha', \nrep, \searchthreshold)$, {\tt 3List} (\Cref{algo:main-3List}) samples $\nrep$ RPC pairs, independently. For each sampled RPC pair $(C,C')$, it obtains corresponding data structures $(D, D')$ from ${\tt Preprocess}(L, C, C')$ (\Cref{algo:preprocess}), and repeatedly applies ${\tt SolutionSearch}(D, D')$ (\Cref{algo:SolutionSearch}).  
In the remainder of this proof, we write $\cT_{\mathrm{sol}}$ as shorthand for $\cT_{\mathrm{sol}}(L, \theta,\theta')$ and, for each $j \in [\nrep]$, we write $(R_j, R'_j)$ for the relations corresponding to the $j$-th sampled RPC pair. 

We claim that, for all sufficiently large $\nrep = \frac{p_{\alpha}}{\cW(\alpha, \alpha \mid \theta)}  \frac{p_{\alpha'}}{\cW(\alpha', \alpha' \mid \theta')} 2^{o(d)}$, the following statements simultaneously hold, except with probability $2^{-\omega(d)}$ over $L$ and the $\nrep$ sampled RPC pairs: 
\begin{enumerate}[label=(\Roman*)]
    \item For all $j \in [\nrep]$, $(R_j, R'_j)$ is well-balanced on $L$ (\Cref{defn:well-balanced}).
    \item For all $j \in [\nrep]$, $|\cT_1(R_j)| \leq_d m^2 \cW(\theta, \alpha \mid \alpha)$.  
    \item For all $(\vx,\vy,\vz) \in \cT_{\mathrm{sol}}$, $|\{j \in [\nrep] \colon (\vx,\vy,\vz) \in \cT_{\mathrm{sol}}(R_j,R'_j)\}| \in [1, 2^{o(d)}]$. 
\end{enumerate} 
By~\Cref{lem:analysis-Search-phase}, (I) implies that for sufficiently large $\searchthreshold = 2^{o(d)}$, for all $j \in [\nrep]$, the $j$-th repetition of the outer loop of {\tt 3List} finds all elements of the corresponding set $\cT_{\mathrm{sol}}(R_j,R'_j)$. 
Since (III) implies $\bigcup_{j=1}^\nrep \cT_{\mathrm{sol}}(R_j,R'_j) = \cT_{\mathrm{sol}}$, the output of ${\tt 3List}$ contains all elements of $\cT_{\mathrm{sol}}$.  

Next, an upper bound on the time complexity of ${\tt 3List}$ is given by
$$T_{\tt 3List} = \sum_{j=1}^\nrep (T_{\tt Sample}(j) + T_{\tt Preprocess}(j) + T_{\tt Search}(j))$$ 
where $T_{\tt Sample}(j), T_{\tt Preprocess}(j), T_{\tt Search}(j)$ are upper bounds on the time complexity of the respective phases in the $j$-th repetition of the outer loop of ${\tt 3 List}$.  
It is immediate that $T_{\tt Sample}(j) = 2^{o(d)}$ for all $j \in [\nrep]$, and the \textbf{Sampling} phases together use at most $2^{o(d)}$ classical memory (as the memory can be reinitialized in each repetition).  
By~\Cref{lem:preprocessing-cost} and (I), $T_{\tt Preprocess}(j) = m 2^{o(d)}$ for all $j \in [\nrep]$, and the \textbf{Preprocessing} phases together use at most $m 2^{o(d)}$ classical memory and QCRAM bits (as the memory can be reinitialized in each repetition). 
By~\Cref{lem:analysis-Search-phase}, (I) and (II) imply that, with probability at least $1 - 2^{-\omega(d)}$, 
\begin{align*}
     T_{\tt Search}(j) = \sqrt{\max\{1, |\cT_{\mathrm{sol}}(R_j,R'_j)|\}} \left(\sqrt{m^2 p_\alpha} + \sqrt{m^3 \cW(\theta, \alpha \mid \alpha) p_{\alpha'}} \right) 2^{o(d)}
\end{align*}   
time and QCRAM queries, $2^{o(d)}$ qubits, and $|\cT_{\mathrm{sol}}(R_j,R'_j)|  2^{o(d)}$ additional classical memory. 
Here, we use that (I) implies $|\cT_0(R_j)| =_d |L|^2 / |C| =_d m^2 p_\alpha$, which is $\omega(d)$ by assumption. 
By the Cauchy-Schwarz inequality and (III), we obtain \[\sum_{j=1}^{\nrep}\sqrt{\max\{1, |\cT_{\mathrm{sol}}(R_j,R'_j)|\}}  \leq \sqrt{\nrep \sum_{j=1}^{\nrep} \max\{1, |\cT_{\mathrm{sol}}(R_j,R'_j)|\}} \leq_d \sqrt{\nrep \max\{\nrep, |\cT_{\mathrm{sol}}|\}}.\] 
By assumption, $\min\{m \cW(\theta, \alpha \mid \alpha), m^2 \cW(\theta', \alpha' \mid \alpha')\} = \omega(d)$, so $m^3 \cW(\theta, \alpha \mid \alpha) \cW(\theta', \alpha' \mid \alpha') = \omega(d)$. Thus, $\nrep \leq_d m^3p_{\theta} p_{\theta'}$ and $\min\{m^2 p_{\theta}, m^2 p_{\theta'}, m^3p_{\theta} p_{\theta'}\} = 2^{\Omega(d)}$, because $p_{\theta} = \cW(\theta, \alpha \mid \alpha) \, 2^{\Omega(d)}$ and $p_{\theta'} = \cW(\theta', \alpha' \mid \alpha') \, 2^{\Omega(d)}$. 
Therefore, with probability at least $1 - 2^{-\Omega(d)}$, $|\cT_{\mathrm{sol}}| \geq_d m^3 p_{\theta} p_{\theta'}$ by~\Cref{lem:size-of-Tsol}, which implies $\sqrt{\nrep \max\{\nrep, |\cT_{\mathrm{sol}}|\}} =_d \sqrt{ \nrep |\cT_{\mathrm{sol}}|}$. 
In particular, we obtain \begin{align*}
    T_{\tt 3List} &\leq_d  \nrep m  + \sqrt{\nrep |\cT_{\mathrm{sol}}|} \left(\sqrt{m^2 p_\alpha} + \sqrt{m^3 \cW(\theta, \alpha \mid \alpha) p_{\alpha'}} \right) 
\end{align*}
from which the lemma follows (as the total amount of classical memory used is $\max\{m, |\cT_{\mathrm{sol}}|\} \, 2^{o(d)}$). 

It remains to prove the conditions (I), (II), and (III) claimed above. By~\Cref{lem:bucket-size},  \Cref{lem:number-of-z-under-distribution-L}, and a union bound over all $j \in [\nrep]$, (I) holds except with probability $2^{-\omega(d)}$ over $L$. Here, we use that the RPCs are of size $1/ p_\alpha$ and $1 / p_{\alpha'}$ (resp.), that $m \min\{p_\alpha, p_{\alpha'}\} = \omega(d)$, and that $m p_{\theta'} \leq_d 1$. 
Moreover, by applying~\Cref{lem:ub-on-T1-R} (using $m\cW(\theta, \alpha \mid \alpha) = \omega(d)$) together with a union bound over the $\nrep = 2^{O(d)}$ sampled RPC pairs, we obtain that (II) follows from (I), except with probability $2^{-\omega(d)}$. 
Finally, to prove (III), let $(\vx,\vy,\vz) \in \cT_{\mathrm{sol}}$ and define $p_{(\vx,\vy,\vz)} \coloneqq \Pr_{(C,C')}[(\vx,\vy,\vz) \in \cT_{\mathrm{sol}}(R,R')]$, where the distribution is over $C \sim \mathrm{RPC}(d,\log d,1/p_{\alpha})$ and $C' \sim \mathrm{RPC}(d,\log d,1/p_{\alpha'})$. 
By~\Cref{lem:probability-being-covered-by-RPC},  \begin{align*}
    p_{(\vx,\vy,\vz)} &= \Pr_{C}\!\big[|R(\vx)| \leq 2^{d /\!\log d}, R(\vx) \cap R(\vy) \neq \emptyset\big] \cdot \Pr_{C'}\!\big[|R'(\tfrac{\vx-\vy}{\norm{\vx-\vy}})| \leq 2^{d /\!\log d}, R'(\tfrac{\vx-\vy}{\norm{\vx-\vy}}) \cap R'(\vz) \neq \emptyset\big] \\ 
    &=_d \frac{\cW(\alpha, \alpha \mid \theta)}{p_{\alpha}} \frac{\cW(\alpha', \alpha' \mid \theta')}{p_{\alpha'}} 
\end{align*} 
where the first equality uses that $C$ and $C'$ are independent. Consequently, for sufficiently large $\nrep = \frac{p_{\alpha}}{\cW(\alpha, \alpha \mid \theta)}  \frac{p_{\alpha'}}{\cW(\alpha', \alpha' \mid \theta')} 2^{o(d)}$, we obtain $\mu \coloneqq \E[|\{j \in [\nrep] \colon (\vx,\vy,\vz) \in \cT_{\mathrm{sol}}(R_j,R'_j)\}|] = \omega(d)$ and $\mu \leq 2^{o(d)}$, where the expectation is taken over the $\nrep$ RPC pairs that are (independently) sampled during {\tt 3List}. 
\Cref{cor:simple-application-of-Chernoff} then implies $|\{j \in [\nrep] \colon (\vx,\vy,\vz) \in \cT_{\mathrm{sol}}(R_j,R'_j)\}| \in [1, 2^{o(d)}]$, except with probability $2^{-\omega(d)}$. 
By applying a union bound over all (at most $2^{O(d)}$) $(\vx,\vy,\vz) \in \cT_{\mathrm{sol}}$, we conclude that (III) holds with probability $1 - 2^{-\omega(d)}$.  
\end{proof} 

\noindent We complete our analysis by proving~\Cref{thm:main-3List}. 

\begin{proof}[Proof of~\Cref{thm:main-3List}] 
Consider an instance $L \sim \cU(\cS^{d-1},m)$ of~\Cref{prob:finding-many-3-tuples} for sufficiently large $m = (27/16)^{d/4 + o(d)}$, and let $\theta, \theta'$ be according to~\Cref{lem:size-of-Tsol-application}. 
By~\Cref{lem:size-of-Tsol-application}, with probability  $1- 2^{-\Omega(d)}$, $\cT_{\mathrm{sol}}(L, \theta,\theta')$ consists only of 3-tuple solutions and satisfies $m \leq |\cT_{\mathrm{sol}}(L, \theta,\theta')| \leq_d m^3 p_{\theta} p_{\theta'} =_d m$. 
Therefore, it suffices to show there exists a quantum algorithm that solves~\Cref{prob:finding-many-3-tuples}  with the claimed time and memory complexity if $m \leq |\cT_{\mathrm{sol}}(L, \theta,\theta')| \leq_d m$, except with probability $2^{-\Omega(d)}$.  

We remark that $\cos(\theta')$ is so close to $1/\sqrt{3}$ (when $d \to \infty$) that we may without loss of generality assume they are equal (as it only affects $p_{\theta'}$ and $\cW(\alpha', \alpha' \mid \theta')$ by subexponential factors, for similar reasons as in the proofs of~\Cref{lem: approx cap volume} and~\Cref{lem: approx wedge volume}).  
Suppose $\alpha, \alpha' \in (0,\pi/2)$ are constants that satisfy: 
\begin{enumerate}[label=(\Roman*)]
    \item $\min\{\theta, \theta', 2\alpha - \theta, 2\alpha' - \theta'\} = \Omega(1)$;
    \item $\min \{m p_{\alpha}, m p_{\alpha'}, m \cW(\theta, \alpha \mid \alpha), m^2 \cW(\theta', \alpha' \mid \alpha')\} = \omega(d)$.  
\end{enumerate}
Then, by~\Cref{lem:main-3List-numberofreps} (using that $mp_{\theta'} \leq 2^{o(d)}$, $m^3 p_{\theta} p_{\theta'} =_d m$, and that $\theta \geq \theta' = \Omega(1)$), there is a quantum algorithm that, with probability $1 - 2^{-\Omega(d)}$ over the randomness of $L$ and the internal randomness of the algorithm, finds $m$ elements of $\cT_{\mathrm{sol}}(L, \theta,\theta')$ (if they exist) with time complexity given by~\Cref{eq:main-3List-time-complexity} for some $\nrep = \frac{p_{\alpha}}{\cW(\alpha, \alpha \mid \theta)}  \frac{p_{\alpha'}}{\cW(\alpha', \alpha' \mid \theta')} 2^{o(d)}$, using $\max\{m, |\cT_{\mathrm{sol}}(L, \theta,\theta')|\} \, 2^{o(d)}$ classical memory and QCRAM bits, and $2^{o(d)}$ qubits. 

To conclude the proof (with  $\cos(\theta) = 1/3$ and $\cos(\theta') = 1/\sqrt{3}$), we show that setting $\alpha, \alpha' \in (0,\pi/2)$ such that $\cos(\alpha) = 0.347606$ and $\cos(\alpha') = 0.427124$ implies that (I) and (II) are satisfied, and results in the desired time complexity. 
It is easy to verify that condition (I) follows from the choice of $\theta,\theta',\alpha,\alpha'$. 
Moreover, this particular choice of $(\alpha, \alpha')$ implies that condition (II) is satisfied and that~\Cref{eq:main-3List-time-complexity} is at most $2^{0.284551d + o(d)}$, as can be verified using the script available at \url{https://github.com/lynnengelberts/3List}.
\end{proof}  

\section{Application to lattice problems}\label{sec:application}

In this section, we detail how our quantum algorithm for~\Cref{prob:finding-many-3-tuples} aids in finding short vectors in a lattice, and can be turned into a quantum algorithm for solving approximate SVP with small approximation factor. 

\subsection{Lattices and the shortest vector problem}

A matrix $\mathbf{B} \in \R^{d\times d}$ with linearly independent columns generates a lattice $\Lambda$ defined as the set of all integer linear combinations of the columns of $\mathbf{B}$. Such a matrix $\mathbf{B}$ is called a \textit{basis} of $\Lambda$, and we remark that (for $d > 1$) such a basis is not unique: by combining elements we can form infinitely many bases for the same lattice. 
Every lattice has at least one shortest nonzero element (with respect to the Euclidean norm), and we denote its length by $\lambda_1>0$.
This gives rise to the Shortest Vector Problem (SVP) that we mentioned in the introduction.

\begin{problem}[Shortest Vector Problem]\label{prob:SVP}
    Given a basis of a lattice $\Lambda \subseteq \R^d$, find a lattice vector of Euclidean norm $\lambda_1$. 
\end{problem} 

\noindent SVP is known to be NP-hard under randomized reductions~\cite{van1981another,Ajtai96}. 
For cryptanalytic purposes, it would suffice to find a nonzero lattice vector of norm $\leq \gamma \lambda_1$ for some reasonably small approximation factor $\gamma \geq 1$ (for fixed $\gamma$, this variant of approximate SVP is often called $\gamma$-SVP), since the security of lattice-based cryptosystems is based on the hardness of (a decision variant of) this problem. 
In particular, many algorithms for attacking lattice-based cryptoschemes --- such as the BKZ algorithm \cite{Schnorr1987AHO,SE94} --- require a subroutine for solving approximate SVP.

\subsection{Sieving algorithms for SVP and the uniform heuristic}\label{sec:explanation-sieving}

This paper is about sieving algorithms, which are an important class of algorithms, including the fastest known classical algorithm~\cite{ADRS15} for (provably) solving SVP. First developed by Ajtai, Kumar, and Sivakumar~\cite{AKS01} and then improved by a series of subsequent works~\cite{Regev04,NV08,PS09,MV10,ADRS15,AS18}, the strategy is to begin by generating numerous lattice vectors, and then iteratively combine and ``sieve'' them to create shorter and shorter vectors. 

To prove the correctness of the runtime of sieving algorithms, existing methods add random perturbations to the lattice vectors (to make them continuous) or carefully control the distribution of the lattice points that are formed in each iteration~\cite{AKS01,MV10,HPS11SVPCVP,ADRS15,AS18}. 
However, these approaches incur extra time and are often the major contributors to the runtime of these --- provably correct --- algorithms. Nguyen and Vidick~\cite[Section 4]{NV08}  therefore suggested a heuristic assumption to simplify the analysis for a subclass of those algorithms, which we will refer to as \textit{heuristic sieving algorithms}.\footnote{The use of heuristics is common in cryptology, since cryptographers are primarily concerned with the practical solvability of problem and with the concrete runtimes of algorithms. Indeed, the practical performance of algorithms directly relates to the security of the cryptosystem, and helps determine appropriate security parameters. When a heuristic assumption holds in the sense that the observed runtime of an algorithm closely matches the asymptotic prediction obtained by assuming the heuristic, then this motivates analyzing the best possible runtime under this heuristic.}
Roughly speaking, this heuristic assumes that after each iteration of the sieving algorithm, the obtained lattice points are uniformly and independently distributed within some thin annulus. 

\begin{heuristic}[Uniform heuristic]\label{heur:iid}
    Let $L \subseteq \R^d$ be a list of lattice vectors obtained after some iteration of a heuristic sieving algorithm. After appropriate rescaling, every element of $L$ behaves as if it is an i.i.d.\ uniform sample from $\textsf{ann}_\rho=\{\vx\in\R^d: \rho \leq \norm{\vx} \leq 1\}$ for some fixed $\rho <1$.
\end{heuristic} 

\noindent This heuristic is typically considered for $\rho\rightarrow 1$ (that is, the list vectors are assumed to be essentially i.i.d.\ uniform on $\cS^{d-1}$). 
Note that this heuristic is technically incorrect, since the output list could include vectors of the form $\vx+\vy$ and $\vx+\vz$, which are correlated to each other (and hence not independent). 
However, it has been verified by several experimental studies that it does appear to capture the empirical behavior of sieving, and under the uniform heuristic, one can end up with more efficient algorithms for solving SVP (both classical and quantum ones) than are known in the provable setting. (E.g., see~\cite{NV08,MV10,BLS2016,ADHKPS19}.) 

Assuming the uniform heuristic, (heuristic) \textit{$2$-tuple sieving} algorithms have the following high-level form (as already described in the introduction). The algorithm first generates a large number of lattice vectors of norm roughly $R$ (for large $R$), and then tries to find pairs of vectors whose sum or difference yields a vector of smaller length. This process is repeated until the algorithm has found a sufficiently short vector to solve (approximate) SVP. 
If we instruct the algorithm in a way that each iteration, given lattice vectors of norm $\leq 1$ (recall that this norm bound is without loss of generality, by scaling the lattice), constructs lattice vectors of norm $\leq 1 - 1/\tau$, then there is a choice of $\tau = \poly(d)$ such that $\tau^2$ (which is polynomial in $d$) sieving iterations suffice to end up with a bunch of sufficiently short lattice vectors.\footnote{
Although the heuristic may fail when vectors become very short, it is typically assumed in this case that the (approximate) SVP instance has already been solved.}  
If each iteration runs in time $\mathsf{T}$, then the uniform heuristic (\Cref{heur:iid}) would ensure we solve SVP in time $\poly(d) \mathsf{T}$. Hence, we typically care about analyzing the time $\mathsf{T}$ and the amount of memory used in each iteration.  
Note that if we start with too few vectors, then  we may end up with not enough vectors for the next iteration. However, the uniform heuristic allows us to calculate how many vectors are needed in the initial list: for 2-tuple sieving, $(4/3)^{d/2+o(d)} \approx 2^{0.2075d}$ vectors suffice (see~\cite{NV08} for a detailed calculation).

Bai, Laarhoven, and Stehl{\'e}~\cite{BLS2016} further generalized the idea of this (2-tuple) sieving algorithm to $k$-tuple sieving: instead of just considering pairs of vectors, they introduced the idea of combining $3$-tuples (or even $k$-tuples) of vectors; that is, trying to find $3$-tuples $(\vx,\vy,\vz)$ in the current list such that $\vx\pm \vy\pm \vz$ is of reduced length. They observed that the memory complexity of $k$-tuple sieving decreases as $k$ increases (at the cost of increased time complexity). Later on, Herold and Kirshanova~\cite{HK2017} extended the uniform heuristic to the case of $k$-tuple sieving (i.e., the output of each iteration of $k$-tuple sieving is distributed uniformly over the sphere) and introduced the $k$-list problem. 
Given a list $L$ of $m$ i.i.d.\ uniform samples from $\cS^{d-1}$, this problem asks to find $m$ $k$-tuples of distinct elements $\vx_1, \dots, \vx_k \in L$ such that $\norm{\vx_1 - \vx_2 - \dots - \vx_k} \leq 1$. Indeed, for $k=3$, this is exactly~\Cref{prob:finding-many-3-tuples}. Herold and Kirshanova observed that under the uniform heuristic (\Cref{heur:iid}), solving the $k$-list problem suffices to solve SVP, at least for a small constant approximation factor.  
Moreover, the uniform heuristic allows us to calculate the minimal required list size $m$ so that $m$ $k$-tuples solutions to the $k$-list problem (and hence $k$-tuple sieving) exist, at least with high probability over the list $L$. When $k$ is constant, this minimal list size is $((k^{\frac{k}{k-1}}) / (k+1))^{\frac{d}{2}}$ up to subexponential factors in $d$~\cite{HK2017}.\footnote{This can be shown by an argument similar to the proof of~\cite[Corollary~1] {HK2017}, using the fact that~\cite[Theorem~1]{HK2017} also gives a lower bound on the probability that a $k$-tuple is a solution. 
}  
When $k=3$, the minimal list size satisfies $|L| = (27/16)^{d/4+o(d)}$, which is exactly the regime we considered for~\Cref{thm:main-3List}.

\subsection{An improved quantum algorithm for SVP using 3-tuple sieving}

Combining~\Cref{heur:iid} and~\Cref{thm:main-3List} yields a quantum algorithm that heuristically solves SVP in time $2^{0.284551d+o(d)}$ using $(27/16)^{d/4 +o(d)} = 2^{0.188722d+o(d)}$ classical bits and QCRAM bits, and subexponentially many qubits. This is the fastest known (heuristic) quantum algorithm for SVP when the total memory is limited to $2^{0.188722d+o(d)}$. 

\begin{hclaim}\label{hclaim:SVP}
    Let $\gamma = O(1)$. Under~\Cref{heur:iid}, there exists a quantum algorithm that solves $\gamma$-SVP in dimension $d$ in time $2^{0.284551d + o(d)}$ with probability $1 - 2^{-\Omega(d)}$. 
    This algorithm uses $(27/16)^{d/4 +o(d)}$ classical memory and QCRAM bits, and $2^{o(d)}$ qubits. 
\end{hclaim}

\noindent As mentioned in~\Cref{sec:result} and shown in~\Cref{table:k234complexities}, under the same heuristic and the same memory complexity (optimal for $3$-tuple sieving), our quantum algorithm improves upon all prior algorithms. Yet, compared to the fastest known heuristic quantum algorithm for SVP~\cite[Proposition~$4$]{BCSS22}, which uses 2-tuple sieving, the gain in memory complexity that we obtain by using $3$-tuple sieving still does not fully compensate for the increase in time complexity (despite our quantum speedup): time-memory product $2^{(0.1887+0.2846)d} \approx 2^{0.4733d}$ versus $2^{(0.2075+0.2563)d} \approx 2^{0.4638d}$.\footnote{Note that a time-memory product may not be the best measure for comparing these algorithms, as time and memory are not directly interchangeable.} However, our quantum algorithm uses only $2^{o(d)}$ qubits, whereas~\cite{BCSS22} uses an exponential number of qubits and QQRAM (i.e., quantum-readable quantum-writable \emph{quantum} memory, which is technologically more demanding than the QCRAM required in our approach) for implementing quantum walks.  
Considering the class of heuristic quantum algorithms for SVP that use at most $2^{o(d)}$ qubits and no QQRAM, the best known time complexity is $2^{0.2571d}$~\cite{Heiser21},  
achieved using $2^{0.2075d}$ classical memory (and also QCRAM of exponential size), which therefore still yields a better time-memory product than our result: $2^{0.4646d}$ versus $2^{0.4733d}$.

\paragraph{Acknowledgements.}
The authors thank L\'{e}o Ducas, Elena Kirshanova, Johanna Loyer, Eamonn Postlethwaite, and Yixin Shen for helpful discussions and references on sieving, and the anonymous reviewers for their valuable feedback. 
 
\bibliographystyle{alpha}
\bibliography{ref-arXiv.bib}

@InProceedings{   CL21,
  author        = {Andr{\'{e}} Chailloux and Johanna Loyer}, 
  title         = {Lattice Sieving via Quantum Random Walks},
  booktitle     = {Advances in Cryptology -- {ASIACRYPT} 2021},
  series        = {Lecture Notes in Computer Science},
  volume        = {13093},
  pages         = {63--91},
  year          = {2021},
  url           = {https://doi.org/10.1007/978-3-030-92068-5\_3}
}

@InProceedings{   KMPR2019,
  author        = {Elena Kirshanova and Erik M{\aa}rtensson and Eamonn W.
                  Postlethwaite and Subhayan Roy Moulik},
  title         = {Quantum Algorithms for the Approximate $k$-List Problem
                  and their Application to Lattice Sieving},
  booktitle     = {Advances in Cryptology -- {ASIACRYPT} 2019},
  series        = {Lecture Notes in Computer Science},
  volume        = {11921},
  pages         = {521--551},
  publisher     = {Springer},
  year          = {2019}, 
  doi           = {10.1007/978-3-030-34578-5\_19}
}

@Article{         ACKS25QSVP,
  author        = {Divesh Aggarwal and Yanlin Chen and Rajendra Kumar and
                  Yixin Shen},
  title         = {Improved Classical and Quantum Algorithms for the Shortest
                  Vector Problem via Bounded Distance Decoding},
  journal       = {{SIAM} Journal on Computing},
  volume        = {54},
  number        = {2},
  pages         = {233--278},
  year          = {2025}, 
  doi           = {10.1137/22M1486959}
}

@InProceedings{   ADRS15,
  author        = {Divesh Aggarwal and Daniel Dadush and Oded Regev and Noah
                  Stephens{-}Davidowitz},
  title         = {Solving the Shortest Vector Problem in {$2^n$} Time Using
                  Discrete {G}aussian Sampling: Extended Abstract},
  booktitle     = {Proceedings of the {ACM} Symposium on Theory of Computing
                  (STOC 2015)},
  pages         = {733--742},
  publisher     = {{ACM}},
  year          = {2015}, 
  doi           = {10.1145/2746539.2746606}
}

@InProceedings{   Ajtai96,
  author        = {Mikl{\'{o}}s Ajtai},
  title         = {Generating Hard Instances of Lattice Problems (Extended
                  Abstract)},
  booktitle     = {Proceedings of the {ACM} Symposium on Theory of Computing
                  (STOC 1996)},
  pages         = {99--108},
  year          = {1996}, 
  doi           = {10.1145/237814.237838}
}

@InProceedings{   AKS01,
  author        = {Ajtai, Mikl\'{o}s and Kumar, Ravi and Sivakumar, D.},
  title         = {A sieve algorithm for the shortest lattice vector
                  problem},
  booktitle     = {Proceedings of the {ACM} Symposium on Theory of Computing
                  (STOC 2001)},
  pages         = {601--610},
  publisher     = {{ACM}},
  year          = {2001},
  url           = {https://doi.org/10.1145/380752.380857},
  doi           = {10.1145/380752.380857}
}

@Article{         Amb10VariableTime,
  author        = {Andris Ambainis},
  title         = {Quantum Search with Variable Times},
  journal       = {Theory of Computing Systems},
  volume        = {47},
  number        = {3},
  pages         = {786--807},
  year          = {2010},
  doi           = {10.1007/S00224-009-9219-1}
}

@InProceedings{   AS18,
  author        = {Divesh Aggarwal and Noah Stephens{-}Davidowitz},
  title         = {Just Take the Average! An Embarrassingly Simple
                  {$2^n$}-Time Algorithm for {SVP} (and {CVP)}},
  booktitle     = {Symposium on Simplicity in Algorithms ({SOSA} 2018)},
  series        = {Open Access Series in Informatics},
  volume        = {61},
  pages         = {12:1--12:19},
  year          = {2018}, 
  doi           = {10.4230/OASICS.SOSA.2018.12}
}

@InProceedings{   BCSS22,
  author        = {Xavier Bonnetain and Andr{\'e} Chailloux and Andr{\'e}
                  Schrottenloher and Yixin Shen},
  title         = {Finding many Collisions via Reusable Quantum Walks},
  booktitle     = {Advances in Cryptology -- {EUROCRYPT}
                  2023},
  volume        = {14008},
  pages         = {221--251},
  year          = {2023},
  doi           = {10.1007/978-3-031-30589-4}
}

@InCollection{    BHMT02,
  author        = {Brassard, Gilles and H{\o}yer, Peter and Mosca, Michele
                  and Tapp, Alain},
  title         = {Quantum amplitude amplification and estimation},
  url           = {http://dx.doi.org/10.1090/conm/305/05215},
  doi           = {10.1090/conm/305/05215},
  booktitle     = {Quantum Computation and Information},
  publisher     = {American Mathematical Society},
  year          = {2002},
  volume        = {305},
  series        = {Contemporary Mathematics},
  pages         = {53--74}
}

@Article{         BLS2016,
  author        = {Shi Bai and Thijs Laarhoven and Damien Stehl{\'{e}}},
  title         = {Tuple lattice sieving},
  journal       = {{LMS} Journal of Computational Mathematics},
  volume        = {19},
  number        = {A},
  pages         = {146--162},
  year          = {2016},
  doi           = {10.1112/S1461157016000292}
}

@Article{         chernoff1952Bound,
  author        = {Herman Chernoff},
  title         = {A Measure of Asymptotic Efficiency for Tests of a
                  Hypothesis Based on the sum of Observations},
  year          = {1952},
  journal       = {The Annals of Mathematical Statistics},
  volume        = {23},
  number        = {4},
  pages         = {493 -- 507},
  doi           = {10.1214/aoms/1177729330}
}

@InProceedings{   CL23,
  title         = {Classical and Quantum 3 and 4-Sieves to Solve {SVP} with
                  Low Memory},
  author        = {Chailloux, Andr{\'e} and Loyer, Johanna},
  booktitle     = {International Conference on Post-Quantum Cryptography},
  pages         = {225--255},
  year          = {2023},
  organization  = {Springer},
  url           = {https://eprint.iacr.org/2023/200.pdf}
}

@InProceedings{   ADHKPS19,
  title         = {The general sieve kernel and new records in lattice
                  reduction},
  author        = {Albrecht, Martin R. and Ducas, L{\'e}o and Herold,
                  Gottfried and Kirshanova, Elena and Postlethwaite, Eamonn
                  W. and Stevens, Marc},
  
  booktitle     = {Advances in Cryptology -- {EUROCRYPT} 2019},
  series        = {Lecture Notes in Computer Science},
  volume        = {11477},
  pages         = {717--746},
  year          = {2019},
  organization  = {Springer}
}

@InProceedings{   Amb04ED,
  author        = {Andris Ambainis},
  booktitle     = {Proceedings of the IEEE Symposium on Foundations of
                  Computer Science (FOCS 2004)},
  pages         = {22--31},
  title         = {Quantum Walk Algorithm for Element Distinctness},
  year          = {2004}
}

@InProceedings{   gilyen2018QSingValTransf,
  title         = {Quantum Singular Value Transformation and Beyond:
                  {E}xponential Improvements for Quantum Matrix Arithmetics},
  author        = {Andr{\'a}s Gily{\'e}n and Yuan Su and Guang Hao Low and
                  Nathan Wiebe},
  booktitle     = {Proceedings of the {ACM} Symposium on Theory of Computing
                  (STOC 2019)},
  year          = {2019},
  pages         = {193--204},
  numpages      = {12},
  doi           = {10.1145/3313276.3316366}
}

@InProceedings{   Gro96,
  author        = {Lov K. Grover}, 
  title         = {A Fast Quantum Mechanical Algorithm for Database Search},
  booktitle     = {Proceedings of the {ACM} Symposium on Theory of Computing
                  (STOC 1996)},
  pages         = {212--219},
  publisher     = {{ACM}},
  year          = {1996}, 
  doi           = {10.1145/237814.237866}
}

@Misc{            Heiser21,
  author        = {Max Heiser},
  title         = {Improved Quantum Hypercone Locality Sensitive Filtering in
                  Lattice Sieving},
  howpublished  = {Cryptology ePrint Archive, Paper 2021/1295},
  pages         = {1295},
  year          = {2021},
  url           = {https://eprint.iacr.org/2021/1295}
}

@InProceedings{   HK2017,
  author        = {Gottfried Herold and Elena Kirshanova},
  title         = {Improved Algorithms for the Approximate k-List Problem in
                  {E}uclidean Norm},
  booktitle     = {Public-Key Cryptography -- {PKC} 2017},
  series        = {Lecture Notes in Computer Science},
  volume        = {10174},
  pages         = {16--40},
  publisher     = {Springer},
  year          = {2017}, 
  doi           = {10.1007/978-3-662-54365-8\_2}
}

@InProceedings{   HKL2017,
  author        = {Gottfried Herold and Elena Kirshanova and Thijs
                  Laarhoven},
  title         = {Speed-ups and time-memory trade-offs for tuple lattice
                  sieving},
  booktitle     = {Public-Key Cryptography -- {PKC} 2018},
  series        = {Lecture Notes in Computer Science},
  volume        = {10769},
  pages         = {407--436},
  publisher     = {Springer},
  year          = {2018}, 
  doi           = {10.1007/978-3-319-76578-5\_14}
}

@InProceedings{   HPS11SVPCVP,
  author        = {Guillaume Hanrot and Xavier Pujol and Damien Stehl{\'{e}}},
  title         = {Algorithms for the Shortest and Closest Lattice Vector
                  Problems},
  booktitle     = {Coding and Cryptology ({IWCC} 2011)},
  series        = {Lecture Notes in Computer Science},
  volume        = {6639},
  pages         = {159--190},
  year          = {2011},
  doi           = {10.1007/978-3-642-20901-7\_10}
}

@Misc{            Iiridayn2020,
  author        = {Iiridayn},
  title         = {{Answer to ``When to stop enumerating a fixed set of
                  unknown cardinality via random sampling?''}},
  year          = {2020},
  howpublished  = {Cross Validated (Stack Exchange)},
  note          = {\url{https://stats.stackexchange.com/a/486813}, accessed
                  October 2025}
}

@InProceedings{   KLL15twoOracles,
  author        = {Shelby Kimmel and Cedric Yen{-}Yu Lin and Han{-}Hsuan
                  Lin},
  title         = {Oracles with Costs},
  booktitle     = {Theory of Quantum Computation, Communication and Cryptography ({TQC} 2015)},
  series        = {LIPIcs},
  volume        = {44},
  pages         = {1--26},
  year          = {2015},
  doi           = {10.4230/LIPICS.TQC.2015.1}
}

@PhDThesis{       Laa16,
  author        = {Thijs Laarhoven},
  title         = {Search Problems in Cryptography: From Fingerprinting to
                  Lattice Sieving},
  year          = {2015},
  school        = {Eindhoven University of Technology},
  url           = {https://thijs.com/docs/phd-final.pdf}
}

@InProceedings{   BHT98,
  author        = {Brassard, Gilles and H{\o}yer, Peter and Tapp, Alain},
  booktitle     = {Latin American Symposium on Theoretical Informatics (LATIN
                  1998)},
  series        = {Lecture Notes in Computer Science},
  pages         = {163--169},
  title         = {Quantum Algorithm for the Collision Problem},
  year          = {1998}
}

@Article{         LMP15,
  author        = {Laarhoven, Thijs and Mosca, Michele and {Pol}, {Joop van
                  de}},
  year          = {2015},
  month         = {12},
  pages         = {},
  title         = {Finding Shortest Lattice Vectors Faster Using Quantum
                  Search},
  volume        = {77},
  journal       = {Designs, Codes and Cryptography},
  doi           = {10.1007/s10623-015-0067-5}
}

@Article{         SE94,
  author        = {Claus{-}Peter Schnorr and M. Euchner},
  title         = {Lattice basis reduction: Improved practical algorithms and
                  solving subset sum problems},
  journal       = {Mathematical Programming},
  volume        = {66},
  pages         = {181--199},
  year          = {1994}, 
  doi           = {10.1007/BF01581144}
}

@Misc{            MR08,
  author        = {Daniele Micciancio and Oded Regev},
  title         = {Lattice-Based Cryptography },
  year          = {2008},
  note          = {\url{https://cims.nyu.edu/~regev/papers/pqc.pdf}}
}

@Book{            MU05probability,
  author        = {Mitzenmacher, Michael and Upfal, Eli},
  title         = {Probability and Computing: Randomized Algorithms and
                  Probabilistic Analysis},
  year          = {2005},
  publisher     = {Cambridge University Press}
}

@Article{         NV08,
  author        = {Phong Q. Nguyen and Thomas Vidick},
  title         = {Sieve Algorithms for the Shortest Vector Problem Are
                  Practical},
  journal       = {Journal of Mathematical Cryptology},
  volume        = {2},
  number        = {2},
  pages         = {181--207},
  year          = {2008}, 
  doi           = {10.1515/JMC.2008.009}
}

@Misc{            PS09,
  author        = {Xavier Pujol and Damien Stehl{\'{e}}},
  title         = {Solving the Shortest Lattice Vector Problem in Time
                  2\({}^{{2.465n}}\)},
  howpublished  = {Cryptology {ePrint} Archive, Paper 2009/605},
  year          = {2009}
}

@Misc{            Regev04,
  author        = {Oded Regev},
  title         = {Lattices in Computer Science, Lecture 8},
  year          = {2004},
  note          = {\url{https://cims.nyu.edu/~regev/teaching/lattices_fall_2004/ln/svpalg.pdf}}
}

@Article{         Schnorr1987AHO,
  title         = {A Hierarchy of Polynomial Time Lattice Basis Reduction
                  Algorithms},
  author        = {Claus-Peter Schnorr},
  journal       = {Theoretical Computer Science},
  year          = {1987},
  volume        = {53},
  pages         = {201-224},
  doi           = {10.1016/0304-3975(87)90064-8}
}

@Article{         Shor97,
  author        = {Peter Shor},
  title         = {Polynomial-Time Algorithms for Prime Factorization and
                  Discrete Logarithms on a Quantum Computer},
  journal       = {{SIAM} Journal on Computing},
  volume        = {26},
  number        = {5},
  pages         = {1484-1509},
  year          = {1997}, 
  doi           = {10.1137/S0097539795293172},
  note          = {Earlier version in FOCS 1994}
}

@InProceedings{   BDGL16,
  author        = {Anja Becker and L{\'{e}}o Ducas and Nicolas Gama and Thijs
                  Laarhoven},
  title         = {New Directions In Nearest Neighbor Searching with
                  Applications to Lattice Sieving},
  booktitle     = {Proceedings of the {ACM-SIAM} Symposium on
                  Discrete Algorithms ({SODA} 2016)},
  pages         = {10--24},
  year          = {2016},
  doi           = {10.1137/1.9781611974331.ch2}
}

@InProceedings{   MV10,
  author        = {Daniele Micciancio and Panagiotis Voulgaris}, 
  title         = {Faster Exponential Time Algorithms for the Shortest Vector
                  Problem},
  booktitle     = {Proceedings of the {ACM-SIAM}
                  Symposium on Discrete Algorithms ({SODA} 2010)},
  pages         = {1468--1480},
  publisher     = {{SIAM}},
  year          = {2010}, 
  doi           = {10.1137/1.9781611973075.119}
}

@TechReport{      techreport:LK2014,
  author        = {Lee, Yongjae and Kim, Woo Chang},
  institution   = {KAIST},
  type          = {Technical Report},
  year          = {2014}, 
  pages         = {},
  title         = {Concise Formulas for the Surface Area of the Intersection
                  of Two Hyperspherical Caps}
}

@Book{            van1981another,
  title         = {Another NP-complete partition problem and the complexity
                  of computing short vectors in a lattice},
  author        = {{Emde Boas}, {Peter van}},
  series        = {Report. Department of Mathematics. University of
                  Amsterdam},
  year          = {1981}
}

@Article{         yoder2014FixedPointSearch,
  title         = {Fixed-Point Quantum Search with an Optimal Number of
                  Queries},
  author        = {Yoder, Theodore J. and Low, Guang Hao and Chuang, Isaac
                  L.},
  journal       = {Physical Review Letters},
  volume        = {113},
  number        = {21},
  pages         = {210501},
  numpages      = {5},
  year          = {2014},
  doi           = {10.1103/PhysRevLett.113.210501}
}

@inproceedings{AR20,
author = {Scott Aaronson and Patrick Rall},
title = {Quantum Approximate Counting, Simplified},
booktitle = {Symposium on Simplicity in Algorithms (SOSA 2020)},
year={2020},
doi = {10.1137/1.9781611976014.5},
}

\appendix

\section{Appendix} 

\subsection{Proofs of~\Cref{lem: approx cap volume} and~\Cref{lem: approx wedge volume}}
\label{app:proofs-approx-lemmas} 

\begin{proof}[Proof of~\Cref{lem: approx cap volume}]
Let $\vx \in \cS^{d-1}$, and define $p\coloneqq \Pr_{\vc \sim \cU(\cS^{d-1})}[\innerP{\vx}{\vc} \approx_\epsapprox  \cos(\alpha)]$. 
By~\Cref{lem: cap volume}, 
\begin{align*}
    p &\leq \Pr_{\vc \sim \cU(\cS^{d-1})}[\innerP{\vx}{\vc} \geq \cos(\alpha) - \epsapprox] \\
    &=_d  (1 - (\cos(\alpha) - \epsapprox)^2)^{d/2}  \\
    &= (1 - \cos^2(\alpha))^{d/2} \left(1 + \frac{2\epsapprox\cos(\alpha) - \epsapprox^2}{1 - \cos^2(\alpha)}\right)^{d/2} \\
    &\leq (1 - \cos^2(\alpha))^{d/2} \left(1 + \frac{2\epsapprox}{1 - \cos^2(\alpha)}\right)^{d/2} \\ 
    &\leq (1 - \cos^2(\alpha))^{d/2} \exp(O(\epsapprox d))
\end{align*}
since $\alpha = \Omega(1)$ implies $1 - \cos^2(\alpha) = \Omega(1)$, and using the fact that $1 + x \leq \exp(x)$ for all $x \in \R$. As $\epsapprox = 1/(\log d)^2$, we obtain $p \leq_d p_\alpha$. 

Moreover, by the union bound and~\Cref{lem: cap volume}, we obtain that for some constant~$k > 0$,  
\begin{align*}
    p &\geq \Pr_{\vc \sim \cU(\cS^{d-1})}[\innerP{\vx}{\vc} \geq \cos(\alpha)]  - \Pr_{\vc \sim \cU(\cS^{d-1})}[\innerP{\vx}{\vc} \geq \cos(\alpha) + \epsapprox] \\
    &\geq \frac{1}{d^k} (1 - \cos^2(\alpha))^{d/2}  - d^k(1 - (\cos(\alpha) + \epsapprox)^2)^{d/2} \\
    &= (1 - \cos^2(\alpha))^{d/2} \left(\frac{1}{d^k} -  d^k \left(1 - \frac{2\epsapprox\cos(\alpha) + \epsapprox^2}{1 - \cos^2(\alpha)}\right)^{d/2} \right).
\end{align*} 
Since $(1 - \frac{2\epsapprox\cos(\alpha) + \epsapprox^2}{1 - \cos^2(\alpha)})^{d/2} \leq (1 - \epsapprox^2)^{d/2} \leq \exp(-\epsapprox^2 d/2)$, the choice $\epsapprox = 1/(\log d)^2$ yields $p \geq_d p_\alpha$.
\end{proof}

\begin{proof}[Proof of~\Cref{lem: approx wedge volume}] 
Let $\vx, \vy \in \cS^{d-1}$ be such that $\innerP{\vx}{\vy} \approx_\epsapprox \cos(\theta)$, and define
\begin{equation*}
    p \coloneqq \Pr_{\vc \sim \cU(\cS^{d-1})}[\innerP{\vx}{\vc} \approx_\epsapprox \cos(\alpha), \innerP{\vy}{\vc} \approx_\epsapprox \cos(\beta)].
\end{equation*} 
By~\Cref{lem: wedge volume}, it suffices to show that $p =_d (1 - \gamma^2)^{d/2}$, where 
\[
    \gamma^2 \coloneqq \frac{\cos^2(\alpha) + \cos^2(\beta) - 2\cos(\alpha)\cos(\beta)\cos(\theta)}{\sin^2(\theta)}.
\]  
We proceed by showing that $(1 - \widetilde{\gamma}^2)^{d/2} \leq_d p \leq_d (1 - \overline{\gamma}^2)^{d/2}$ for some $\widetilde{\gamma}^2 \leq \gamma^2 + O(\epsapprox)$ and  $\overline{\gamma}^2 \geq \gamma^2 - O(\epsapprox)$ that are defined in terms of $(\alpha, \beta, \theta, \epsapprox)$, allowing us to prove $p =_d (1 - \gamma^2)^{d/2}$. 

More precisely, define $\cos(\phi) \coloneqq \innerP{\vx}{\vy}$. Then $\cos(\theta) - \epsapprox \leq \cos(\phi) \leq \cos(\theta) + \epsapprox$, which implies $- 3\epsapprox \leq \sin^2(\phi) - \sin^2(\theta) \leq 2 \epsapprox$. 
In addition, the assumption $|\alpha - \beta| + \Omega(1) \leq \theta \leq \alpha + \beta - \Omega(1)$ implies $\theta = \Omega(1)$. Using  $\varepsilon = o(1)$, it follows that there exists a constant $\kappa > 0$ such that (for sufficiently large $d$) $\sin^2(\phi) \geq \kappa$. 
Using standard trigonometric identities, we also obtain 
\begin{align*}
    \sin^2(\theta)(1 - \gamma^2) &= \sin^2(\theta) - \cos^2(\alpha) - \cos^2(\beta) + 2\cos(\alpha)\cos(\beta)\cos(\theta) \\
    &= 4\sin\left(\frac{(\beta + \theta) + \alpha}{2}\right)\sin\left(\frac{(\beta + \theta) - \alpha}{2}\right)\sin\left(\frac{\alpha + (\beta - \theta)}{2}\right)\sin\left(\frac{\alpha - (\beta - \theta)}{2}\right)
\end{align*} 
and thus the assumption on $\alpha,\beta,\theta$ implies $1 - \gamma^2 \geq \sin^2(\theta)(1 - \gamma^2) = \Omega(1)$. In other words, there also exists a constant $\kappa' > 0$ such that $1 - \gamma^2 \geq \kappa'$ (for sufficiently large $d$). 

\noindent Now, for the upper bound, note that~\Cref{lem: wedge volume} implies 
\begin{align*}
    p \leq \Pr_{\vc \sim \cU(\cS^{d-1})}[\innerP{\vx}{\vc} \geq \cos(\alpha) - \epsapprox, \innerP{\vy}{\vc} \geq \cos(\beta) - \epsapprox] 
    =_d (1 - \overline{\gamma}^2)^{d/2}
\end{align*}
where \begin{align*}
    \overline{\gamma}^2 &\coloneqq \frac{(\cos(\alpha) - \epsapprox)^2 + (\cos(\beta) - \epsapprox)^2 - 2(\cos(\alpha)- \epsapprox)(\cos(\beta)- \epsapprox)\cos(\phi)}{\sin^2(\phi)} \\ 
    &\geq \frac{\cos^2(\alpha) + \cos^2(\beta) - 2\cos(\alpha)\cos(\beta)\cos(\theta) - 2\epsapprox(\cos(\alpha) + \cos(\beta) + \cos(\alpha)\cos(\beta))}{\sin^2(\phi)} \\
    &\geq \frac{\cos^2(\alpha) + \cos^2(\beta) - 2\cos(\alpha)\cos(\beta)\cos(\theta) - 6\epsapprox}{\sin^2(\phi)} \\
    &= \gamma^2 \frac{\sin^2(\theta)}{\sin^2(\phi)} - \frac{6\epsapprox}{\sin^2(\phi)} \\
    &= \gamma^2 \left(1  -  \frac{\sin^2(\phi) - \sin^2(\theta)}{\sin^2(\phi)}\right) - \frac{6\epsapprox}{\sin^2(\phi)}.
\end{align*}
Since $\gamma^2 \leq 1$, $\sin^2(\phi) - \sin^2(\theta) \leq 2 \epsapprox$, and $\sin^2(\phi) \geq \kappa$, we obtain $\overline{\gamma}^2 \geq \gamma^2 - \frac{8 \epsapprox}{\kappa}$. The desired upper bound now follows, because (using $1 - \gamma^2 \geq \kappa'$ and $\epsapprox = 1/(\log d)^2$)
\begin{align*}
    (1 - \overline{\gamma}^2)^{d/2} = (1 - \gamma^2)^{d/2} \left( 1 + \frac{\gamma^2 - \overline{\gamma}^2}{1 - \gamma^2}\right)^{d/2}  \leq (1 - \gamma^2)^{d/2} \left(1 + \frac{8\epsapprox}{\kappa \kappa'}\right)^{d/2} \leq_d (1 - \gamma^2)^{d/2}. 
\end{align*}
For the lower bound, note that the union bound implies
\begin{align*}
    p \geq &\Pr_{\vc \sim \cU(\cS^{d-1})}[\innerP{\vx}{\vc} \geq \cos(\alpha) - \epsapprox, \innerP{\vy}{\vc} \geq \cos(\beta) - \epsapprox] \\ 
    &- \Pr_{\vc \sim \cU(\cS^{d-1})}[\innerP{\vx}{\vc} > \cos(\alpha) + \epsapprox, \innerP{\vy}{\vc} \geq \cos(\beta) - \epsapprox] \\ 
    &- \Pr_{\vc \sim \cU(\cS^{d-1})}[\innerP{\vx}{\vc} \geq \cos(\alpha) - \epsapprox, \innerP{\vy}{\vc} > \cos(\beta) + \epsapprox] \\
    \geq&_d \Pr_{\vc \sim \cU(\cS^{d-1})}[\innerP{\vx}{\vc} \geq \cos(\alpha) - \epsapprox, \innerP{\vy}{\vc} \geq \cos(\beta) - \epsapprox] \\
    \geq &\Pr_{\vc \sim \cU(\cS^{d-1})}[\innerP{\vx}{\vc} \geq \cos(\alpha), \innerP{\vy}{\vc} \geq \cos(\beta)]
\end{align*}
where the second inequality can be shown using similar methods as how we have shown the upper bound, using that $\epsapprox = 1/(\log d)^2$. 
By~\Cref{lem: wedge volume}, we thus obtain 
\begin{align*}
    p \geq_d ( 1 - \widetilde{\gamma}^2)^{d/2} = (1 - \gamma^2)^{d/2} \left(1 - \frac{\widetilde{\gamma}^2 - \gamma^2}{1 - \gamma^2}\right)^{d/2}
\end{align*}
where \begin{align*}
    \widetilde{\gamma}^2 &\coloneqq \frac{\cos^2(\alpha) + \cos^2(\beta) - 2\cos(\alpha)\cos(\beta)\cos(\phi)}{\sin^2(\phi)} \\
    &= \frac{\cos^2(\alpha) + \cos^2(\beta) - 2\cos(\alpha)\cos(\beta)(\cos(\theta) + \cos(\phi) - \cos(\theta))}{\sin^2(\theta)} \frac{\sin^2(\theta)}{\sin^2(\phi)} \\
    &= \left(\gamma^2 - \frac{2\cos(\alpha)\cos(\beta)(\cos(\phi) - \cos(\theta))}{\sin^2(\theta)}\right) \frac{\sin^2(\theta)}{\sin^2(\phi)}.
\end{align*} 
If $\sin^2(\theta) \geq \sin^2(\phi)$ (meaning $\theta \geq \phi$ and $\cos(\theta) \leq \cos(\phi)$), then $\widetilde{\gamma}^2 \leq \gamma^2 \frac{\sin^2(\theta)}{\sin^2(\phi)}$. 
Otherwise, 
\begin{align*}
    \widetilde{\gamma}^2 &= \left(\gamma^2 + \frac{2\cos(\alpha)\cos(\beta)(\cos(\theta) - \cos(\phi))}{\sin^2(\theta)}\right) \frac{\sin^2(\theta)}{\sin^2(\phi)} 
    \leq  
    \gamma^2 \frac{\sin^2(\theta)}{\sin^2(\phi)} + \frac{2\epsapprox}{\sin^2(\phi)}.
\end{align*}
In either case, we have $\sin^2(\theta) - \sin^2(\phi) \leq 3 \epsapprox$, so $\gamma^2 \leq 1$ and $\sin^2(\phi) \geq \kappa$ imply 
\begin{align*}
    \gamma^2 \frac{\sin^2(\theta)}{\sin^2(\phi)}
     = \gamma^2 + \gamma^2\frac{\sin^2(\theta) - \sin^2(\phi)}{\sin^2(\phi)} \leq \gamma^2 + \frac{3\epsapprox}{\kappa}.
\end{align*}
In other words, we have shown $\widetilde{\gamma}^2 - \gamma^2 \leq \frac{5\epsapprox}{\kappa}$. 
Hence, \begin{align*}
    p &\geq_d (1 - \gamma^2)^{d/2} \left(1 - \frac{\widetilde{\gamma}^2 - \gamma^2}{1 - \gamma^2}\right)^{d/2} 
    \geq (1 - \gamma^2)^{d/2} \left(1 - \frac{5\epsapprox}{\kappa \kappa'}\right)^{d/2} 
    \geq_d (1 - \gamma^2)^{d/2}
\end{align*}
as desired. 
\end{proof}

\end{document}